\documentclass[prd,tightenlines,nofootinbib,superscriptaddress]{revtex4}

\usepackage{amsfonts,amssymb,amsthm,bbm}

\usepackage{amsmath}

\usepackage{tensor}

\allowdisplaybreaks[3]

\usepackage{hyperref}

\usepackage{subcaption}
\usepackage{color,psfrag}
\usepackage[dvips]{graphicx}

\usepackage{tikz}
\usepackage{tikz-cd}
\usetikzlibrary{calc,matrix,patterns,angles,quotes}
\usetikzlibrary{arrows,decorations.pathmorphing,decorations.pathreplacing,decorations.markings}
\usetikzlibrary{hobby,knots,celtic,shapes.geometric,calc}

\tikzset{
    edge/.style={draw, postaction={decorate},
        decoration={markings,mark=at position .55 with {\arrow{>}}}},
    linking-D/.style={draw, postaction={decorate},
        decoration={markings,mark=at position 1 with { rectangle, draw, inner sep=1pt, minimum size=2mm, fill=cyan }}},
}

\usepackage{mathtools}

\newcounter{footnotemarknum}

\newcommand{\C}{{\mathbb C}}
\newcommand{\N}{{\mathbb N}}

\newcommand{\cF}{{\mathcal F}}

\newcommand{\cH}{{\mathcal H}}

\newcommand{\cO}{{\mathcal O}}

\newcommand{\cV}{{\mathcal V}}

\newcommand{\cS}{{\mathcal S}}

\newcommand{\cW}{{\mathcal W}}

\newcommand{\SU}{\mathrm{SU}}

\newcommand{\be}{\begin{equation}}
\newcommand{\ee}{\end{equation}}
\newcommand{\beq}{\begin{eqnarray}}
\newcommand{\eeq}{\end{eqnarray}}

\renewcommand{\u}{{\mathfrak{u}}}

\newcommand{\la}{\langle}
\newcommand{\ra}{\rangle}

\newcommand{\tr}{{\mathrm{Tr}}}
\newcommand{\f}{\frac}
\newcommand{\tl}{\widetilde}
\def\nn{\nonumber}
\def\pp{\partial}

\def\rd{\mathrm{d}}
\def\hH{\wh{H}}

\def\vphi{\varphi}
\def\veps{\varepsilon}
\def\eps{\epsilon}

\newcommand{\id}{\mathbb{I}}

\def\act{\triangleright}

\def\tk{\tilde{k}}

\def\ri{ { \rm{i} } }

\def\wh{\widehat}

\def\bGamma{{\Gamma}^{\circ}}

\newtheorem{theorem}{Theorem}[section]

\newtheorem{prop}[theorem]{Proposition}
\newtheorem{res}[theorem]{Result}
\newtheorem{coro}[theorem]{Corollary}
\newtheorem{defi}[theorem]{Definition}

\begin{document}

\title{Intertwiner Entanglement Excitation and Holonomy Operator}

\author{{\bf Qian Chen}}\email{chenqian.phys@gmail.com}
\affiliation{ENSL, CNRS, Laboratoire de Physique, F-69342 Lyon, France}
\author{{\bf Etera R. Livine}}\email{etera.livine@ens-lyon.fr}
\affiliation{ENSL, CNRS, Laboratoire de Physique, F-69342 Lyon, France}

\date{\today}

\begin{abstract}
In the loop quantum gravity  framework, spin network states carries entanglement between quantum excitations of the geometry at different space points. This intertwiner entanglement is gauge-invariant and comes from quantum superposition of spins and intertwiners.
Bipartite entanglement can be interpreted as a witness of distance, while multipartite entanglement reflects the curvature of the quantum geometry.
The present work investigates how the bipartite and multipartite intertwiner entanglement changes under the action of the holonomy operator, which is the basic building block of loop quantum gravity's dynamics.
We reveal the relation between entanglement excitation and the dispersion of the holonomy operator.
This leads to a new interesting connection between bulk geometry and boundary observables via
% and find an interesting access to connect the bulk and boundary via
the dynamics of entanglement.
\end{abstract}

\maketitle
%%%%%%%%%%%%%%%%%%%%%%%%%%%%%%%%%%%%%%%%%%%%%%%%%%%%%%%%
\tableofcontents

%%%%%%%%
\section{Introduction}
%%%%%%%%

Loop Quantum Gravity (LQG) proposes a background independent framework for a theory of quantum general relativity (see \cite{Gaul:1999ys,Thiemann:2007zz,Rovelli:2014ssa,Bodendorfer:2016uat} for reviews), provides a local definition of quantum states of space geometry as spin networks and a canonical description for their dynamics through a Hamiltonian constraint.
The geometry of 3d space slices is described by a pair of canonical fields, the (co-)triad and
the Ashtekar-Barbero connection. Standard loop quantum gravity approach performs a canonical quantization of the holonomy-flux algebra, of observables smearing the Ashtekar-Barbero connection along 1d curves - holonomies, and the (co-)triad along 2d surfaces - fluxes, and defines quantum states of geometry as polymer structures or graph-like geometries with edges and vertices. Those spin network states represent the excitations of the Ashtekar-Barbero connection as Wilson lines in gauge field theory. On the other hand, geometric observables are raised to quantum operators acting on the Hilbert space spanned by spin networks, leading to the celebrated result of  discrete spectra for areas (living along edges) and volumes (living at vertices) \cite{Rovelli:1994ge,Ashtekar:1996eg,Ashtekar:1997fb}.

So spin networks are the kinematical states of the theory and the game is to describe their dynamics, i.e. their evolution in time generated by the Hamiltonian constraints\footnotemark. 
\footnotetext{
Although the traditional canonical point of view is to attempt to discretize, regularize and quantize the Hamiltonian constraints \cite{Thiemann:1996aw,Thiemann:1996av}, this often leads to anomalies. The formalism naturally evolved towards a path integral formulation. The resulting spinfoam models, constructed from (extended) topological quantum field theories (TQFTs) with defects, define transition amplitudes for histories of spin networks \cite{Reisenberger:1996pu,Baez:1997zt,Barrett:1997gw,Freidel:1998pt} (see \cite{Livine:2010zx,Dupuis:2010kqh,Perez:2012wv} for reviews). The formalism then evolves in a third quantization, where so-called “group field theories” define non-perturbative sums over random spin network histories in a similar way than matrix model partition functions define sums over random 2d discrete surfaces \cite{DePietri:1999bx,Reisenberger:2000zc,Freidel:2005qe} (see \cite{Oriti:2006se,Carrozza:2013oiy,Oriti:2014yla} for reviews).
}
Evolving spin networks, formalized as spinfoams, describe the four-dimensional quantum space-time at the Planck scale. This quantum space-time is defined without reference to a background classical geometry, with quantum states of geometry are defined up to diffeomorphisms with no reference to any intrinsic coordinate system or background structure. Then concepts in classical geometry, such as distance, area, curvature, become emergent notions, in a continuum limit after suitably coarse-graining Planck scale quantum fluctuations. They can only be reconstructed from the interaction between subsystems, quantified by correlation and entanglement shared between subsystems 
(see e.g. \cite{Donnelly:2016auv,Feller:2017jqx})
%\cite{Livine:2006xk,Livine:2007sy,Donnelly:2008vx,Livine:2008iq,Donnelly:2011hn,Donnelly:2014gva,Bianchi:2015fra,Feller:2015yta,Bianchi:2016tmw,Feller:2016zuk,Chirco:2017xjb,Livine:2017fgq,Anza:2016fix,Girelli:2005ii}. 
%
For instance, bipartite correlations, equivalent to 2-point functions, allow to reconstruct a notion of distance while curvature should be reflected by multi-body correlations.
This perspective sets the field of quantum information at the heart of research in quantum gravity, with essential roles to play for entanglement, decoherence and quantum localization in probing quantum states of geometries and thinking about the quantum-to-classical transition for the space-time geometry.

\smallskip

In the context of loop quantum gravity, the work investigating the entanglement carried by spin networks states have slowly built since the birth of the theory \cite{Livine:2005mw,Livine:2007sy,Donnelly:2008vx,Livine:2008iq,Donnelly:2011hn,Donnelly:2014gva,Feller:2015yta,Bianchi:2015fra,Feller:2016zuk,Bianchi:2016hmk,Delcamp:2016eya,Livine:2017fgq,Baytas:2018wjd}, but has definitely sped up in the past few years with the burst of interest in the bulk-to-boundary propagator and bulk-from-boundary reconstruction in the light of holography, see for instance \cite{Anza:2016fix,Chirco:2017xjb,Colafranceschi:2021acz,Chirco:2021chk,Colafranceschi:2022ual}.
In this paper, we focus on intertwiner entanglement\footnotemark{}, i.e. the entanglement between quanta of volume carried by a spin network state,
\footnotetext{
As shown in \cite{Livine:2017fgq}, boundary state entanglement typically involves the gauge symmetry breaking and is defined with respect to gauge-variant physical measurements and the corresponding choice of boundary reference frame. On the other hand, intertwiner entanglement is gauge-invariant, and can be argued to quantify genuine spin network entanglement.
}
and study the entanglement created by the holonomy operator, which is recognized both as a discretized measure of curvature in loop quantum gravity  and as the basic building block of the Hamiltonian operator. We wish to understand how the excitation of curvature is related to the excitation of entanglement on spin network states. This is part of the larger line of research towards answering the questions: what is the relation between geometry and entanglement? Can one understand the dynamics of the quantum geometry directly in terms of evolution of entanglement and quantum information?

Furthermore, since the spin network is viewed as a many-body quantum system, it is necessary to introduce the notion of multipartite entanglement \cite{PhysRevA.68.042307,Amico:2007ag,GUHNE20091}.
It is still an ongoing topic in the field of quantum information. There are many candidates for quantifying the entanglement of multipartite system. These entanglement measures generally are not equivalent. In order to explore the entanglement structure for LQG, it is required to find a suitable entanglement measure.
Our work shows that the notion of geometric measure of entanglement may be a good one, and justify the geometric measure since it admits a straightforward generalization from bipartite to multipartite spin network.

\smallskip

Section \ref{Section:HolonomyOperatorSNHS} sets up the mathematical definitions of the Hilbert space of spin network states, the defines and describes the action of the holonomy operator on those states. A gauge-invariant holonomy operator defined on a loop acts on all the intertwiner states living at the nodes of the spin network.
Section \ref{Section:MultipartiteEntanglement} views spin network states as many-body quantum systems, living in the tensor product of the space of intertwiners attached to each node of the network. Spin network basis states, with fixed spins on the links and fixed intertwiners on the nodes, is a separable state and carries no entanglement.
We explore the notion of multipartite entanglement \cite{PhysRevA.68.042307,Amico:2007ag,GUHNE20091} and define the notion of geometric  entanglement for spin networks. Our main result is the computation of  the growth of entanglement due to the action of the holonomy on a spin network basis state.
Finally, sections \ref{section:CandyGraph} and \ref{section:TriangleGraph} apply this general result  to the simplest network structures: the 2-vertex graph with boundary -or candy graph- consisting in two nodes and a triangular graph consisting in three nodes. These simple settings allow for explicit computations of the evolution (generated by the holonomy operator) of both the bipartite entanglement and the geometric multipartite entanglement, showing that they match at leading order.

%is to study entanglement excitation on candy graph, also to show the geometric measure of entanglement indeed complies the bipartite entanglement entropy.
%The Section \ref{section:TriangleGraph} studies the entanglement excitation on triangle graph, displays how the geometric measure of entanglement generalizes bipartite entanglement in our issue, and compare the bipartite entanglement of tripartite system with the genuine tripartite entanglement measure.

%%%%%%

%%%%%%%%
\section{Holonomy operators on spin network Hilbert space} \label{Section:HolonomyOperatorSNHS}
%%%%%%%%
%%%
\subsection{Spin network Hilbert space}
\label{subsection:SP-SNWF}
%%%

%For globally hyperbolic four-dimensional space-times $\cM=\Sigma\times\R$ with closed three-dimensional spatial slices $\Sigma$, 

Loop quantum gravity proceeds to a canonical quantization of general relativity (see \cite{Thiemann:2007zz} for detailed lectures, or \cite{Ashtekar:2021kfp} for a recent overview), describing the evolution of a 3d (space-like) slice in time, thereby generating the four-dimensional space-time. It defines quantum states of geometry and describes their constrained evolution in time.
A state of geometry is defined as a wave-function $\psi$ of the Ashtekar-Barbero connection on the canonical hypersurface. Loop quantum gravity choose cylindrical wave-functions, that depend on the holonomies of that connection along the edges of a graph $\Gamma$. These holonomies are $\SU(2)$ group elements,  $g_{e}\in\SU(2)$ for each edge $e$ of the graph. And the wave-functions are required to be gauge-invariance under local $\SU(2)$ transformations, which act at every vertices of the graph.

For closed 3d spatial slices, without boundary, we consider closed graphs, i.e. without open links.
%A state of 3d geometry are defined by  a closed oriented graph $\Gamma$ and a wave-function $\psi$ on it.
A wave-function $\psi$ on a closed oriented graph $\Gamma$ is a function of one $\SU(2)$ group element $g_{e}$ for each edge $e$, and is assumed to be invariant under the $\SU(2)$-action at each vertex $v$ of the graph:
\be
\label{eq:GaugeTransformation}
\begin{array}{rlcl}
\psi_{\Gamma}: 
&\SU(2)^{\times E}
&\longrightarrow
&\C \\
&\{g_{e}\}_{e\in\Gamma}
&\longmapsto
&\psi(\{g_{e}\}_{e\in\Gamma})
%=
%\la\{g_{e}\}_{e\in\Gamma} | \psi_{\Gamma} \ra
=
\psi_{\Gamma}
(\{h_{t(e)}g_{e}h_{s(e)}^{-1}\}_{e\in\Gamma})\,, \quad
\forall h_{v}\in\SU(2)\,
\end{array}
%
%\psi_{\Gamma}
%(\{g_{e}\}_{e\in\Gamma})
%=
%\la\{g_{e}\}_{e\in\Gamma} | \psi_{\Gamma} \ra
%=
%\psi_{\Gamma}
%(\{h_{t(e)}g_{e}h_{s(e)}^{-1}\}_{e\in\Gamma})\,, \quad
%\forall h_{v}\in\SU(2)\,
\ee
where $t(e)$ and $s(e)$ respectively refer to the target and source vertices of the edge $e$. We write $E$ and $V$ respectively for the number of edges and vertices of the considered graph $\Gamma$.
The scalar product between such wave-functions is given by the Haar measure on the Lie group $\SU(2)$:
\be \label{eq:ScalarProduct-SpinNetwork-Integration}
\la \psi_{\Gamma}
| \tl{\psi}_{{\Gamma} } \ra
=\int_{\SU(2)^{{\times E}}}\prod_{e}\rd g_{e}\,
\overline{\psi_{\Gamma}
(\{g_{e}\}_{e\in\Gamma}) }\,
\tl{\psi}_{ {\Gamma} }
(\{g_{e}\}_{e\in{\Gamma} })
\,.
\ee
The Hilbert space of quantum states with support on the closed graph $\Gamma$ is thus realized as a space of square-integrable functions, $\cH_{\Gamma}=L^{2}(\SU(2)^{{\times E}}/\SU(2)^{{\times V}})$.

\smallskip

For a 3d slice with boundary, we consider graphs with open links puncturing the slice boundary. Those open links are connected to one vertex of the graph, while their other extremity is left loose. Without loss of generality, we can assume that all those open links are outward oriented, i.e. that their source vertex belongs to the graph, while their target vertex does not. Those open links are referred to as the boundary links.
We can use the same definition as above for a closed graph, considering wave-functions of one group element per link, including both the standard links in the interior and the boundary links. The difference is that gauge transformations will only act at the graph vertices and will not act on the open ends.

Another way to proceed, as advocated in \cite{Chen:2021vrc}, is to explicitly partition the graph set of edges into interior links and boundary links,
\be
\Gamma= \Gamma^{o}\sqcup \pp\Gamma\,.
\ee
One introduces the concept of boundary state as functions of the holonomies along the boundary links. Then
one considers the bulk wave-functions as mapping holonomies along the interior links to boundary states:
\be
\label{eq:bulkwavefunction}
\begin{array}{rlcl}
\psi_{\Gamma}: 
&\SU(2)^{\{e\in \Gamma^{o}\}}
&\longrightarrow
&\cF\big{[}\SU(2)^{\{e\in\pp\Gamma\}}\big{]} \\
&\{g_{e}\}_{e\in\Gamma^{o}}
&\longmapsto
&\psi[\{g_{e}\}_{e\in\Gamma^{o}}]\,,
\end{array}
\ee
with the bulk gauge invariance turning into a covariance property of the boundary states:
\be
\psi[\{g_{e}\}_{e\in\Gamma^{o}}](\{g_{e}\}_{e\in\pp\Gamma})
=
\psi[\{h_{t(e)}g_{e}h_{s(e)}^{-1}\}_{e\in\Gamma^{o}}](\{g_{e}h_{s(e)}^{-1}\}_{e\in\pp\Gamma})
\,, \quad
\forall h_{v}\in\SU(2)\,.
\ee
This follows the logic of state-sum models and discrete topological path integrals: boundary cells are dressed with quantum states, while bulk cells are dressed with morphisms of their boundary Hilbert space.
%
%If one dresses a $d$-dimensional cell without boundary with quantum states, then, when one considers $d$-dimensional cells with boundary, the $(d-1)$-dimensional boundary is dressed with quantum states
%
When decomposing into spins, wave-functions on a closed graph can be decomposed on spin network states, dressed with one irreducible spin representation on each link $j_{e}$ and with one intertwiner state (or singlet state)  on each vertex $I_{v}$. Wave-functions on an open graph can be similarly decomposed into bulk spin network states. In that case, as explained in  \cite{Chen:2021vrc}, the boundary Hilbert space is the tensor product of spin states living on the open links:
\be
\cH_{\pp\Gamma}=\bigotimes_{e \in \pp\Gamma}\cH_{e}
\,,\qquad
\cH_{e}=\bigoplus_{j_{e}\in\f\N2}\cV_{j_{e}}\,.
\ee
A bulk state can be understood as a boundary map, i.e. a function from bulk group elements to the boundary Hilbert space:
\be
\psi[\{g_{e}\}_{e\in\Gamma^{o}}]\in \cH_{\pp\Gamma}
\,.
\ee
We refer the interested reader to \cite{Chen:2021vrc}, and to the recent work \cite{Anza:2017dkd,Colafranceschi:2021acz,Colafranceschi:2022ual}, for more details on wave-functions for bounded regions and interesting work on typical bulk states.

\smallskip

Let us now briefly review spin network states and notations.
A basis of this Hilbert space can be constructed using the spin decomposition of $L^{2}$ functions on the Lie group $\SU(2)$ according to the Peter-Weyl theorem. A {\it spin} $j\in\f\N2$ defines an irreducible unitary representation of $\SU(2)$, with the action of $\SU(2)$ group elements realized on a $(2j+1)$-dimensional Hilbert space $\cV_{j}$. We use the standard orthonormal basis $|j,m\ra$, labeled by the spin $j$ and the magnetic index $m$ running by integer steps from $-j$ to $+j$, which diagonalizes the $\SU(2)$ Casimir $\vec{J}^{2}$ and the $\u(1)$ generator $J_{z}$. Group elements $g$ are then represented by the $(2j+1)\times (2j+1)$ Wigner matrices $D^{j}(g)$:
\be
D^{j}_{mm'}(g)=\la j,m|g|j,m'\ra\,,\qquad
\overline{D^{j}_{mm'}(g)}
=
D^{j}_{m'm}(g^{-1})
\,.
\ee
These Wigner matrices form an orthogonal basis of $L^{2}(\SU(2))$:
\be
\int_{\SU(2)}\rd g\,\overline{D^{j}_{ab}(g)}\,{D^{k}_{cd}(g)}
=
\int_{\SU(2)}\rd g\,\overline{D^{j}_{ba}(g^{-1})}\,{D^{k}_{cd}(g)}
=
\f{\delta_{jk}}{2j+1}\delta_{ac}\delta_{bd}
\,,\qquad
\delta(g)
=\sum_{j\in\f\N2}(2j+1)\chi_{j}(g)
\,, \label{eq:Peter-Weyl}
\ee
where $\chi_{j}$ is the spin-$j$ character defined as the trace of the Wigner matrix, $\chi_{j}(g)=\tr D^{j}(g)=\sum_{m=-j}^{j}\la j,m|g|j,m\ra$.
Applying this to gauge-invariant wave-functions allows to build the {\it spin network} basis states of $\cH_{\Gamma}$, which depend on one spin $j_{e}$ on each edge and one intertwiner $I_{v}$ at each vertex:
\be
\Psi_{\Gamma,\{j_{e},I_{v}\}}(\{g_{e}\}_{e\in\Gamma})
=
\la\{g_{e}\}) | \{j_{e},I_{v}\}\ra
=
\sum_{m_{e}^{t,s}}
\prod_{e}\sqrt{2j_{e}+1}\,\la j_{e}m_{e}^{t}|g_{e}|j_{e}m_{e}^{s}\ra
%\prod_{e}\la j_{e}m_{e}^{t}|g_{e}|j_{e}m_{e}^{s}\ra
\,\prod_{v} \la \bigotimes_{e|\,v=s(e)} j_{e}m_{e}^{s}|\,I_{v}\,|\bigotimes_{e|\,v=t(e)}j_{e}m_{e}^{t}\ra
\,. \label{eq:SpinNetwork}
\ee
As illustrated on fig.\ref{fig:intertwiner}, an {\it intertwiner} is a $\SU(2)$-invariant  state -or singlet- living in the tensor product of the incoming and outgoing spins at the vertex $v$:
\be \label{eq:IntertwinerHilbertSpace}
I_{v}\in\textrm{Inv}_{\SU(2)}\Big{[}
\bigotimes_{e|\,v=s(e)} \cV_{j_{e}}
\otimes
\bigotimes_{e|\,v=t(e)} \cV_{j_{e}}^{*}
\Big{]}
\,.
\ee
The intertwiners are singlet states, recoupling $\SU(2)$ irreducible representations. The simplest cases are bivalent and trivalent intertwiner which are uniquely determined when spins are given, since there is an unique way to recouple two or three spins into a singlet state - null angular momentum. For higher valent vertex, to label a recoupling amongst spins, one needs to choose a channel to recouple them. This amounts to split a higher valent vertex into trivalent vertices and link them via a tree graph (no forming loop), e.g., fig.\ref{fig:Unfolding-5Valent}.
Once such channel, or tree, is chosen, an intertwiner basis state is defined by the assignment of a spin to each internal link, or intermediate spins. Two different intertwiner basis states that have different intermediate spins (i.e. at least one different intermediate spins) are mutually orthogonal. An generic intertwtiner state will then be a arbitrary superposition of those basis states.
There are various possible ways to choose a channel to label intertwiner basis states. The unitary map between basis associated to different channels is given by the $\{3nj\}$ spin recoupling symbols.

\smallskip

We will alleviate the notation for intertwiner.
An intertwiner $\vert \{ j_e\}_{e\ni v} \,, I_{v}^{ ( \{ j_e\}_{ e\ni v } ) } \ra$ is labeled by $\{ j_e\}_{e\ni v}$, the spins attached, and $I_{v}^{ ( \{ j_e\}_{ e\ni v } ) }$ the internal indices when attached spins $\{ j_e\}_{e\ni v}$ are given. For instance, a trivalent intertwiner only needs attached spins $\{ j_1,j_2,j_3 \}$ for labeling, while a four-valent intertwiner needs $\{ j_1,j_2,j_3,j_4\}$ for attached spins plus an internal index $j_{12}$ for recoupled spin of $\cV_{j_1} \otimes \cV_{j_2}$, likewise, for five-valent vertex unfolding in fig.\ref{fig:Unfolding-5Valent}, attached spins $\{ j_1,\cdots, j_5 \}$ and internal indices $\{ j_{12}, j_{45} \}$ are needed for labeling.
For each vertex $v$, the attached spins are implicitly expressed in internal indices $I_{v}^{ ( \{ j_e\}_{ e\ni v } ) }$, we don't need to explicitly specify the attached spins, and adopt $I_{v}$ simply instead of $I_{v}^{ ( \{ j_e\}_{ e\ni v } ) }$, unless in some necessary cases. Hence from now on, $\vert \{ j_e\}_{e\ni v} \,, I_{v}^{ ( \{ j_e\}_{ e\ni v } ) } \ra \equiv | I_{v} \ra$.
Under the alleviation, the scalar product between two spin network basis states on the same graph $\Gamma$ is then given by the product of the scalar products between their intertwiners:
\be \label{eq:ScarlarProduct-SpinNetwork-BasisStates}
\la \Psi_{ \Gamma_{ \{j_{e},I_{v}\}} } | \Psi_{ \Gamma_{ \{\tilde{j}_{e},\tilde{I}_{v}\}} } \ra
%=\la \Gamma, {\{j_{e},I_{v}\}}| \Gamma, {\{\tilde{j}_{e},\tilde{I}_{v}\}}\ra
=\prod_{e}
\delta_{j_{e},\tilde{j}_{e}} \,
\prod_{v}
\la I_{v}|\tilde{I}_{v}\ra
\equiv
\prod_{v}
\la I_{v}|\tilde{I}_{v}\ra
\,.
\ee

\begin{figure}[htb]
	\centering
\begin{subfigure}[t]{0.45\linewidth}
	\begin{tikzpicture} [scale=1.2]
\coordinate  (O) at (0,0);

\node[scale=0.7] at (O) {$\bullet$} node[below] {$I_v$};

\draw[thick] (O)  --  node[midway,sloped]{$>$} ++ (1,1) node[right] {$| j_1, m_1 \ra$};

\draw[thick] (O)  to[bend left=20]  node[midway,sloped]{$<$} ++ (0.9,-0.9) node[right] {$| j_2, m_2 ]$};

\draw[thick] (O)  to[bend left=20] node[midway,sloped]{$>$} ++ (0,1.5) node[above] {$| j_5, m_5 \ra$};

\draw[thick] (O)  to[bend left=10]  node[midway,sloped]{$<$} ++ (-1.2,-0.5) node[left] {$| j_3, m_3 \ra$};

\draw[thick] (O)  to[bend left=10]  node[midway,sloped]{$>$} ++ (-1.1,0.6) node[left] {$| j_4, m_4 ]$};

\end{tikzpicture}
\caption{An five-valent intertwiner $I_v$ at vertex $v$.}
\label{fig:intertwiner}
\end{subfigure}
\begin{subfigure}[t]{0.45\linewidth}
\begin{tikzpicture}[scale=1]

\coordinate  (O) at (0,0);
\coordinate  (A) at (-1,0);
\coordinate  (B) at (-1,1);

\node[scale=0.7] at (O) {$\bullet$};
\node[scale=0.7] at (A) {$\bullet$};
\node[scale=0.7] at (B) {$\bullet$};

\draw[thick] (O) --node[midway,sloped]{$<$} node[midway,below] {$j_{12}$} (A) --node[midway,sloped]{$>$} node[midway,left] {$j_{45}$} (B);

\draw[thick] (O)  --  node[midway,sloped]{$>$} ++ (1,1) node[right] {$| j_1, m_1 \ra$};

\draw[thick] (O)  --  node[midway,sloped]{$<$} ++ (0.9,-0.9) node[right] {$| j_2, m_2 ]$};

\draw[thick] (B)  -- node[midway,sloped]{$>$} ++ (0.3,1.2) node[above] {$| j_5, m_5 \ra$};

\draw[thick] (A)  -- node[midway,sloped]{$<$} ++ (-1.2,-0.5) node[left] {$| j_3, m_3 \ra$};

\draw[thick]  (B)  -- node[midway,sloped]{$>$} ++ (-1.1,0.6) node[left] {$| j_4, m_4 ]$};

\end{tikzpicture}

\caption{Unfolding the five-valent intertwiner with trivalent virtual vertices.
}
\label{fig:Unfolding-5Valent}
\end{subfigure}
\caption{The notation for a higher valent intertwiner $I_{v}$ in terms of virtual spins.
}
\label{fig:Indices-5Valent}
\end{figure}
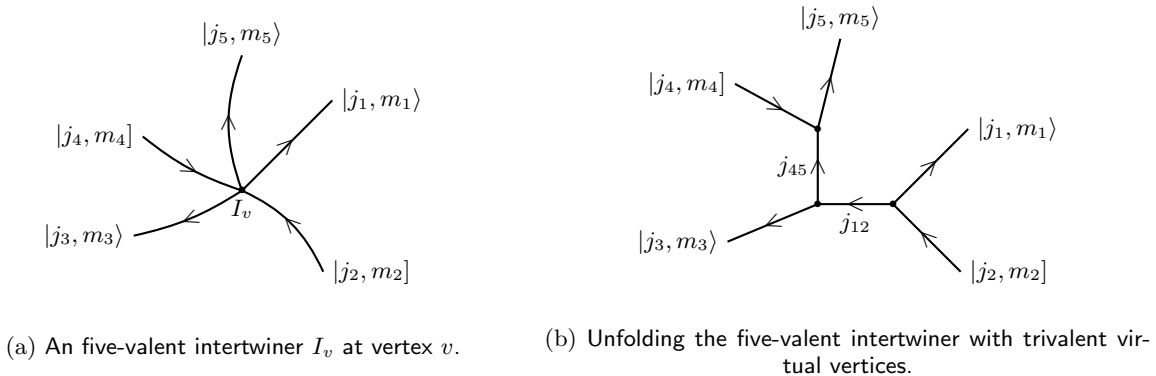
%

%We have shown the definition of spin networks, also the intertwiner basis states for high valent vertex. The next part is to define the action of loop holonomy operator on spin networks.

The spin networks form a basis of the Hilbert space of kinematical quantum states of geometry. Dynamics is then built as operators acting on spin networks. Most operators, including the various version of the Hamiltonian constraint operators, are constructed from the basic holonomy operator acting on closed loops on the graph. The next section is dedicated to analyzing the action of the holonomy operator on spin network basis states, in order to later investigate the entanglement created by holonomy operators.

%%%%%%

%%%
\subsection{Loop holonomy operator on spin networks}
%%%

The dynamics of spin network states implement the flow generated by the Hamiltonian constraints on the embedded geometry of the canonical hypersurface. At the quantum level, the Hamiltonian constraint operators involve the holonomy operator and geometric observables, such as areas and volumes. The holonomy operator  is analogous to the Wilson-loop in QCD. It corresponds to the quantization of the curvature in the polymer quantization scheme used in loop quantum gravity, where one does not access to point-like excitations, but only to gauge-invariant observables smeared along 1d structures.
Let us analyze the action of the holonomy operator on spin network basis states, along the lines of \cite{Bonzom:2009zd,Borja:2010gn}.

%To match the picture of general relativity that spacetimes itself is dynamical, it is inevitable to involve the dynamics of spin network states from which geometric notions would emerge. For this goal, we study the holonomy operator in LQG \cite{Bonzom:2009zd,Borja:2010gn}, which is interpreted as the quantization of curvature. In addition, it is analogous to the Wilson-loop in QCD. The holonomy operator's representation underlying spin network Hilbert space will be presented, which also admits a closed relation between corner angles when comes to asymptotic regime.

\smallskip

Let us look at the holonomy operator with spin $\ell$ acting on a single edge $e$. This operator takes the tensor product of the spin-$\ell$ with the spin $k_{e}$ carried by the edge, and its action can be expressed in terms of Clebsch-Gordan coefficients decomposing this tensor product $\ell \otimes k_{e}$ into irreducible representations.
%
%To understand the operation of holonomy operator, we need the ClebschâGordan coefficients for tensoring irreducible representations.
%Let a holonomy operator act on edge $e$ carrying holonomy $g_{e}$. Holonomy operator defines tensor product of spin-$k_{e}$ representation $D^{k_{e} } : G \to \mathrm{GL}(2k_{e}+1,\C)$ and spin-$\ell$ representation. Via $\SU(2)$ recoupling, the tensor product $D^{\ell}\otimes D^{k_{e} }$ of Wigner D-matrices admits direct sum decomposition with respect to ClebschâGordan coefficients:
%
Indeed the Wigner matrices for the $\SU(2)$  group element $g_{e}$ carrying the holonomy along the edge $e$ satisfies the following algebraic relations:
\beq \label{eq:HolonomyOperator-TensorProduct}
\widehat{ D^{\ell}_{a_{e} b_{e} } }
\act
\Big[ D^{k_{e} }_{m_{e} n_{e}}(g_{e}) \Big]
&=&
D^{\ell}_{a_{e} b_{e} } (g_{e} ) \,
D^{k_{e} }_{m_{e} n_{e}}(g_{e} )
\,, \qquad
D^{\ell}
\otimes
D^{k_{e} }
=
\bigoplus_{K_{e}
=
| k_{e}-\ell | }^{ k_{e}+\ell }
D^{K_{e}}
\,,
\\
D^{\ell }_{a_{e} b_{e} }(g_{e} ) \,
D^{k_{e} }_{m_{e} n_{e} }(g_{e} )
&=&
\sum_{K_{e}=| k_{e}-\ell | }^{k_{e} + \ell } \,
\sum_{ M_{e}, N_{e}=-K_{e} }^{K_{e} } \,
(-1)^{M_{e}-N_{e} }
(2K_{e}+1)
\begin{pmatrix}
   \ell & k_{e} & K_{e} \\
   a_{e} & m_{e} & -M_{e}
  \end{pmatrix}
\overline{ \begin{pmatrix}
   \ell & k_{e} & K_{e} \\
   b_{e} & n_{e} & -N_{e}
  \end{pmatrix} } \,
D^{K_{e} }_{M_{e} N_{e} }(g_{e})
  \,, \nn \\
   && \label{eq:HolonomyOperator-3j}
\eeq
where the recoupled spin $K_{e}$ is bounded by the triangular inequalities $| k_{e}-\ell | \leq K_{e} \leq k_{e} +\ell$.

%\smallskip
%
%The  holonomy operator along a single edge spoils the gauge invariance. To restore the gauge invariance, define the holonomy operator for at least a couple of edges $(ij)$ by multiplication \cite{Borja:2010gn}:
%\beq
%\left( \wh{ \chi_{ \ell } } \, \act_{(ij)} \psi \right) ( g_1\, \cdots \, g_{E} )
%=
%\chi_{ \ell } ( g_{i} g_{j}^{-1} ) \psi ( g_1\, \cdots \, g_{E} )
%\,.
%\eeq
%Here $\chi_{ \ell }(g)=\tr D^{\ell}(g)$ is the trace of representation of spin-$\ell$.
%The operator is clearly Hermitian and symmetric under the exhange $i \leftrightarrow j$.

\medskip

The  holonomy operator along a single edge spoils the gauge invariance. In order to produce a gauge-invariant holonomy operator, one must consider a closed loop on the graph $\Gamma$ underlying the spin network state.
Consider a loop $W \subseteq \Gamma$ with $n$  edges, and assume the simplifying condition that it does not go through a vertex more than that once. The oriented loop $W$ can be described as the path $W[v_1 \overset{ e_{1} }{\to} \cdots \overset{ e_{n-1} }{\to} v_n \overset{ e_{n} }{\to} v_1]$ such that the edge $e_{i}$ links the vertex $v_{i}$ to $v_{i+1}$, with $i=1,\cdots, n$ and the implicit convention $n+1\equiv 1$. The loop holonomy operator is defined as a multiplicative operator on the wave-functions:
\beq
\left( \wh{ \chi_{ \ell } } \, \act_{W} \psi_{\Gamma} \right) ( \{g_e\}_{e\in\Gamma}  )
=
\chi_{\ell} (G_{W}) \psi_{\Gamma} ( \{g_e\}_{e\in\Gamma}  )
\,, \qquad \text{with} \quad 
G_{W}=\overleftarrow{\prod_{ e_{i} \in W } g_{ e_{i} }   }
\,,
\eeq
where  $\chi_{ \ell }(g)=\tr D^{\ell}(g)$ is the character of the spin-$\ell$ representation.
We take the inverse of a group element if the edge is oriented in the opposite direction than the loop.
Since the factor $\chi_{\ell} (G_{W})$ is a gauge invariant function, the resulting wave-function is still gauge-invariant. Thus the map $\wh{\chi_{\ell} } \, \act_{W}$
%: \cH_{\Gamma} \to \cH_{\Gamma}$
acts legitimately on the Hilbert space $\cH_{\Gamma}$ and we can write its action on the spin network basis:
%
% and we can decompose it on the spin network basis using the Peter-Weyl theorem.
%so any gauge invariant operator defined along loop $\cL$ could be decomposed in terms of Peter-Weyl theorem
%\beq
%\wh{F} \, \act_{W}
%=
%\sum_{\ell=0(\f12) }^{\infty} %\sum_{ {\ell} \in \f{\N}{2}}
%f_{\ell} \, \wh{\chi}_{ {\ell} } \, \act_{W} 
%\,.
%\eeq
%
%\smallskip
%
%In order to represent the holonomy operator in spin network Hilbert space, note that the operator defines a map $\wh{\chi_{\ell} } \, \act_{W} : \cH_{\Gamma} \to \cH_{\Gamma}$, defining matrix $\tensor{Z(\Gamma) }{ _{\tensor{\chi}{_{\ell} }\, \act_{W}    }  } \in \textrm{End} \big[\cH_{\Gamma}\big]$. Choosing spin network basis state $\{ | \Psi_{ \Gamma, \{I_{v}\} } \ra \}$, we write the action of $\wh{\chi_{\ell} } \, \act_{W}$ in terms of:
\beq \label{eq:LoopHolonomyOperator-SpinNetwork}
\wh{\chi_{\ell} } \, \act_{W} \,
| \Psi_{ \Gamma, \{I_{v}\} } \ra
=
\sum_{ \{I'_v\} } \,
\tensor{   {\Big[ \tensor{Z(\Gamma) }{ _{\tensor{\chi}{_{\ell} }\, \act_{W}    }   } \Big] }   } {^{ \{ I'_{v} \} } _{ \{ I_{v} \} } }
\,
| \Psi_{ \Gamma, \{I'_{v}\} } \ra
\,,
\eeq
where the matrix elements  $\tensor{Z(\Gamma) }{ _{\tensor{\chi}{_{\ell} }\, \act_{W}    }   }$ are given by the following integrals:
\beq
\tensor{   {\Big[ \tensor{Z(\Gamma) }{ _{\tensor{\chi}{_{\ell} }\, \act_{W}    }    } \Big] }   } {^{ \{ I'_{v} \} } _{ \{ I_{v} \} } }
=
\int
\prod_{e\in\Gamma}\rd g_{e} \,
\overline{ 
\Psi_{\Gamma, \{I'_{v}\} }
(\{ g_{e} \}_{e\in\Gamma}) }
\, \chi_{\ell} ( G_{W} ) \,
\Psi_{\Gamma, \{I_{v}\} } ( \{ g_e \}_{e\in\Gamma}  )
%\,, \quad \text{where} \quad
%G_{W}=\overleftarrow{\prod_{ e_{i} \in W } g_{ e_{i} }  }
\,.
\label{eq:Amplitudes}
\eeq
%The intertwiner indices $\{ I'_{v} \}$ and $\{ I_{v} \}$ serve as row and column indices for matrix $Z$.
This matrix $\tensor{Z(\Gamma) }{ _{\tensor{\chi}{_{\ell} }\, \act_{W}    }   }$ satisfies a composition rule:
\beq \label{eq:Composition-Amplitudes}
\sum_{ \{ I'_{v} \} }
\,
\tensor{   {\Big[ \tensor{Z(\Gamma) }{   _{\tensor{\chi}{_{\ell_{1} } }\, \act_{W}    }   } \Big] }   } {^{ \{ I''_{v} \} } _{ \{ I'_{v} \} } }
\,
\tensor{   {\Big[ \tensor{Z(\Gamma) }{ _{\tensor{\chi}{_{ \ell_{2} } }\, \act_{W}    }  } \Big] }   } {^{ \{ I'_{v} \} } _{ \{ I_{v} \} } }
&=&
\sum_{ s=| {\ell}_{1} - {\ell}_{2} | }^{ {\ell}_{1} + {\ell}_{2} }
\,
\tensor{   {\Big[ \tensor{Z(\Gamma) }{    _{\tensor{\chi}{_s}\, \act_{W}    }    } \Big] }   } {^{ \{ I''_{v} \} } _{ \{ I_{v} \} } }
=
\tensor{   {\Big[ \tensor{Z(\Gamma) }{    _{ ( \tensor{\chi}{_{ {\ell}_{1} } } \cdot \tensor{\chi}{_{ {\ell}_{2} } } ) \, \act_{W}    }    } \Big] }   } {^{ \{ I''_{v} \} } _{ \{ I_{v} \} } }
\,,
\eeq
which is inherited from the character recoupling formula $\chi_{ {\ell}_{1} } \chi_{ {\ell}_{2} } = \sum_{ s=| {\ell}_1- {\ell}_2 | }^{ {\ell}_1+{\ell}_2 } \, \chi_{s}$.
%The composite rule not only carries out matrix multiplication but also summation over multiplicity.
%
An interesting fact to keep in mind is that the matrices $\tensor{Z(\Gamma) }{   _{\tensor{\chi}{_{\ell_1} }\, \act_{W}    }   }$ and $\tensor{Z(\Gamma) }{ _{\tensor{\chi}{_{ {\ell}_2 } }\, \act_{W}    }  }$ commute with each other with arbitrary spins ${\ell}_1$ and ${\ell}_2$.

\medskip

The transition matrix $Z$ can be expressed in terms of the $\{6j\}$ symbols of spin recoupling. To this purpose, we introduce  bouquet spins for the vertices along the loop $W$ following the previous work \cite{Chen:2021vrc}. The bouquet spin is a recoupled spin used to define a convenient intertwiner basis. As illustrated on fig.\ref{fig:bouquet}, we distinguish at each vertex $v$ the two edges belonging to the loop $\{ e \in W \}$ and we recouple all the spins carried by the other edges $\{e\notin W\}$ into a spin $J_{v}$.
If a vertex around the loop is 3-valent, then the bouquet spin is obvious simply the spin carried by the third edge, not belonging to the loop. In general, the definition of the bouquet spin allows to consider  all the intertwiners around the loop as 3-valent for all matters of operators acting on the loop edges.
%
%Then all the interwiners around the loop can be considered as 3-valent.
%
%To exhibit the explicit form of transition matrix $Z$, we firstly review the bouquet spins \cite{Chen:2021vrc}, in order to include situations that each vertex's valency can be arbitrarily high (at least trivalent). A bouquet spin is just a recoupled spin that is convenient for using in problem. In our case, the loop $W$ separates edges into two classes, $\{ e \in W \}$ and $\{e\notin W\}$. For each vertex $v$ laying along $W$, all the edges $e \in W$ attached to $v$, recouple into $J_{v}$, meanwhile all the edges $e \notin W$ attached to $v$, recouple into $J'_{v}$. Since all spins attached to $v$ recouple into a singlet state, it is necessary to have $J_{v}=J'_{v}$.
%
%The recoupling procedure changes nothing to the spin network states. It is only a convenient way to represent the action of loop holonomy such that three quantities are required per vertex along $W$, i.e., two spins along $W$ plus one bouquet spin.
%Via bouquet spins, we are able to present the action of loop holonomy operator on the corresponding equivalent trivalent graph. The fig \ref{fig:bouquet} is an example of the usage of bouquet spins.
%Then it is enough the express the form of transition matrix $Z$.
%
This convenient unfolding of the intertwiners allows to write the action of the holonomy operator in terms of the spins along the loop and the bouquet spins:
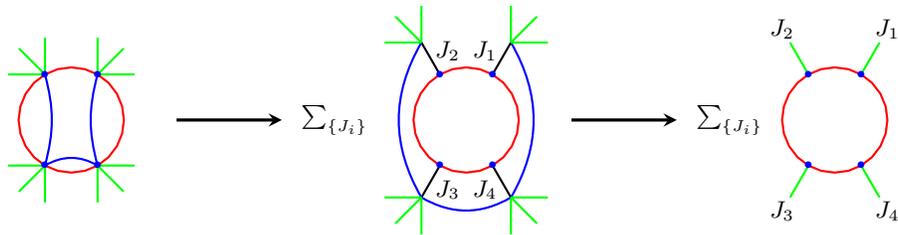
\begin{figure}[hbt!]
\centering
\begin{tikzpicture}[scale=0.7]
\coordinate (O) at (1.5,0);
  % a circle
%\draw (O) circle (1);
\path (O) ++(60:1) coordinate (O1);
\path (O) ++(120:1) coordinate (O2);
\path (O) ++(240:1) coordinate (O3);
\path (O) ++(300:1) coordinate (O4);

   %\draw [green,thick,domain=120:240] plot ({cos(\x)}, {sin(\x)});
   \draw [domain=0:360,thick,red] plot ({1.5+cos(\x)}, {sin(\x)});

\draw[thick] (O1) -- ++ (60:0.7) node [left,near end] {$J_1$} coordinate (O5);
\draw[thick] (O2) -- ++ (120:0.7) node [right,near end] {$J_2$} coordinate (O6);
\draw[thick] (O3) -- ++ (240:0.7) node [right,near end] {$J_3$} coordinate (O7);
\draw[thick] (O4) -- ++ (300:0.7) node [left,near end] {$J_4$} coordinate (O8);

\draw[thick,blue] (O6) to[bend right=30] (O7);
\draw[thick,blue] (O7) to[bend right=30] (O8);
\draw[thick,blue] (O8) to[bend right=30] (O5);

\draw[thick,green] (O1) ++ (60:0.7) -- ++ (0:0.7);
\draw[thick,green] (O1) ++ (60:0.7) -- ++ (45:0.7);
\draw[thick,green] (O1) ++ (60:0.7) -- ++ (90:0.7);

\draw[thick,green] (O2) ++ (120:0.7) -- ++ (90:0.7);
\draw[thick,green] (O2) ++ (120:0.7) -- ++ (135:0.7);
\draw[thick,green] (O2) ++ (120:0.7) -- ++ (180:0.7);

\draw[thick,green] (O3) ++ (240:0.7) -- ++ (180:0.7);
\draw[thick,green] (O3) ++ (240:0.7) -- ++ (225:0.7);
\draw[thick,green] (O3) ++ (240:0.7) -- ++ (270:0.7);

\draw[thick,green] (O4) ++ (300:0.7) -- ++ (270:0.7);
\draw[thick,green] (O4) ++ (300:0.7) -- ++ (315:0.7);
\draw[thick,green] (O4) ++ (300:0.7) -- ++ (0:0.7);

\draw[->,>=stealth,very thick] (-4,0) -- (-2,0);

\draw (-2,0) ++ (1,0) node {$\sum_{ \{ J_i \} }$};
\coordinate (A) at (-6,0);
   \draw [domain=0:360,red,thick] plot ({-6+cos(\x)}, {sin(\x)});
   %\draw [thick,domain=60:300] plot ({-6+cos(\x)}, {sin(\x)});
%\draw (A) circle (1);
\path (A) ++(60:1) coordinate (A1);
\path (A) ++(120:1) coordinate (A2);
\path (A) ++(240:1) coordinate (A3);
\path (A) ++(300:1) coordinate (A4);

\draw[thick,blue] (A2) to[bend left=15] (A3);
\draw[thick,blue] (A3) to[bend left=30] (A4);
\draw[thick,blue] (A4) to[bend left=15] (A1);

\draw[thick,green] (A1) -- ++ (0:0.7);
\draw[thick,green] (A1) -- ++ (45:0.7);
\draw[thick,green] (A1) -- ++ (90:0.7);

\draw[thick,green] (A2) -- ++ (90:0.7);
\draw[thick,green] (A2) -- ++ (135:0.7);
\draw[thick,green] (A2) -- ++ (180:0.7);

\draw[thick,green] (A3) -- ++ (180:0.7);
\draw[thick,green] (A3) -- ++ (225:0.7);
\draw[thick,green] (A3) -- ++ (270:0.7);

\draw[thick,green] (A4) -- ++ (270:0.7);
\draw[thick,green] (A4) -- ++ (315:0.7);
\draw[thick,green] (A4) -- ++ (0:0.7);

\draw[->,>=stealth,very thick] (3.5,0) -- (5.5,0);

\draw (5.5,0) ++ (1,0) node {$\sum_{ \{ J_i \} }$};
\coordinate (B) at (8.5,0);
   \draw [domain=0:360,red,thick] plot ({8.5+cos(\x)}, {sin(\x)});
   
\draw (B) ++(-0.6,0);% node {$G$};
\path (B) ++(60:1) coordinate (B1);
\path (B) ++(120:1) coordinate (B2);
\path (B) ++(240:1) coordinate (B3);
\path (B) ++(300:1) coordinate (B4);

\draw[thick,green] (B1) -- ++ (60:0.7) ++ (60:0.3) node [black] {$J_1$};
\draw[thick,green] (B2) -- ++ (120:0.7) ++ (120:0.3) node [black] {$J_2$};
\draw[thick,green] (B3) -- ++ (240:0.7) ++ (240:0.3) node [black] {$J_3$};
\draw[thick,green] (B4) -- ++ (300:0.7) ++ (300:0.3) node [black] {$J_4$};

\draw (A1) node[scale=0.7,blue] {$\bullet$};
\draw (A2) node[scale=0.7,blue] {$\bullet$};
\draw (A3) node[scale=0.7,blue] {$\bullet$};
\draw (A4) node[scale=0.7,blue] {$\bullet$};

\draw (O1) node[scale=0.7,blue] {$\bullet$};
\draw (O2) node[scale=0.7,blue] {$\bullet$};
\draw (O3) node[scale=0.7,blue] {$\bullet$};
\draw (O4) node[scale=0.7,blue] {$\bullet$};

\draw (B1) node[scale=0.7,blue] {$\bullet$};
\draw (B2) node[scale=0.7,blue] {$\bullet$};
\draw (B3) node[scale=0.7,blue] {$\bullet$};
\draw (B4) node[scale=0.7,blue] {$\bullet$};

\end{tikzpicture}
\caption{
The {\color{red}{red}} circle is the path $W$ to be acted by holonomy operator (left). The boundary edges {\color{green}{green}} recouple with bulk edges {\color{blue}{blue}} into bouquet spins $J_1,J_2,J_3,J_4$ at respect vertices (middle). Note that bouquet spins have superposition due to recoupling. The consequence of the loop holonomy operator amounts to acting on a trivalent graph (right).
}
\label{fig:bouquet}
\end{figure}
%%%%

\begin{prop} \label{prop:Amplitudes-6jsymbols}
Given an oriented loop $W[v_1 \overset{ e_{1} }{\to} \cdots \overset{ e_{n-1} }{\to} v_n \overset{ e_{n} }{\to} v_1]$ on the graph $\Gamma$,
%passing through vertices $v_i\,, i=1,\cdots,n$,  connected by $n$ edges $e_i\,, i=1,\cdots,n$. Along the $W$, $e_{i}$ links vertices $v_{i}$ to $v_{i+1}$, $i=1,\cdots, n$ with cycle enumeration $n+1\equiv 1$.
the loop holonomy operator $\wh{\chi_{\ell} } \, \act_{W}$ acts on the spin network basis, labeled by the spins $k_{i}$ on the loop edges and the bouquet spins $J_{i}$ on the loop vertices, by the following transition matrix:
%
%For each vertex $v_{i}$ along $W$, is attached by $e_{i}, e_{i+1}$ (the edges passed through by $\wh{\chi_{\ell} } \, \act_{W}$) and a bouquet edge whose spin is recoupled by other spins attached to the vertex. Each bouquet edge amounts to selecting a channel for intertwiners at vertex $v_{i}$, and bouquet spin $J_{i}$ is a part of intertwiner indices. Then loop holonomy operator $\wh{\chi_{\ell} } \, \act_{W}$'s operation on spin networks can be represented by:
\beq
\tensor{   {\Big[ Z(\Gamma)^{ \{J_i\} }_{\tensor{\chi}{_{\ell} }\, \act_{W}   } \Big] }   } {^{ \{ K_i \} } _{ \{ k_i \} } }
=
(-1)^{\sum_{i=1}^{n} (J_i +k_i +K_i +{\ell} ) }
\,
\overleftarrow{ \prod_{i=1}^{n} \,
\begin{Bmatrix}
J_{i} & K_{i} & K_{i+1} \\
\ell & k_{i+1} & k_{i}
\end{Bmatrix}
 }
\,
\prod_{i=1}^{n} \sqrt{ (2K_i+1)(2k_i+1) }
\,.
\label{eq:Amplitudes-6jsymbols}
\eeq
%In particular, if all vertices $\{ v_{i} \}$ along $W$ are trivalent, then intertwiner living at each vertex is uniquely determined by three spins, i.e., $\vert I_{v_i} \ra \equiv \vert j_{i},k_{i},k_{i+1}\ra$. In other words, $J_{i}=j_{i}$ in (\ref{eq:Amplitudes-6jsymbols}).
\end{prop}
\begin{figure}[hbt!]
\begin{subfigure}[t]{1.0\linewidth}
\begin{tikzpicture}[scale=0.7]

\coordinate (O) at (0,0);

\draw[thick,blue] (O) -- %node[midway,sloped] {$>$} 
++(3,0) node[midway,below=2.5] {$k[g]$};
\draw[red] (O) ++ (1.5,0.75) node [above] {$D^{\ell}_{ab}$} [dashed]-- ++(0,-0.65) node[midway,right] {$\otimes$};

\draw (O) ++(5.5,0) node {$= \ \displaystyle{ \sum_{ K=|k-{\ell} |}^{k+{\ell} } } \, (2K+1)$}
++ (2.1,0) coordinate (O1) [thick,blue]-- %node[midway,sloped] {$<$} 
node[midway,below=2.5] {$k[\id]$}  ++ (2,0) -- %node[midway,sloped] {$>$} 
node[midway,below=2.5] {$K[g]$} ++ (2,0) -- %node[midway,sloped] {$>$}  
node[midway,below=2.5] {$k[\id]$} ++(2,0);

\draw[thick,red] (O1) ++ (2,0) node[scale=0.7] {$\bullet$} -- %node[midway,sloped] {$>$}
++(0,1.5) node[near end,left] {${\ell}[\id]$} node[above] {$b$};
\draw[thick,red] (O1) ++ (4,0) node[scale=0.7] {$\bullet$} -- %node[midway,sloped] {$>$} 
++(0,1.5) node[near end,right] {${\ell}[\id]$} node[above] {$a$};

\end{tikzpicture}
%
%\hspace{2mm}
%
\caption{The graphical illustration for the operation of holonomy operator. The red nodes are virtual vertices introduced by $D^{\ell}$.
}
\label{fig:HolonomyOperator}
\end{subfigure}
\begin{subfigure}[t]{1.0\linewidth}
\begin{tikzpicture}[scale=0.7]

\coordinate (O) at (0,0);
\coordinate (O3) at (-30:0.3);

\draw[thick] (O) -- ++(170:1.5);
\draw[thick] (O) -- ++(110:1.5);
\draw[thick] (O) -- ++(140:1.5);

\begin{knot}[
  consider self intersections=true,
  flip crossing=3,
]
\strand[thick,blue] (O) -- ++ (-50:1.5) --  ++ (-50:0.5);
\strand[thick,blue] (O) -- ++(-20:1.5) -- ++ (-20:0.5);
\strand[thick,blue] (O) -- node[midway,left] {$k_{i+1}$} ++(-90:1.5) ++ (-20:0.3) coordinate (O1);
\strand[thick,blue] (O) -- node[near end,above] {$k_{i}$} ++(20:1.5) ++ (-20:0.3) coordinate (O2);
\strand [line width=1pt,red,->,>=latex,rounded corners] (O2) -- (O3) -- (O1);
\end{knot}

\draw (O) node[scale=0.7,blue] {$\bullet$} ++(60:0.5)node {$v_i$};

\draw (O2) node[right,red] {$\widehat{\chi_{\ell} }$};

\draw[->,>=stealth,very thick] (O) ++(3,0) -- ++ (1.5,0);

\coordinate (A) at (9,0);

\draw (A)  ++ (-3,0) node {$ \displaystyle{ \sum_{ J_i } }$};

\draw[thick] (A) -- ++(170:1.5);
\draw[thick] (A) -- ++(110:1.5);
\draw[thick] (A) -- ++(140:1.5);

\draw[thick,blue] (A) -- ++(90:1.5);
\draw[thick,blue] (A) -- ++(190:1.5);

\draw[thick] (A) -- node[below,near start] {$J_i$} ++(-45:1) coordinate (A1);

\path (A1) ++ (-30:0.3) coordinate (A4);

\draw[thick,blue] (A1) -- node[midway,left] {$k_{i+1}$} ++(-90:1.5) ++ (-20:0.3) coordinate (A2);
\draw[thick,blue] (A1) -- node[midway,above] {$k_{i}$} ++(20:1.5) ++ (-20:0.3)coordinate (A3);

\draw [line width=1pt,red,->,>=latex,rounded corners] (A3) -- (A4) -- (A2);

\draw (A) node[scale=0.7] {$\bullet$};
\draw (A1) node[scale=0.7,blue] {$\bullet$};

\draw (A3) node[right,red] {$\widehat{\chi_{\ell} }$};

\end{tikzpicture}

\caption{For a vertex $v_i$ living along the $W$, its attached spins around either acted by, or not acted by the loop holonomy operator. For the spins which are not acted, they recouple into spin-$J_i$ attached to the vertex $v_i$.
}
\label{fig:BouquetSpins}
\end{subfigure}

\begin{subfigure}[t]{1.0\linewidth}
\begin{tikzpicture}[scale=0.7]

\coordinate (O) at (0,0);

\path (O) ++(18:2.5) coordinate (O1);
\path (O) ++(90:2.5) coordinate (O2);
\path (O) ++(162:2.5) coordinate (O3);
\path (O) ++(234:2.5) coordinate (O4);
\path (O) ++(306:2.5) coordinate (O5);

\draw[thick,blue] (O1) -- (O2) node[pos=0.25, inner sep=0pt, outer sep=0pt](A1) {}
node [pos=0.75, inner sep=0pt, outer sep=0pt] (B1) {}
--  (O3) node [pos=0.25, inner sep=0pt, outer sep=0pt](A2) {} 
node [pos=0.75, inner sep=0pt, outer sep=0pt](B2) {}
--  (O4) node [pos=0.25, inner sep=0pt, outer sep=0pt](A3) {}
node [pos=0.75, inner sep=0pt, outer sep=0pt](B3) {}
 -- (O5) node [pos=0.25, inner sep=0pt, outer sep=0pt](A4) {}
 node [pos=0.75, inner sep=0pt, outer sep=0pt](B4) {}
 -- (O1) node [pos=0.25, inner sep=0pt, outer sep=0pt](A5) {}
 node [pos=0.75, inner sep=0pt, outer sep=0pt](B5) {};
 
\draw[thick] (O1) node[scale=0.7,blue] {$\bullet$} --++(18:0.75);
\draw[thick] (O2) node[scale=0.7,blue] {$\bullet$} --++(90:0.75);
\draw[thick] (O3) node[scale=0.7,blue] {$\bullet$} --++(162:0.75);
\draw[thick] (O4) node[scale=0.7,blue] {$\bullet$} --++(234:0.75);
\draw[thick] (O5) node[scale=0.7,blue] {$\bullet$} --++(306:0.75);

\draw[thick,red] ($(A1)!0!90:(O1)$) node[scale=0.7] {$\bullet$} -- ($(A1)!0.65cm!-90:(O1)$);
\draw[thick,red] ($(B1)!0!90:(O1)$) node[scale=0.7] {$\bullet$} -- ($(B1)!0.65cm!-90:(O1)$);

\draw[thick,red] ($(A2)!0!90:(O2)$) node[scale=0.7] {$\bullet$} -- ($(A2)!0.65cm!-90:(O2)$);
\draw[thick,red] ($(B2)!0!90:(O2)$) node[scale=0.7] {$\bullet$} -- ($(B2)!0.65cm!-90:(O2)$);

\draw[thick,red] ($(A3)!0!90:(O3)$) node[scale=0.7] {$\bullet$} -- ($(A3)!0.65cm!-90:(O3)$);
\draw[thick,red] ($(B3)!0!90:(O3)$) node[scale=0.7] {$\bullet$} -- ($(B3)!0.65cm!-90:(O3)$);

\draw[thick,red] ($(A4)!0!90:(O4)$) node[scale=0.7] {$\bullet$} -- ($(A4)!0.65cm!-90:(O4)$);
\draw[thick,red] ($(B4)!0!90:(O4)$) node[scale=0.7] {$\bullet$} -- ($(B4)!0.65cm!-90:(O4)$);

\draw[thick,red] ($(A5)!0!90:(O5)$) node[scale=0.7] {$\bullet$} -- ($(A5)!0.65cm!-90:(O5)$);
\draw[thick,red] ($(B5)!0!90:(O5)$) node[scale=0.7] {$\bullet$} -- ($(B5)!0.65cm!-90:(O5)$);

\draw[->,>=stealth,very thick] (O) ++(4.5,0) -- node[above] {$\tr$} ++ (2,0);

\path (O) ++ (11,0) coordinate (P);

\path (P) ++(18:2.5) coordinate (P1);
\path (P) ++(90:2.5) coordinate (P2);
\path (P) ++(162:2.5) coordinate (P3);
\path (P) ++(234:2.5) coordinate (P4);
\path (P) ++(306:2.5) coordinate (P5);

\draw[thick] (P1) node[scale=0.7,blue] {$\bullet$} --++(18:0.75);
\draw[thick] (P2) node[scale=0.7,blue] {$\bullet$} --++(90:0.75);
\draw[thick] (P3) node[scale=0.7,blue] {$\bullet$} --++(162:0.75);
\draw[thick] (P4) node[scale=0.7,blue] {$\bullet$} --++(234:0.75);
\draw[thick] (P5) node[scale=0.7,blue] {$\bullet$} --++(306:0.75);

\draw[thick,blue] (P1) -- (P2) node[pos=0.25, inner sep=0pt, outer sep=0pt](C1) {}
node [pos=0.75, inner sep=0pt, outer sep=0pt] (D1) {}
--  (P3) node [pos=0.25, inner sep=0pt, outer sep=0pt](C2) {} 
node [pos=0.75, inner sep=0pt, outer sep=0pt](D2) {}
--  (P4) node [pos=0.25, inner sep=0pt, outer sep=0pt](C3) {}
node [pos=0.75, inner sep=0pt, outer sep=0pt](D3) {}
 -- (P5) node [pos=0.25, inner sep=0pt, outer sep=0pt](C4) {}
 node [pos=0.75, inner sep=0pt, outer sep=0pt](D4) {}
 -- (P1) node [pos=0.25, inner sep=0pt, outer sep=0pt](C5) {}
 node [pos=0.75, inner sep=0pt, outer sep=0pt](D5) {};

\draw[thick,red] (D5) node[scale=0.7] {$\bullet$} to [bend left=60] (C1) node[scale=0.7] {$\bullet$};
\draw[thick,red] (D1) node[scale=0.7] {$\bullet$} to [bend left=60] (C2) node[scale=0.7] {$\bullet$};
\draw[thick,red] (D2) node[scale=0.7] {$\bullet$} to [bend left=60] (C3) node[scale=0.7] {$\bullet$};
\draw[thick,red] (D3) node[scale=0.7] {$\bullet$} to [bend left=60] (C4) node[scale=0.7] {$\bullet$};
\draw[thick,red] (D4) node[scale=0.7] {$\bullet$} to [bend left=60] (C5) node[scale=0.7] {$\bullet$};

\end{tikzpicture}

\caption{The graphical illustration for loop holonomy operator. The $D^{\ell}$ amounts to generates bubbles around the corners.
}
\label{fig:LoopHolonomyOperator}
\end{subfigure}
\begin{subfigure}[t]{1.0\linewidth}
\begin{tikzpicture}[scale=0.7]

\coordinate (O) at (0,0);

\path (O) --++(180:1) coordinate (A);
\path (O) --++(30:1.5) coordinate (B);
\path (O) --++(330:1.5) coordinate (C);

\draw[thick] (O) -- (A) ++ (180:0.3) node {$J_1$};
\draw[thick,blue] (O)node[scale=0.7] {$\bullet$} --node [sloped,above,near start] {$k_2$} (B) node[pos=0.65, inner sep=0pt, outer sep=0pt](B1) {} ++ (30:0.35) node {$K_2$};
\draw[thick,blue] (O) -- node [sloped,below,near start] {$k_3$} (C) node[pos=0.65, inner sep=0pt, outer sep=0pt](C1) {} ++ (330:0.35) node {$K_3$};

\draw[thick,red] (B1) node[scale=0.7] {$\bullet$} -- node[right] {${\ell}$} (C1) node[scale=0.7] {$\bullet$};

\draw (O) ++ (3,0) node{$=$} ++(3,0) coordinate (P);

\path (P) --++(180:1) coordinate (D);
\path (P) --++(45:1.2) coordinate (E);
\path (P) --++(315:1.2) coordinate (F);

\draw[thick] (P) -- (D) ++ (180:0.3) node {$J_1$};
\draw[thick,blue] (P)node[scale=0.7] {$\bullet$} -- (E) ++ (30:0.35) node {$K_2$};
\draw[thick,blue] (P) -- (F) ++ (330:0.35) node {$K_3$};

\draw (P) ++ (2,0) node{$\times$} ++(3,0) coordinate (Q);

\draw[thick] (Q) -- node[right, near start]{$J_1$} ++(90:1.5) coordinate (D1);
\draw[thick,blue] (Q) -- node[below, near start]{$k_3$} ++(210:1.8) coordinate(E1);
\draw[thick,blue] (Q) -- node[below, near start]{$k_2$} ++(-30:1.8) coordinate (F1);

\draw[thick,red] (E1) -- node[below] {$\ell$} (F1);

\draw[thick,blue] (D1) --node[left,midway]{$K_3$}  (E1) ;
\draw[thick,blue] (D1) -- node[right,midway]{$K_2$} (F1);

\end{tikzpicture}

\caption{Each corner brings a $6j$-symbol to wave function.
}
\label{fig:LoopHolonomyOperator-6jsymbols}
\end{subfigure}

\caption{The graphical representations for holonomy and loop holonomy operator respectively.
}
\label{fig:LoopHolonomyOperator-Spinfoam}
\end{figure}

%, about which the reader can find details and explanations in appendix \ref{prop:Amplitudes-6jsymbols}.
%
\begin{proof}
The proof is a straightforward spin recoupling computation.
%
%Consider the particular trivalent vertices situation.
%
We split the character ${\chi_{\ell} }(G_{W})$ into local Wigner D-matrices for each edge of the loop,
\be
\label{eqn:character}
{\chi_{\ell} }(G_{W})=
\overleftarrow{ \prod_{i=1}^{n} D^{\ell}_{m_{i} m_{i+1} } (g_{i}) }
\,,
\qquad \text{cycling} \quad n+1\equiv 1
\,,
\ee
then apply the formula (\ref{eq:HolonomyOperator-3j}) for the holonomy operators. As illustrated on fig.\ref{fig:LoopHolonomyOperator-Spinfoam}, the bulk spins $k_{i}$  undergo a spin shift  $k_i \to K_i$ with the output spins $K_i$ are bounded by triangular inequalities $| k_i -{\ell} | \leq K_i \leq k_i+{\ell}$, while the boundary spins $j_{i}$ and thus the recouped bouquet spins $J_{i}$ are left unchanged. Hence we  use spins the $\{ j_i,J_{i} \} \,, \{k_i\}$ to label the initial intertwiners $\{I_v\}$  and $\{ j_{i},J_{i} \}, \{ K_{i} \}$ to label the final intertwiners $\{I'_{v}\}$,
\beq
\tensor{   {\Big[ Z(\Gamma)^{ \{j_i,J_{i}\} }_{\tensor{\chi}{_{\ell}}\, \act_{W}   } \Big] }   } {^{ \{ K_i \} } _{ \{ k_i \} } }
\equiv
\tensor{   {\Big[ \tensor{Z(\Gamma) }{ _{\tensor{\chi}{_{\ell}}\, \act_{W}    }    } \Big] }   } {^{ \{ I'_{v} \} } _{ \{ I_{v} \} } }
\,. \nn
\eeq
As drawn on fig.\ref{fig:HolonomyOperator}, the action of a holonomy operator on an edge amounts to creating two open edges with weight factor $(2K_{i}+1)$.
Summing over the magnetic indices  of the character formula \eqref{eqn:character} amounts to linking those open edges and closing the corners all around the loop.
This leads to a $6j$-symbol factor for each corner/vertex as shown on fig.\ref{fig:LoopHolonomyOperator-6jsymbols}, leading to the overall  amplitude for the loop holonomy operator $\wh{\chi_{\ell} } \, \act_{W} \,$ given by
\beq
 \label{eq:WaveFunction-Amplitudes-6jsymbols}
 \tensor{   {\Big[ \cW(\Gamma)^{ \{j_i,J_{i} \} }_{\tensor{\chi}{_{\ell} }\, \act_{W}   } \Big] }   } {^{ \{ K_i \} } _{ \{ k_i \} } }
=
(-1)^{\sum_{i=1}^{n} (J_i+k_i +K_i +{\ell} ) }
\,
\overleftarrow{ \prod_{i=1}^{n} \,
\begin{Bmatrix}
{J}_i & K_{i} & K_{i+1} \\
\ell & k_{i+1} & k_{i}
\end{Bmatrix}
 }
\,
\prod_{i=1}^{n} (2K_i+1)
\,, \qquad \text{with} \quad n+1\equiv 1
\,,
\eeq
which leads to the desired $Z$'s representation (\ref{eq:Amplitudes-6jsymbols}).
%Putting it into (\ref{eq:Amplitudes}), we obtain $Z$'s representation (\ref{eq:Amplitudes-6jsymbols}) for particular case $j_i=J_i$.
%The weight factor $\sqrt{ (2k_i+1)(2K_i+1) }$ is obtained because of group integral (\ref{eq:Amplitudes}).

%With the help of recoupled basis, we can generalize above analysis towards the situation with arbitrary valency vertices. Each attached spin of $v_{i} \in W$ is either acted by or not acted by $\wh{\chi_{\ell} } \, \act_{W}$. Let $k_{i}$ and $k_{i+1}$ be two successive spins acted by $\wh{\chi_{\ell} } \, \act_{W}$ where $k_{i}$ ingoing and $k_{i+1}$ outgoing with respect to the vertex $v_{i}$. Meanwhile, for the rest of attached spins to $v_i$, they recouple into spins-$J_{i}$ living at $v_{i}$ (as fig.\ref{fig:BouquetSpins}) - probably in spin superposition. Substituting the recoupled spins-$\{ J_{i} \}$ for boundary spins-$\{ j_i \}$, the (\ref{eq:Amplitudes}) for arbitrary valency is obtained.
\end{proof}
The proposition shows that the dynamics generated by the loop holonomy operator depends at each vertex on three spins: the two spins living on the edges attached to the vertex on the loop and the bouquet spin encoding the recoupling of the spin living on all the other edges attached to the considered vertex.
%
%the dynamics information of loop holonomy operator is carried by spins $\{ k_{i} \}$ along $W$ and bouquet spins $\{ J_{i} \}$ attached to vertices along $W$. More precisely, for each vertex along $W$, the dynamics information is carried by three spins. Hence the dynamics driven by loop holonomy operator can reflect onto the corresponding graph where all vertices are trivalent and bouquet spins amount to being boundary spins of $W$.

\medskip

We can provide the action of the loop holonomy operator with geometrical meaning in the semi-classical regime at the large spin limit.
This is provided by  the asymptotic behavior of $Z$, which can be expressed in terms of the angles of the triangles dual to the spin network graph, as shown on figure \ref{fig:LoopHolonomyOperator-DualTriangulation}. More precisely, considering that the spins $\{ j_i,k_i \}$ of spin network are much larger than the spin-$\ell$ of holonomy operator, we employ Edmonds's asymptotic formula and get the following corollary of the previous proposition:
\begin{coro}
If the bulk spins $\{ k_i \}$ and boundary spins $\{j_i\}$ much larger than $\ell$, i.e., $\{ j_i \,, k_i \} \gg \ell$, matrix $Z$ has following asymptotic formula in terms of Wigner d-matrices:
\beq
\tensor{   {\Big[ Z(\Gamma)^{ \{j_i\} }_{\tensor{\chi}{_{\ell} }\, \act_{W}   } \Big] }   } {^{ \{ K_i \} } _{ \{ k_i \} } }
\approx
\overleftarrow{ \prod_{i=1}^{n}   d^{\ell}_{\veps_{i} \veps_{i+1} }( \theta_{i} ) }
\,, \quad \text{with} \quad \veps_i=K_i-k_i
\,, \quad
\cos\theta_{i}
=
\f{ k_i(k_i+1) +k_{i+1}(k_{i+1}+1) - j_i(j_i+1)  }{ 2 \sqrt{ k_{i} (k_{i}+1) k_{i+1} (k_{i+1}+1) } }
\,.
\label{eq:Angle-Asymptotic-6j}
\eeq
\end{coro}
\begin{proof}
Using Edmonds's asymptotic formula (see equations (3.6) and (3.7) in \cite{osti_4824659}) for $6j$-symbols,
\be \label{eq:EdmondsAsymptoticFormula}
\begin{Bmatrix}
c & a & b \\
{\ell} & b +\veps_2 & a+\veps_1
\end{Bmatrix}
\approx
\f{(-1)^{a+b+c+{\ell}+\veps_2} }{ \sqrt{ (2a+1)(2b+1) } } d^{\ell}_{ \veps_2 \veps_1 } (\theta)
\,, \quad \text{with} \quad
\cos\theta=\f{ a(a+1)+b(b+1)-c(c+1) }{ 2\sqrt{ a(a+1)b(b+1) } }
\,,
\ee
where $d^{\ell}$ is the reduced Wigner d-matrix for spin-${\ell}$,
\be
d^{\ell}_{ m n}(\theta)
=
\sqrt{ ({\ell}+m)!({\ell}-m)!({\ell}+n)!({\ell}-n)! }
\sum_{s} \f{ (-1)^{m-n+s} \big( \cos\f{\theta}{2} \big)^{2{\ell}-m+n-2s} \big( \sin\f{\theta}{2} \big)^{m-n-2s} }{ ({\ell}+n-s)!s!(m-n+s)!({\ell}-m-s)! }
\,.
\ee
Inserting this asymptotic approximation of the Wigner d-matrices in the expression (\ref{eq:Amplitudes-6jsymbols}) of the $6j$-symbols, we get the wanted formula (\ref{eq:Angle-Asymptotic-6j}) up to a pre-factor $\sqrt{ \f{ 2K_{i}+1 }{ 2k_{i}+1 }  }$, which actually goes to 1 when $k_i,K_{i}$ are  large compared to $\ell$.

\end{proof}
\begin{figure}[t]
\begin{subfigure}[t]{0.6\linewidth}
\begin{tikzpicture}[scale=0.7]

\coordinate (O) at (0,0);

\path (O) ++(18:1.5) coordinate (O1);
\path (O) ++(90:1.5) coordinate (O2);
\path (O) ++(162:1.5) coordinate (O3);
\path (O) ++(234:1.5) coordinate (O4);
\path (O) ++(306:1.5) coordinate (O5);

\draw[thick] (O1) node[scale=0.7,blue] {$\bullet$} --++(18:0.75);
\draw[thick] (O2) node[scale=0.7,blue] {$\bullet$} --++(90:0.75);
\draw[thick] (O3) node[scale=0.7,blue] {$\bullet$} --++(162:0.75);
\draw[thick] (O4) node[scale=0.7,blue] {$\bullet$} --++(234:0.75);
\draw[thick] (O5) node[scale=0.7,blue] {$\bullet$} --++(306:0.75);

\draw[thick,blue] (O1) -- (O2) node[pos=0.25, inner sep=0pt, outer sep=0pt](A1) {}
node [pos=0.75, inner sep=0pt, outer sep=0pt] (B1) {}
--  (O3) node [pos=0.25, inner sep=0pt, outer sep=0pt](A2) {} 
node [pos=0.75, inner sep=0pt, outer sep=0pt](B2) {}
--  (O4) node [pos=0.25, inner sep=0pt, outer sep=0pt](A3) {}
node [pos=0.75, inner sep=0pt, outer sep=0pt](B3) {}
 -- (O5) node [pos=0.25, inner sep=0pt, outer sep=0pt](A4) {}
 node [pos=0.75, inner sep=0pt, outer sep=0pt](B4) {}
 -- (O1) node [pos=0.25, inner sep=0pt, outer sep=0pt](A5) {}
 node [pos=0.75, inner sep=0pt, outer sep=0pt](B5) {};

\path (O) ++(-18:2.3) coordinate (C1);
\path (O) ++(54:2.3) coordinate (C2);
\path (O) ++(126:2.3) coordinate (C3);
\path (O) ++(198:2.3) coordinate (C4);
\path (O) ++(270:2.3) coordinate (C5);
\path (O) ++(342:2.3) coordinate (C6);

\draw[thick,brown] (C1) --(C2) -- (C3) --(C4) -- (C5) --(C6);
\draw[thick,brown] (O) --(C1);
\draw[thick,brown] (O) --(C2);
\draw[thick,brown] (O) --(C3);
\draw[thick,brown] (O) --(C4);
\draw[thick,brown] (O) --(C5);
\draw[thick,brown] (O) --(C6);

\draw[thick] (O) node[scale=0.7,brown] {$\bullet$};

\pic [draw, ->, "$\theta_1$", angle radius=0.3cm, angle eccentricity=1.75] {angle = C1--O--C2};
\pic [draw, ->, "$\theta_2$", angle radius=0.3cm, angle eccentricity=1.75] {angle = C2--O--C3};
\pic [draw, ->, "$\theta_3$", angle radius=0.3cm, angle eccentricity=1.75] {angle = C3--O--C4};
\pic [draw, ->, "$\theta_4$", angle radius=0.3cm, angle eccentricity=1.75] {angle = C4--O--C5};
\pic [draw, ->, "$\theta_5$", angle radius=0.3cm, angle eccentricity=1.75] {angle = C5--O--C6};

\end{tikzpicture}

\caption{The {\color{brown}{brown}} lines portray the dual triangulation of spin network. The angles $\theta_i$ are arguments of Wigner d-matrices $d^{\ell}$ in asymptotic formula (\ref{eq:Angle-Asymptotic-6j}).
}
\label{fig:LoopHolonomyOperator-DualTriangulation}
\end{subfigure}
\begin{subfigure}[t]{0.35\linewidth}
\begin{tikzpicture}[scale=0.7]

\coordinate (P) at (11,0);
\coordinate (A) at (11,2);
\coordinate (B) at (10,-0.5);
\coordinate (C) at (12,-0.5);

\draw[brown] (A) -- (B) -- (C) -- cycle;

\draw (P) -- ++ (160:1.5) ++ (160:0.35) node {$a$};

\draw (P) -- ++ (25:1.6) ++ (25:0.35) node {$b$};

\draw (P) -- ++ (270:1.0) ++ (270:0.35) node {$c$} ;

\draw (P) node[scale=0.7] {$\bullet$};

\pic [draw, "$\theta$", angle radius=0.3cm, angle eccentricity=1.75] {angle = B--A--C};

\end{tikzpicture}

\caption{The angle variable in Edmonds’s asymptotic formula (\ref{eq:EdmondsAsymptoticFormula}), where $a,b,c$ are spins.
}
\label{fig:EdmondsAsymptoticFormula}
\end{subfigure}
\caption{The illustration for asymptotic formula.}
\end{figure}
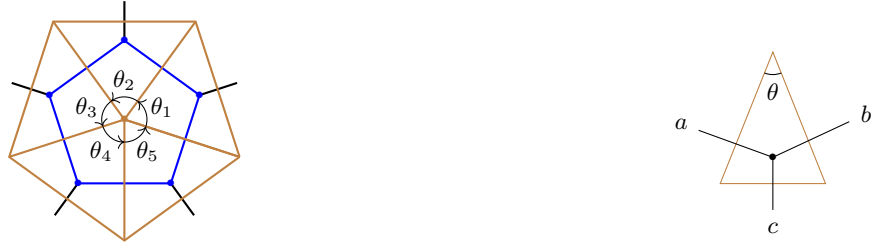
This approximation expresses the arguments $\{j_i\,k_i,K_i \}$ in terms of corner angles $\theta_i$ and spin shiftings $\veps_i=K_i-k_i$, as illustrated in {{fig. \ref{fig:EdmondsAsymptoticFormula}}}.
%
%meanwhile expresses the argument $\ell$ in Wigner d-matrix.
In addition, one can check the large spin approximation (\ref{eq:Angle-Asymptotic-6j})  inherits the same exact composition rule as the exact amplitude (\ref{eq:Composition-Amplitudes}).
%{\bf TRUE? PROOF??}{\bf  This property does not rely on the scale of the arguments, i.e., the below equality always holds.}

%Now we explain why we said the composite rule (\ref{eq:CompositeRule-Wigner-small-d}) inherited from composite rule (\ref{eq:Composition-Amplitudes}). 
For instance, in the case $n=2$, to ease the notations,  the composition rule \eqref{eq:Composition-Amplitudes}  for the $Z$-matrix elements reads
\beq
&&\sum_{ K_1,K_2}
(-1)^{k_1+k_2-\tk_1-\tk_2+2\ell_1-2\ell_2}
  \begin{Bmatrix}
   j_1 & k_1 & k_2 \\
   \ell_1 &  K_2 & K_1
  \end{Bmatrix}
\begin{Bmatrix}
   j_{2} & k_1 & k_2 \\
   \ell_1 &  K_2 & K_1
  \end{Bmatrix}
  \begin{Bmatrix}
   j_1 & \tk_1 & \tk_2 \\
   \ell_2 &  K_2 & K_1
  \end{Bmatrix}
\begin{Bmatrix}
   j_{2} & \tk_1 & \tk_2 \\
   \ell_2 &  K_2 & K_1
  \end{Bmatrix}
  \nn \\
&& \times \sqrt{(2k_1+1)(2k_2+1)(2\tk_1+1)(2\tk_2+1)}(2K_1+1)(2K_2+1)
\nn \\
&=&
  \sum_{ s=|\ell_1-\ell_2|}^{\ell_1+\ell_2}
(-1)^{j_{1}+j_{2}+\tk_1+\tk_2+k_1+k_2+2s}
\begin{Bmatrix}
   j_1 & k_1 & k_2 \\
   s &  \tk_2 & \tk_1
  \end{Bmatrix}
\begin{Bmatrix}
   j_{2} & k_1 & k_2 \\
   s &  \tk_2 & \tk_1
  \end{Bmatrix}
  \sqrt{(2k_1+1)(2k_2+1)(2\tk_1+1)(2\tk_2+1)}
\,,  \label{eq:CompositeRelation-candygraph-6js}
\eeq
which is equivalent to Biedenharn–Elliott identity (see e.g. \cite{Bonzom:2009zd}).
% together with the sum rules of $6j$-symbols (referring to \cite{brink1968angular}, page 142).
For large spins $\{ j_i \,, k_i \} \gg \ell$, one can apply Edmonds’s asymptotic formula for $6j$-symbols to get:
\be
\sum_{\veps_1,\veps_2}
d^{\ell_1}_{\veps_1 \veps_2}(\theta_1) \, d^{\ell_1}_{\veps_2 \veps_1}(\theta_2) \, d^{\ell_2}_{\veps_1+\Delta_1, \veps_2+\Delta_2}(\theta_1) \, d^{\ell_2}_{\veps_2+\Delta_2, \veps_1+\Delta_1}(\theta_2)
\approx
\sum_{s=| \ell_1-\ell_2| }^{\ell_1+\ell_2}
d^{s}_{\Delta_1 , \Delta_2}(\theta_1) \, d^{s}_{\Delta_2, \Delta_1}(\theta_2)
\,.
%\label{eq:CompositeRule-Wigner-small-d-n=2}
\ee
Below we show that this approximative composition rule actually holds exctaly.
% (\ref{eq:CompositeRule-Wigner-small-d-n=2}).
%We demand angle of triangle $\theta_1=(k_1,j_1,k_2)$ equal to angle of triangle $\tilde{\theta}_1=(\tk_1,j_1,\tk_2)$, because $| k_i-\tk_i |$ are sufficient small due to $| k_i-\tk_i | \leq \ell_1+\ell_2 \ll \{j_i,k_i\}$.
%
\begin{prop}
Suppose there is $n$ vertices along loop $W$, then the composite rule for eqn.(\ref{eq:Angle-Asymptotic-6j}) is
\be
\sum_{ \{ \veps_i \} }
\overleftarrow{ \prod_{i=1}^{n} d^{\ell_1}_{\veps_i , \veps_{i+1} } (\theta_{i}) \, d^{\ell_2}_{\veps_i + \Delta_i, \veps_{i+1} + \Delta_{i+1} } (\theta_{i}) }
=
\sum_{s=| \ell_1-\ell_2| }^{\ell_1+\ell_2}
\overleftarrow{ \prod_{i=1}^{n} d^{s}_{\Delta_i , \Delta_{i+1} }(\theta_i) }
\,,
\qquad \text{cycling} \quad n+1\equiv 1
\,.
\label{eq:CompositeRule-Wigner-small-d}
\ee
\end{prop}
\begin{proof}
Let us consider the simplest case $n=2$. The proof is straightforward to generalize to arbitrary $n$. We would like to prove the following composition rule:
\be
\sum_{\veps_1,\veps_2}
d^{\ell_1}_{\veps_1 \veps_2}(\theta_1) \, d^{\ell_1}_{\veps_2 \veps_1}(\theta_2) \, d^{\ell_2}_{\veps_1+\Delta_1, \veps_2+\Delta_2}(\theta_1) \, d^{\ell_2}_{\veps_2+\Delta_2, \veps_1+\Delta_1}(\theta_2)
=
\sum_{s=| \ell_1-\ell_2| }^{\ell_1+\ell_2}
d^{s}_{\Delta_1 , \Delta_2}(\theta_1) \, d^{s}_{\Delta_2, \Delta_1}(\theta_2)
\,.
\label{eq:CompositeRule-Wigner-small-d-n=2}
\ee
This equation can be proven by means of recoupling Wigner d-matrices. Firstly, we flip the sign of magnetic indices in $d^{\ell_1}$ by equation $d^{j}_{mn}(\theta)=(-1)^{m-n}d^{j}_{-m,-n}(\theta)$. The overall phase has to be $1$ because every $\veps_i$ appears twice. We then recouple Wigner d-matrices $d^{\ell_1}$ and $d^{\ell_2}$ with argument $\theta_1$
\beq
&&
d^{\ell_1}_{-\veps_1,-\veps_2}(\theta_1) \, d^{\ell_2}_{\veps_1+\Delta_1, \veps_2+\Delta_2}(\theta_1)
\nn
\\
&=&
\sum_{s=|\ell_1-\ell_2|}^{\ell_1+\ell_2}
(-1)^{\Delta_1-\Delta_2} (2s+1)
\begin{pmatrix}
\ell_1 & \ell_2 & s \\
-\veps_1 & \veps_1 + \Delta_1 & -\Delta_1
\end{pmatrix}
\begin{pmatrix}
\ell_1 & \ell_2 & s \\
-\veps_2 & \veps_2 + \Delta_2 & -\Delta_2
\end{pmatrix}
d^{s}_{\Delta_1,\Delta_2}(\theta_1)
\,,
\eeq
where tensor product $d^{\ell_1} \otimes d^{\ell_2}$ are decomposed into Wigner d-matrices labeled by $s$. Repeat the step $d^{\ell_1} \otimes d^{\ell_2}$ for argument $\theta_2$, and label the Wigner d-matrices by $d^{s'}$.
Now fixing $s$ and $s'$, we take all $3j$-symbols into account, and deal with them by orthogonality of $3j$-symbols so
\be
\sum_{\veps_1,\veps_2}
\begin{pmatrix}
\ell_1 & \ell_2 & s \\
-\veps_1 & \veps_1 + \Delta_1 & -\Delta_1
\end{pmatrix}
\begin{pmatrix}
\ell_1 & \ell_2 & s \\
-\veps_2 & \veps_2 + \Delta_2 & -\Delta_2
\end{pmatrix}
\begin{pmatrix}
\ell_1 & \ell_2 & s' \\
-\veps_2 & \veps_2 + \Delta_2 & -\Delta_2
\end{pmatrix}
\begin{pmatrix}
\ell_1 & \ell_2 & s' \\
-\veps_1 & \veps_1 + \Delta_1 & -\Delta_1
\end{pmatrix}
=
\f{ \delta_{ss'} }{(2s+1)(2s'+1)}
\,.
\ee
The orthogonality eliminates all the $3j$-symbols appearing in the spin recoupling, and we finally we recover composite rule (\ref{eq:CompositeRule-Wigner-small-d-n=2}).
%The $n=2$ composite rule can be straightforwardly generalized to arbitrary $n$.
\end{proof}

%\medskip

The sum of corner angles defines a deficit angle around a dual vertex in triangulation by $\delta=2\pi-\sum_{i}\theta_i$, which is interpreted as the discretization of spatial curvature \cite{Rovelli:2014ssa}.
At the end of the day, we have shown how to define and represent the action of loop holonomy operator on spin networks. The dynamics of the operator is encoded in the spins along loop and bouquet spins around. In the semi-classical regime defined in the large spin limit, the operator relates to the corner angles along the loop.

Now that we have clarified the geometrical interpretation of the loop holonomy operator, we are interested in the entanglement it creates, in the purpose of exploring the relation between quantum geometry and quantum information in the framework of loop quantum gravity.

%%%%%%%%%
\section{Multipartite entanglement and geometric measure of entanglement}
\label{Section:MultipartiteEntanglement}
%%%%%%%%%

In order to investigate and quantify the entanglement structure in LQG, the notion of multipartite entanglement is required. Indeed, bipartite entanglement, such as the entanglement between two vertices, is understood to reflect the distance between parts of the quantum state of geometry. General operators, such the loop holonomy, will inevitably create multipartite entanglement between the vertices it acts upon. Since the loop holonomy operator creates curvature excitations, we wish to shed light on the relation between geometric curvature and multipartite entanglement, in order to open the door to the possibility of defining curvature at the quantum level directly in quantum information terms.

At the technical level now, due to the failure of generalized Schmidt decomposition in most of cases of multipartite system, the von Neumann entropy of reduced density matrix is not a suitable entanglement measure anymore. Measures for multipartite entanglement are needed.
Many bipartite entanglement measures, such as relative entropy of entanglement, are generalized for studying multipartite entanglement. More subsystems and higher dimension of individual Hilbert space, leads to more parameters to describe the entanglement, thus many  entanglement measures can be constructed. They are not a priori equivalent. The partial objective of the present work is to suggest a suitable multipartite entanglement for LQG.

%%%
\subsection{Separable and entangled spin network states}
%%%

The quantum entanglement between spin sub-networks, is quantified as the \textit{intertwiner entanglement} \cite{Livine:2017fgq} where spin networks are understood as a many-body quantum system. In order to study multipartite entanglement, we employ the geometric measure of entanglement, which requires us to classify the set of states and distinguish fully separable states \cite{GUHNE20091}, i.e., product states as for instance $\rho_{ABC}=\rho_{A} \otimes \rho_{B} \otimes \rho_{C}$. In a fully separable state, subsystems are unentangled.

%\medskip

The present work's purpose is to investigate bulk entanglement on spin network and not to focus on the boundary structures. The difference between bulk and boundary entanglement is described in \cite{Livine:2017fgq}. We thus look at the Hilbert space of bulk spin network states $\cH_{\Gamma^{o}}$ as  the tensor product of the intertwiner Hilbert spaces at every vertex,
%
%The spin network Hilbert space $\cH_{\Gamma}$ is the tensor product Hilbert space of intertwiner spaces,
%
%{\bf PROBLEMS with spin matching constraints !! Also include graph with boundary. Here focus on Hilbert space of bulk states and bulk entanglement. Write $\cH_{\Gamma^{o}}$ ??}
%
\be
\cH_{ \Gamma^{o} }
=
\bigoplus_{\{j_{e}\}_{e\in\Gamma}}
\cH_{v}^{\{j_{e}\}_{e\ni v}}
\subset
\bigotimes_{v\in\Gamma } \, \cH_{v}\,,
%\,, \qquad \text{where} \quad
\ee
where the vertex Hilbert spaces are defined as
\be
\cH_{v}^{\{j_{e}\}_{e\ni v}}
=
\textrm{Inv}_{\SU(2)}\Big{[}
\bigotimes_{e|\,v=s(e)} \cV_{j_{e}}
\otimes
\bigotimes_{e|\,v=t(e)} \cV_{j_{e}}^{*}
\Big{]}
\quad\textrm{and}\quad
\cH_{v}
=
\bigoplus_{ \{j_{e} \}_{e\ni v} }
\cH_{v}^{\{j_{e}\}_{e\ni v}}
\,.
\ee
We consider spin networks as states in the larger Hilbert space $\bigotimes_{v\in\Gamma } \, \cH_{v}$ of tensor products of intertwiners without imposing the spin matching constraints along the bulk edges $e\in \Gamma^{o}$. The advantage with this starting point is that we are directly looking at correlations and entanglement between $\SU(2)$-gauge invariant excitations -the intertwiners- and that we do not have to worry about gauge breaking and correlations between non-gauge invariant observables (see e.g. \cite{Donnelly:2008vx,Donnelly:2016auv,Livine:2017fgq} for a discussion on this issue).

%
%{\bf Here, too much repetition, this is just the decomposition on the spin network basis.}
%
%Here for each $\cH_{v}$, we allow the superposition over the spins carried by edges that are attached to $v$. For trivalent vertex, the $\SU(2)$ invariant subspace is entirely determined by the three spins. For higher (than three) valency vertex, the intertwiners (belonging to the $\SU(2)$ invariant subspace) is not only determined by the spins, but also internal indices for intertwiner space, i.e., the multiplicity of spin recoupling. We also allow the superposition over these internal indices. Nevertheless, since $I_{v}$ could contain both spin indices and internal indices for intertwiner at $v$,
%
A general spin network state can be decomposed as a superposition over spin network basis states:
\beq  \label{eq:SPN-decomposition}
| \psi_{\Gamma} \ra
&=&
\sum_{\{I_v\}}
C_{ \Gamma } ( \{ I_{v} \} )
 \,
\bigotimes_{v\in\Gamma}
| \Psi_{ v, I_v } \ra
\,, \qquad \text{where} \quad
| \Psi_{ \Gamma, \{ I_v \} } \ra
=
\bigotimes_{v\in\Gamma}
| \Psi_{ v, I_v } \ra
\,.
\eeq
Here the intertwiner basis state $| \Psi_{ v, I_v } \ra \in \cH_{v}$ have definite spins and intertwiner, with spins and internal intertwiner indices packaged in the labels $I_{v}$. Then the coefficients $C_{ \Gamma } ( \{ I_{v} \} )$ for a general state allows for superpositions of both spins and intertwiners, thus leading to correlation between intertwiner states located at different vertices.

\medskip

To define multipartite entanglement and understand how spin sub-networks are entangled, we need to identify the set of fully separable spin network state. Then we will define the geometric entanglement  carried by a state as its distance to the set of separable states. 
%
%To do this, we get to know separable or entangled states by endowing spin- or intertwiner-superposition (correlation) separately.
%
Let us thus describe the hierarchy of potential ways to entangle a spin network state.
%
%{\bf TO BE REVISE by underlining relevance and role of boundary !!}
%
To start with, fully separable (pure) states are states with definite values of spins and intertwiners, that is spin network basis states (up to the choice of a local basis of intertwiner at each vertex). Then entanglement amounts to non-locally factorizable coefficients $C_{ \Gamma } ( \{ I_{v} \} )$. We distinguish three sources of entanglement:
\begin{enumerate}
\item Entanglement resulting from the correlation between intertwiner states at different vertices for fixed spins on the edges;

\item Entanglement resulting from spin superpositions on edges and thereby creating entanglement between the intertwiners living at the vertices linked by edges carried such spin superpositions;

\item Entanglement resulting from spin correlations between different edges.

\end{enumerate}

Below, we characterize these three levels in more details.

\medskip

Let us start by considering two adjacent vertices linked by a certain number of edges. We assume no spin-superposition over these edges. As long as the intertwiners at the two vertices remain uncorrelated (i.e. the coefficient $C(I_{v_1},I_{v_2})=C(I_{v_1})C(I_{v_2})$ factorizes), no matter the intertwine superposition at each vertex, the state remains unentangled.
Explicitly, this type of uncorrelated but superposed $v_{1},v_{2}$ state (with no spin-superposition) reads as:
\be
| \psi \ra
=
\cdots
\left( \sum_{ I_{v_{1} }^{ ( \{ j_e\}_{ e\ni v_{1} } ) } }
C( I_{v_{1} }^{ ( \{ j_e\}_{ e\ni v_{1} } ) }  )
\vert \{ j_e\}_{e\ni v_{1} } \,, I_{v_{1} }^{ ( \{ j_e\}_{ e\ni v_{1} } ) } \ra
\right)
\otimes
\left( \sum_{ I_{v_{2} }^{ ( \{ j_e\}_{ e\ni v_{2} } ) } }
C( I_{v_{2} }^{ ( \{ j_e\}_{ e\ni v_{2} } ) }  )
\vert \{ j_e\}_{e\ni v_{2} } \,, I_{v_{2} }^{ ( \{ j_e\}_{ e\ni v_{2} } ) } \ra
\right)
\cdots
\,. \label{eq:SeparableState-IntertwinerSuperposition}
\ee
In other words, the uncorrelated intertwiner superposition is simply a change of local intertwiner basis at each vertex.
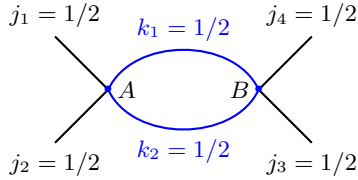
\begin{figure}[htb!]
\begin{tikzpicture}%[scale=0.7]

\coordinate (O1) at (3,0);
\coordinate (O2) at (5,0);

\draw[thick] (O1) node[right] {$A$}-- ++ (135:1) node[above] {$j_1=1/2$};
\draw[thick] (O1) -- ++ (225:1) node[below] {$j_2=1/2$};
\draw[thick] (O2) node[left] {$B$} -- ++ (-45:1) node[below] {$j_3=1/2$};
\draw[thick] (O2) -- ++ (45:1) node[above] {$j_4=1/2$};

\draw[blue,thick,in=115,out=65,rotate=0] (O1) to node[above,midway] {$k_1=1/2$} (O2) node[scale=0.7] {$\bullet$} to [out=245,in=-65] node[below] {$k_2=1/2$} (O1) node[scale=0.7] {$\bullet$};

\end{tikzpicture}
\caption{Labeling by internal indices, spin network state is written as $\sum_{I_A,I_B} C(I_A, I_B) | I_A, I_B \ra$. It has intertwiner-correlation via internal indices through nontrivial correlation coefficient $C(I_A, I_B)$, for example, $C(j_{12}=0, j_{34}=0)=C(j_{12}=1, j_{34}=1)=1/\sqrt{2}$ and $C(j_{12}=0, j_{34}=1)=C(j_{12}=1, j_{34}=0)=0$.
}
\label{fig:eg-candygraph-4valent}
\end{figure}

On the other hand, if we still keep spins fixed but consider non-factorizable coefficients $C(I_{v_1},I_{v_2})\neq C(I_{v_1})C(I_{v_2})$, then $v_1$ and $v_2$ are automatically entangled.
A simple example is shown in fig. \ref{fig:eg-candygraph-4valent}. Note that the cases only happen when $v_1$ and $v_2$ both have higher valency than three. This is encapsulated by the following definition:

\begin{defi}[Entangled states of intertwiner-correlation via internal indices]
A spin network state is said to carry intertwiner correlation between vertices $v_{1}$ and $v_{2}$ via internal indices if the correlation coefficient of internal indices is nontrivial (i.e., unfactorizable). That is, 
\be
| \psi \ra
=
\cdots
\left( \sum_{ I_{v_{1} }^{ ( \{ j_e\}_{ e\ni v_{1} } ) }, I_{v_{2} }^{ ( \{ j_e\}_{ e\ni v_{2} } ) } }
C( I_{v_{1} }^{ ( \{ j_e\}_{ e\ni v_{1} } ) }, I_{v_{2} }^{ ( \{ j_e\}_{ e\ni v_{2} } ) } )
\vert \{ j_e\}_{e\ni v_{1} } \,, I_{v_{1} }^{ ( \{ j_e\}_{ e\ni v_{1} } ) } \ra
\otimes
\vert \{ j_e\}_{e\ni v_{2} } \,, I_{v_{2} }^{ ( \{ j_e\}_{ e\ni v_{2} } ) } \ra
\right)
\cdots
\,. \label{eq:EntangledState-IntertwinerCorrelation}
\ee
The set of such states is denoted by $\cS_{C_{I}(\bGamma) }$.
\end{defi}

\smallskip

This type of intertwiner  correlation is only possible due to the non-trivial structure of the intertwiner space for vertices with valence strictly larger than 3. Nevertheless, it turns out possible to correlate two adjacent vertices, even 3-valent ones, and create entanglement by unfreezing the spins and allowing for spin superpositions on the edges linking the two vertices~:
%
%On the other hand, even vertices are all trivalent, the adjacent vertices are entangled if the common edge  has spin-superposition. Explicitly, the spin-superposed trivalent $v_{1},v_{2}$ state is written as:
\begin{defi}[Entangled states of bulk spin-superposition]
A spin network state has intertwiner-correlation between vertices $v_{1}$ and $v_{2}$ via bulk spin-superposition if there exists at least one edge $e \in \bGamma$ that links two distinct vertices $v_1$ and $v_2$ ($v_1 \neq v_2$), the associated spin $k_{e}$ has spin-superposition. Let $e' \in \{ e \}_{ e \ni v_{1} } \cap \{ \tl{e} \}_{ \tl{e} \ni v_{2} }$ be a common edge linked $v_1$ and $v_2$, and $C( k_{e'} )$ be the spin-superposition coefficient, then any state in the form of
\be
| \psi \ra
=
\cdots
\left( \sum_{ k_{e'} }
C( k_{e'} ) \,
\vert \{ j_e\}_{e\ni v_{1} }, I_{v_{1} }^{ ( \{ j_e\}_{ e\ni v_{1} } ) } \ra
\otimes
\vert \{ j_{ \tl{e} } \}_{e\ni v_{2} }, I_{v_{2} }^{ ( \{ j_{ \tl{e} } \}_{ \tl{e} \ni v_{2} } ) } \ra
\right)
\cdots
\label{eq:EntangledState-BulkSpinSuperposition}
\ee
is entangled. Then we denote the set of such states by $\cS_{S_{j}(\bGamma) }$.
\end{defi}

A simple example is shown in fig. \ref{fig:eg-treegraph-3valent}. Indeed, as long as two vertices have one common edge with spin-superposition, their intertwiners are entangled, since the spin is a common index for intertwiners at $v_1$ and $v_2$.
%
%Notice the condition $v_1 \neq v_2$. In fact, in the situation that $v_1=v_2$, there exists possibility that $e$ is a loop \cite{Charles:2016xwc} attached to that vertex. In this case, the spin-superposition of the loopy edge does not build correlation to other vertices, so loopy spin-superposition does not introduce entanglement.
%
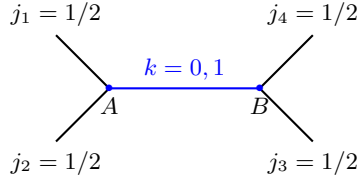
\begin{figure}[htb]
\begin{tikzpicture}%[scale=0.7]

\coordinate (O1) at (3,0);
\coordinate (O2) at (5,0);

\draw[thick] (O1) node[below] {$A$}-- ++ (135:1) node[above] {$j_1=1/2$};
\draw[thick] (O1) -- ++ (225:1) node[below] {$j_2=1/2$};
\draw[thick] (O2) node[below] {$B$} -- ++ (-45:1) node[below] {$j_3=1/2$};
\draw[thick] (O2) -- ++ (45:1) node[above] {$j_4=1/2$};

\draw[blue,thick] (O1) node[scale=0.7] {$\bullet$} -- node[above,midway] {$k=0,1$} (O2) node[scale=0.7] {$\bullet$};

\end{tikzpicture}
\caption{Spin network state is written as $\sum_{k} C(k) | j_1, j_2, k \ra \otimes | j_3,j_4,k \ra$. The intertwiners $| j_1, j_2, k \ra$ and $| j_3,j_4,k \ra$ are entangled if $k$ has superposition, e.g., $C(k=0)=C(k=1)=1/\sqrt{2}$.
}
\label{fig:eg-treegraph-3valent}
\end{figure}

\smallskip

Finally, the vertices $v_1$ and $v_2$ can be entangled if their have spin-correlation, i.e., there exists nontrivial correlation coefficient $C(j_1,j_2,\cdots)$ for spins where $j_1$ and $j_2$ are two spins attached to two different vertices.
\begin{defi}[Entangled states of spin-correlation]
A spin network state has intertwiner-correlation between vertices $v_{1}$ and $v_{2}$ via spin-correlation if there exists spin-correlation between two edges $e_1,e_2$ that $e_1 \ni v_1$, $e_2 \ni v_2$ and $v_1 \neq v_2$. Let $C( j_{e_1}, j_{e_2} )$ be the correlation coefficient, then any state in the form of
\be
| \psi \ra
=
\cdots
\left( \sum_{ j_{e_1}, j_{e_2} }
C( j_{e_1}, j_{e_2} ) \,
\vert \{ j_e\}_{e\ni v_{1} }, I_{v_{1} }^{ ( \{ j_e\}_{ e\ni v_{1} } ) } \ra
\otimes
\vert \{ j_{ \tl{e} } \}_{e\ni v_{2} }, I_{v_{2} }^{ ( \{ j_{ \tl{e} } \}_{ \tl{e} \ni v_{2} } ) } \ra
\right)
\cdots
\label{eq:EntangledState-BoundarySpinCorrelation}
\ee
is entangled. The set of such states is denoted by $\cS_{C_{j}(\Gamma) }$.
\end{defi}

We look at an example shown in fig. \ref{fig:eg-candygraph-3valent-BoundarySpinCorrelation}. Consider below two spin network states
\be
| \phi \ra
=
\left( \sum_{ j_1 }
C( j_1 ) \,
\vert j_1,k_1,k_2 \ra
\right)
\otimes
\left( \sum_{ j_2 }
C( j_2 ) \,
\vert j_2,k_1,k_2 \ra
\right)
\,, \qquad
| \psi \ra
=
\sum_{ j_1, j_2 }
C( j_1, j_2 ) \,
\vert j_1,k_1,k_2 \ra
\otimes
\vert j_2,k_1,k_2 \ra
\,. \label{eq:Examples-BoundarySpinSuperposition}
\ee
Both states have boundary spin-superposition. The distinction is that $| \psi \ra$ has spin-correlation while $| \phi \ra$ has not, thus $| \psi \ra$ is entangled (e.g., fig.\ref{fig:eg-candygraph-3valent-BoundarySpinCorrelation}) while $| \phi \ra$ is unentangled.
\begin{figure}[htb]
\begin{tikzpicture}%[scale=0.7]

\coordinate (O1) at (3,0);
\coordinate (O2) at (5,0);

\draw[thick] (O1) node[right] {$A$}-- ++ (180:1) node[above] {$j_1$};
\draw[thick] (O2) node[left] {$B$} -- ++ (0:1) node[above] {$j_2$};

\draw[blue,thick,in=115,out=65,rotate=0] (O1) to node[above,midway] {$k_1=1$} (O2) node[scale=0.7] {$\bullet$} to [out=245,in=-65] node[below] {$k_2=1$} (O1) node[scale=0.7] {$\bullet$};

\end{tikzpicture}
\caption{Spin network state is written as $\sum_{j_1,j_2} C(j_1, j_2) | j_1,k_1,k_2 \ra \otimes | j_2,k_1,k_2 \ra$ which has boundary spin correlation via unfactorizable correlation coefficient $C(j_1, j_2)$, for example, $C(j_1=1, j_2=1)=C(j_{1}=2, j_{2}=2)=1/\sqrt{2}$ and $C(j_{1}=1, j_{}=2)=C(j_{1}=2, j_{1}=1)=0$.
}
\label{fig:eg-candygraph-3valent-BoundarySpinCorrelation}
\end{figure}
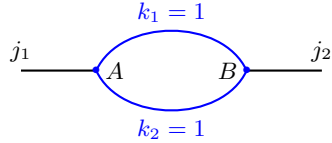

\medskip

Let us summarize the structures described above in the following statement: 

\begin{res}
\label{Prop:ProductStates-SpinNetworks}
In the spin network Hilbert space $\cH_{\Gamma}$ for a graph, possibly with a boundary, the set of fully separable states is defined as the set of states which do not carry intertwiner correlation, bulk spin superposition or spin correlations:
\be
\cS_{separable(\Gamma) }
=
\cH_{\Gamma} \setminus
( \cS_{C_{I}(\bGamma) } \cup \cS_{S_{j}(\bGamma) } \cup \cS_{C_{j}(\Gamma) } )
\,.
\ee
Separable states are thus states which might have boundary spin superpositions, as long as the boundary spins remain uncorrelated, and that for each set of boundary spins, the bulk spins and intertwiners are fixed.

In particular, spin network basis states are unentangled,
%Denote $\cS_{basis}$ the set of . 
%Any spin network basis state is unentangled, i.e.,
\be
\cS_{basis(\Gamma) } \subset \cS_{separable(\Gamma) }
\,.
\ee
\end{res}

It is interesting to characterize the separable states which are not spin network basis states. First, we put aside the possibility of uncorrelated intertwiner superpositions, since these are still spin network basis states up to a mere change of intertwiner basis locally at each vertex. We are then left with two possibilities:
\begin{itemize}
\item (uncorrelated) superpositions of boundary spins;
\item (uncorrelated) superpositions of spins on a self-loop, i.e. an edge linking a vertex to itself.
\end{itemize}

\begin{res}[Unentangled states on trivalent spin network] \label{Prop:ProductStates-TrivalentSpinNetworks}
For spin network states on trivalent graphs without any self-loop, if we restrict to study a subspace of spin network with fixed boundary spins, then the set of fully separable states is the set of spin network basis states (with fixed spins on every bulk edges).
\end{res}

In the present work,
%up to unfolding all vertices along 3-valent trees,
for the sake of simplicity, we focus on the entanglement carried by trivalent spin networks, in which the above result   \ref{Prop:ProductStates-TrivalentSpinNetworks} applies.
%
%Since no internal space living at vertex needs to be concerned.
%
In fact, we could easily extend the analysis to higher valent vertices since the  loop holonomy operator does not affect the internal space of intertwiner and can be understood as acting as on 3-valent vertices obtained by unfolding the vertices in terms of bouquet spins, as explained earlier. This would simply complicate notations.
%
%More general situation could be studied in same approach, since loop holonomy operator does not affect the internal space of intertwiner. Nevertheless, from now on we only consider the situation mentioned in Proposition \ref{Prop:ProductStates-TrivalentSpinNetworks}.

\medskip

Knowing what fully separable states are in LQG, we are able to adopt the geometric measure of entanglement \cite{PhysRevA.68.042307,Amico:2007ag} as a witness of multipartite entanglement. It quantifies entanglement of a pure state through the minimal distance of the state from the set of pure fully separable states
\be \label{eq:GeometricMeasureEntanglement}
S_{g}(\Psi)=-\ln \max_{\Phi} | \la \Phi | \Psi \ra |^2
\,,
\ee
where the maximum is on the set of fully separable states $\Phi$. Via variational method, it turns out that $\la \Phi | \Psi \ra$ is a real number when it reaches the extremal value \cite{PhysRevA.68.042307}. The maximal value of $| \la \Phi | \Psi \ra |^2$ is called \textit{entanglement eigenvalue}. In particular for bipartite system, the entanglement eigenvalue is the maximal Schmidt eigenvalue of reduced density matrix, and moreover, each Schmidt eigenvalue is an extremal value for $| \la \Phi | \Psi \ra |^2$.

Usually, the maximal projection is probably not easy to find. For the situation the Proposition \ref{Prop:ProductStates-TrivalentSpinNetworks} concerns, to obtain the maximum, what we need to do is to project the considered possibly entangled state onto the spin network basis states, then identify the most probable basis state. The projection determines the entanglement eigenvalue, so determines the geometric measure of entanglement.

\smallskip

So far we have discussed the classification of unentangled spin networks and entangled spin networks, and definition of geometrical measure of entanglement. The next question is how the entanglement evolves under dynamics, for instance generated by a loop holonomy operator.

%%%%%%

%%%
\subsection{The leading order evolution of geometric entanglement}
%\subsection{The 1st- and 2nd-order derivative of geometric measure of entanglement with respect to time} \label{Section:GME-LQG}
 \label{Section:GME-LQG}
%%%
In this part, we investigate how geometrical measure of entanglement evolves under dynamics driven by a given hamiltonian. Assume that $| \Psi_0 \ra$ is the initial state and that the evolution is generated by the exponential map $e^{-\ri \wh{H} t}$ with respect to a hermitian operator $\wh{H}$. We will show that at least the 1st- and 2nd- order derivative of geometric measure of entanglement with respect to time parameter $t$ can be expressed in a simple fashion in terms of the hermitian operator.

\smallskip

The state $| \Psi_0 \ra$ evolves as $ | \Psi (t) \ra=e^{-\ri \wh{H} t} | \Psi_0 \ra$ with $ | \Psi (t=0) \ra=| \Psi_0 \ra$. Expanding it up to the 2nd order around the initial time reads
\be
| \Psi (t) \ra
=
e^{-\ri \wh{H} t} | \Psi_0 \ra
=
| \Psi_0 \ra - \ri t \wh{H} | \Psi_0 \ra - \f{t^2}{2} \wh{H}^2 | \Psi_0 \ra + \cO(t^3)
\,.
\ee
The definition of the geometric measure of entanglement involves the fully separable state $| \Phi(t) \ra$ corresponding to $| \Psi(t) \ra$, which maximizes the probability $\vert \la \Phi(t) | \Psi (t) \ra |^2$ at every instant $t$, i.e., $\vert \la \Phi(t) | \Psi (t) \ra |^2=\max_{\Phi'} | \la \Phi' | \Psi(t) \ra |^2$ for any $t$ where $| \Phi' \ra$ runs over the set of fully separable states. As we are not dealing with a rigged Hilbert space which may allow diverging distribution, the scalar product reamins bounded, $0 \leq \vert \la \Phi' | \Psi (t) \ra |^2 \leq 1$.

\smallskip

However, notice that the definition of the optimal separable state $\{ | \Phi(t) \ra \}$ might be ambiguous. We provide two simple, hopefully helpful, examples \ref{eg:TDBell} and \ref{eg:BHtoy} in appendix.
Nevertheless, the geometric entanglement value is always continuous. This follows from the continuity of dynamics. For instance, thinking of a bipartite system, the dynamics of the Schmidt eigenvalues $\lambda(t)$ of the reduced density matrix can be described by a master equation, which relates the 1st-order derivative of the Schmidt eigenvalues with respect to $t$, to the commutator $[ \wh{H}, \rho]$. Since $\wh{H}$ and $\rho$ are assumed to be well-behaved operators, the $\rd \lambda(t) / \rd t$ is also well-behaved.
%Hence all $\lambda(t)$ continuous and differentiable.
%
One possible  concern is  that entanglement eigenvalue $\lambda_{\max}$ might be discontinuous. Indeed starting with the maximal eigenvalue $\lambda_1(t)$ at time $t_1$,  it is possible that it is not anymore the  maximal eigenvalue at a later time $t_2$. Namely, such a discontinuity happens when the order of the eigenvalues switches, i.e. if $\lambda_1(t) > \lambda_2(t)$ when $t \leq t_2$ while $\lambda_1(t) < \lambda_2(t)$ when $t > t_2$: the maximal eigenvalue would jump from the branch $\lambda_1$ to the  branch $\lambda_2$. However, even in that case, the entanglement eigenvalue is still continuous: since $\lambda_1(t)$ and $\lambda_2(t)$ are both continuous, there exists a transition time $t$ such that $\lambda_1(t)=\lambda_2(t)$. The example \ref{eg:TDBell} is an example for the case where the entanglement eigenvalue switches the branches at $t=\pi/4$ while remaining continuous.

%Beyond the continuity, we can looking into the differentiability of the entanglement, in order to analyze its leading order evolution. The switch between the branches of eigenvalues causes the discontinuity of $\{ | \Phi(t) \ra \}_t$
%% $\lim_{t \to s} || | \Phi(t) \ra - | \Phi(s) \ra ||$.
%%But again back to the well-behave of dynamics,
%Nevertheless, the entanglement eigenvalue remains both left and right differentiable with respect to $t$ at the transition time, even though the left derivative might differ from the right derivative.

%%%%%%%%%%%%%%%

Let us look more closely at the evolution of the entanglement close to the initial time (keeping in mind that one can arbitrarily swift the initial choice). Let us expand the scalar product $| \la \Phi(t) | \Psi (t) \ra |^2$ in a Taylor series in $t$:
\beq
| \la \Phi(t) | \Psi (t) \ra |^2&=&| \la \Phi(t) | e^{-it\hH}|\Psi_{0} \ra |^2
%\\&=&
%=
%\Big{|}
% \la \Phi(t) | \Psi_{0} \ra -it \la \Phi(t) |\hH| \Psi_{0}) \ra -\f{t^{2}}2 \la \Phi(t) | \hH^{2}\Psi_{0} \ra 
%\Big{|}^{2}
\\
&=&
\big{|} \la \Phi(t) | \Psi_{0} \ra\big{|}^{2}
+it\Big{[}
\la \Phi(t) | \Psi_{0} \ra\la \Psi_{0}|\hH|\Phi(t)  \ra-\la \Phi(t) |\hH| \Psi_{0} \ra\la \Psi_{0}|\Phi(t)  \ra
\Big{]}
\nn\\
&&
+\f{t^{2}}2\Big{[}
2\la \Phi(t)|\hH | \Psi_{0} \ra\la \Psi_{0}|\hH|\Phi(t)  \ra
-\la \Phi(t) |\hH^{2}| \Psi_{0} \ra\la \Psi_{0}|\Phi(t)  \ra
-\la \Phi(t) | \Psi_{0} \ra\la \Psi_{0}|\hH^{2}|\Phi(t)  \ra
\Big{]}
+
{\cal O}(t^{3})
\,.\nn
\eeq
This is not exactly a full Taylor expansion since $\Phi(t)$ still depends on time. The first term $\big{|} \la \Phi(t) | \Psi_{0} \ra\big{|}^{2}$  actually reaches its maximal value at $t=0$, by definition of the state $\Phi(t)$, and thus has vanishing first derivative:
\be
\big{|} \la \Phi(t) | \Psi_{0} \ra\big{|}^{2}=\big{|} \la \Phi_{0} | \Psi_{0} \ra\big{|}^{2}+\cO(t^{2})
\,.
\ee
But there is a priori  no obvious further simplification.

\medskip

Let us make a first assumption:
\begin{itemize}
\item The initial state is separable, thus $\Phi_{0}=\Psi_{0}$.
\end{itemize}
We can then prove that the first derivative of $| \la \Phi(t) | \Psi (t) \ra |^2$ vanishes and that the leading order of the geometric entanglement $\ln\,| \la \Phi(t) | \Psi (t) \ra |^2$ is in $\cO(t^{2})$. Let us assume that the separable projection is smooth in a neighbourhood of the initial time and expand it in a Taylor series up to second order for $t>0$:
\be
\Phi(t)=\Psi_{0}+t\Phi^{(1)}+t^{2}\Phi^{(2)}+\dots
\ee
The normalization condition on that state reads:
\be
1=\la \Phi(t)|\Phi(t) \ra
=
1
+t\underbrace{\Big{[}
\la \Psi_{0}|\Phi^{(1)} \ra+\la \Phi^{(1)}|\Psi_{0} \ra
\Big{]}}_{=0}
+t^{2}\underbrace{\Big{[}
\la \Phi^{(1)}|\Phi^{(1)} \ra
+\la \Psi_{0}|\Phi^{(2)} \ra+\la \Phi^{(2)}|\Psi_{0} \ra
\Big{]}}_{=0}
+\dots
\ee

This allows us to expand the terms in the scalar product $| \la \Phi(t) | \Psi (t) \ra |^2$:
\be
\big{|} \la \Phi(t) | \Psi_{0} \ra\big{|}^{2}
=
1
+t\underbrace{\Big{[}
\la \Psi_{0}|\Phi^{(1)} \ra+\la \Phi^{(1)}|\Psi_{0} \ra
\Big{]}}_{=0}
+t^{2}
\underbrace{\Big{[}
\la \Psi_{0}|\Phi^{(1)} \ra\la \Phi^{(1)}|\Psi_{0} \ra
+\la \Psi_{0}|\Phi^{(2)} \ra+\la \Phi^{(2)}|\Psi_{0} \ra
\Big{]}}_{=\la \Psi_{0}|\Phi^{(1)} \ra\la \Phi^{(1)}|\Psi_{0} \ra-\la \Phi^{(1)}|\Phi^{(1)} \ra}
+\cO(t^{3})
\,.
\ee
Similarly, the second term, $\Big{[}
\la \Phi(t) | \Psi_{0} \ra\la \Psi_{0}|\hH|\Phi(t)  \ra-\la \Phi(t) |\hH| \Psi_{0} \ra\la \Psi_{0}|\Phi(t)  \ra
\Big{]}$, vanishes at $t=0$ and its first order depends on $\Psi_{0}$ and $\Phi^{(1)}$.
%
%\be
%\Big{[}
%\la \Phi(t) | \Psi_{0} \ra\la \Psi_{0}|\hH|\Phi(t)  \ra-\la \Phi(t) |\hH| \Psi_{0} \ra\la \Psi_{0}|\Phi(t)  \ra
%\Big{]}
%=
%t\Big{[}
%(\la \Phi^{(1)} | \Psi_{0} \ra-\la \Psi_{0}|\Phi^{(1)}  \ra)\la \Psi_{0}|\hH|\Psi_{0}  \ra
%+
%(\la \Psi_{0}|\hH|\Phi^{(1)}  \ra-\la \Phi^{(1)} |\hH| \Psi_{0} \ra)
%\Big{]}
%+\cO(t^{2})
%\ee
As a consequence, the scalar product has a vanishing first derivative and is  trivial up to second order, $| \la \Phi(t) | \Psi (t) \ra |^2=1+\cO(t^{2})$, thus
\beq
%| \la \Phi(t) | \Psi (t) \ra |^2=1+\cO(t^{2})\,,
%\qquad\textrm{thus}\quad
S_{g}=-\ln| \la \Phi(t) | \Psi (t) \ra |^2
&=&
%\cO(t^{2})
t^{2}\Big{[}
\la \Psi_{0}|\hH^{2} | \Psi_{0}\ra
-\la \Psi_{0}|\hH|\Psi_{0}\ra^{2}
+
\la \Phi^{(1)}|\Phi^{(1)} \ra - \la \Psi_{0}|\Phi^{(1)} \ra\la \Phi^{(1)}|\Psi_{0} \ra
\\
&&+i(\la \Psi_{0}|\Phi^{(1)}  \ra - \la \Phi^{(1)} | \Psi_{0} \ra)\la \Psi_{0}|\hH|\Psi_{0}  \ra
+
i(\la \Phi^{(1)} |\hH| \Psi_{0} \ra - \la \Psi_{0}|\hH|\Phi^{(1)}  \ra)
\Big{]}
+\cO(t^{3})
\,,
\nn
\eeq
with the $t^{2}$-coefficient depending explicitly on the linear deviation $\Phi^{(1)}$ of the separable projection.

\medskip

Let us then make a further assumption, which is tailor-suited to the present case of study and allows us to determine exactly $\Phi^{(1)}$ :
\begin{itemize}
\item the set of separable states is discrete, i.e. separable states are isolated points in the Hilbert space.
\end{itemize}
This happens for trivalent spin networks, since the spin network basis states are entirely determined by the spin labels on the edges and there is local degrees of freedom at the vertices once the spins are fixed. This leads to a countable set of isolated separable states. The deep consequence is that $\Phi(t)$ is a step function, jumping from separable state to separable state. Let us keep in mind that, although $\Phi(t)$ is discontinuous, the scalar product $| \la \Phi(t) | \Psi (t) \ra |^2$ and resulting entanglement remain continuous functions of the time $t$. Therefore, $\Phi(t)$ is constant in a neighbourhood of the initial time, it is equal to the initial state $\Psi_{0}$ and its first derivative $\Phi^{(1)}$ vanishes. The scalar product,
\be
| \la \Phi(t) | \Psi (t) \ra |^2\underset{t\sim 0}{=}| \la \Psi_{0} | \Psi (t) \ra |^2\,,
\ee
reduces to the projection of the evolving state $\Psi(t)$ onto the initial set and geometric entanglement's leading order is simply given by the dispersion of the Hamiltonian operator: 
\be
\label{eq:Prop-Limit-2ndDerivative}
S_{g}(t)=t^{2}\Big{[}
\la \Psi_{0}|\hH^{2} | \Psi_{0}\ra
-\la \Psi_{0}|\hH|\Psi_{0}\ra^{2}
\Big{]}+\cO(t^{3})
\ee
This will simplify all the entanglement calculations, as we will see in explicit examples in the next sections. It will be validated to the comparison to the linear entanglement entropy (for bipartitions), which will give exactly the same leading order in $t^{2}$. One should nevertheless remember that, if we consider spin networks with four-valent or higher-valent vertices, the leading order will remain in $t^{2}$ but the precise factor will probably acquire corrections to the $\hH$-dispersion depending on the precise dynamics of the separable projection and its linear deviation $\Phi^{(1)}$.

\subsection{Entanglement excitation and closure defect distribution}
%%%

%It is straightforward to apply the conclusion in last subsection to the cases of trivalent spin networks. In particular, we can apply to the question that how much geometric measure of entanglement excites from a spin network basis state, which is unentangled. Let us begin by taking $| \Phi(0) \ra=| \Psi_0 \ra$ into equation (\ref{eq:Prop-Limit-2ndDerivative}), which leads to the following proposition.

Let us apply the results from the previous section to the action of the loop holonomy operator. We consider an unentangled initial state $| \Psi_0 \ra$, given by a spin network basis state. Its separable projection is  itself,  $| \Phi(t=0) \ra=| \Psi_0 \ra$. We would like to know the entanglement excitation created by the  loop holonomy operator. Applying the formula (\ref{eq:Prop-Limit-2ndDerivative}) derived above to the holonomy operator leads to the following result:

\begin{res} \label{Prop:GME-LoopHolonomy-Trivalent-2ndDerivative}
Let $| \Psi_0 \ra$ be any trivalent spin network basis state, and $\wh{\chi_{\ell} } \, \act_{W}$ be a loop holonomy operator where $W$ is loop through more than one vertex (i.e., $W$ is not a self-loop).
The $| \Psi(t) \ra$ is the state driven by $\wh{\chi_{\ell} } \, \act_{W}$ from initial state $| \Psi_0 \ra$.
 Then the 1st-order and 2nd-order derivative of geometric measure of entanglement at $t=0$ are given by
\beq
\f{ \rd S_{g} [\Psi(t)] }{ \rd t} \Big\vert_{t=0}
=0
\,, \qquad
\f12\f{ \rd^2 S_{g} [\Psi(t)] }{ \rd t^2} \Big\vert_{t=0}
=
\la \Psi_0 | \left( \wh{\chi_{\ell} } \, \act_{W} \right)^2 | \Psi_0 \ra
-\la \Psi_0 | \wh{\chi_{\ell} } \, \act_{W} | \Psi_0 \ra^2
\,.
\label{eq:GME-LoopHolonomy-TrivalentBasisState-1st,2ndDerivative}
\eeq
\end{res}

Let us apply this to a spin network made of a single loop with boundary edge insertions, as drawn in fig.\ref{fig:LoopySpinNetwork}. This case illustrates an interesting relation between the leading order entanglement evolution, as given by equation (\ref{eq:GME-LoopHolonomy-TrivalentBasisState-1st,2ndDerivative}), and the closure defect, defined as the recoupled spin of the boundary spins . This relation is realized by relating the dispersion of loop holonomy operator $\wh{\chi_{\ell} } \, \act_{W}$, which gives the 2nd-order derivative of the entanglement, to the probability distribution of the closure defect.
\begin{figure}[hbt!]
\centering
\begin{tikzpicture}[scale=0.7]

\coordinate (A) at (-7,0);

\draw [domain=60:360,thick] plot ({-7+1.5 * cos(\x)}, {1.5 * sin(\x)});
\draw [domain=0:60,blue,thick] plot ({-7+1.5 * cos(\x)}, {1.5 * sin(\x)});
\draw[thick] (A) ++(-0.6,0);% node {$G$};
\draw[thick] (A) ++(0:1.5) --++ (0:1);% ++(0:0.35) node {$j_1$};
\draw[thick] (A) ++(60:1.5) --++ (60:1) ++(60:0.35) node {$j_1$};
\draw[thick] (A) ++(120:1.5) --++ (120:1) ++(120:0.35) node {$j_2$};
\draw[thick] (A) ++(180:1.5) --++ (180:1) ++(180:0.35) node {$j_3$};
\draw[thick] (A) ++(240:1.5) --++ (240:1);% ++(240:0.35) node {$j_5$};
\draw[thick] (A) ++(300:1.5) --++ (300:1);% ++(300:0.35) node {$j_6$};

\draw[blue] (A) ++(30:1.8) node {$k$};

\draw [thick, loosely dotted,domain=195:230] plot ({-7+2.6 * cos(\x)}, {2.6 * sin(\x)});

\coordinate (O) at (3.5,0);

\draw[->,>=stealth,very thick] (-3.25,0) -- node[above] {gauge-fixing} (-0.75,0);

\draw[thick,red] (O) -- ++ (315:1.5) node[very near end,above=2] {$J$} coordinate (B) node[blue,scale=0.7] {$\bullet$};
\draw[blue,thick,in=-90,out=0,scale=4.5,rotate=0] (B)  to[loop] node[near start,sloped] {$>$} node[near end,left=2] {$k$} (B) ++(315:0.35) node {$G$};

\draw[thick] (O) -- ++ (0:1.5);% ++(0:0.35) node {$j_1$};
\draw[thick] (O) -- ++ (45:1.5) ++(45:0.35) node {$j_1$};
\draw[thick] (O) -- ++ (90:1.5) ++(90:0.35) node {$j_2$};
\draw[thick] (O) -- ++ (135:1.5) ++(135:0.35) node {$j_3$};
\draw[thick] (O) -- ++ (180:1.5)  ++ (180:1.2) node {$\displaystyle{\sum_{J} }$}; % ++(240:0.35) node {$j_1$};
\draw[thick] (O) -- ++ (225:1.5); %++(300:0.35) node {$j_1$};

\draw (O) node[scale=0.7] {$\bullet$};

\draw [thick, loosely dotted,domain=160:200] plot ({3.5+1.8 * cos(\x)}, {1.8 * sin(\x)});

\end{tikzpicture}
%
%\hspace{2mm}
%
\caption{
The left hand side is a trivalent spin network where all bulk edges lie along a circle. We can choose a maximal tree such that all of bulk edges except the blue one form a path where all the holonomies can be set into identity element of $\SU(2)$. By means of contracting the maximal tree, the blue edge with spin-$k$ becomes a loopy edge with same spin-$k$ and holonomy $G$. The spin-$J$ is recoupled from boundary spins
}
\label{fig:LoopySpinNetwork}
\end{figure}
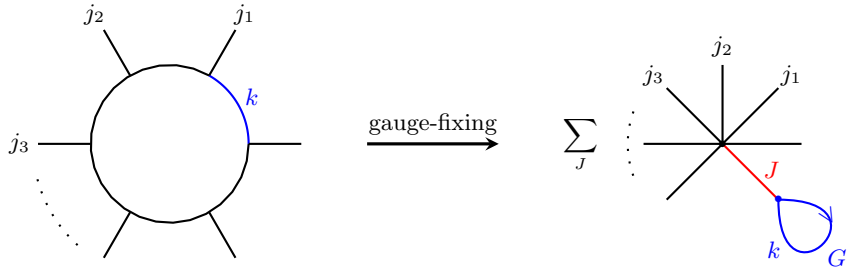
%
%expresses as the , which is obvious gauge invariant form. The gauge invariance allows us to use the technique in \cite{Chen:2021vrc} to evaluate the dispersion.

\smallskip

%Consider a trivalent spin network in Fig. \ref{fig:LoopySpinNetwork} where all bulk edges lie along a circle. The number of boundary spins equals to the number of vertices, say $V$, and here the number $E$ of bulk edges equals to $V$ as well. This is a one-loop spin network in the sense of topology $L=E-V+1$.

We use the techniques introduced in \cite{Chen:2021vrc}. The wave-function of spin network (on left hand side of fig.\ref{fig:LoopySpinNetwork}) can be thought as a boundary map mapping bulk holonomies (on the edges along the loop) to vectors in the boundary Hilbert space $\bigotimes_{e\in\pp\Gamma} \cV_{j_e}$. 
Using the gauge-invariance property of spin networks and proceeding to a gauge-fixing of all but one edges around the loop (see \cite{Chen:2021vrc} and refs therein), we can reduce the the one-loop spin network to its gauge-fixed counterpart, drawn on the right hand side of fig.\ref{fig:LoopySpinNetwork}). This allows to write the spin network functional as a function of a single group element $G$, representing the holonomy around the loop:
%
%Thus the spin network on left hand side can be presented in boundary Hilbert space. Moreover, we can implement gauge-fixing on the spin network (cf. \cite{Chen:2021vrc} and refs therein), which amounts to transforming the spin network on left side to the simpler spin network on right side, such that the boundary map is written as
\be
| \psi_{\pp\Gamma} (G) \ra
=
\sum_{J} e^{\ri \vphi_{k}(J) }\sqrt{ p_{k} (J) } \sum_{a,b=-k}^{k} \sqrt{2k+1} (-1)^{k-a} D^{k}_{ab}(G)
\begin{pmatrix}
J & k & k \\
M & b & -a
\end{pmatrix} | J,M \ra
\qquad\in \bigotimes_{e\in\pp\Gamma} \cV_{j_e}\,.
\ee
%up to some boundary holonomies and phase factor, which do not matter to the dispersion to be evaluated.
The spin-$J$ is  the recoupled spin of all the boundary spins and is called the closure defect. The  probability amplitude  $\sqrt{p_{k}(J)}\,e^{\ri \vphi_{k}(J) }$ is a function of the spin $k$ living on the gauge-fixed loop (the blue edge on fig.\ref{fig:LoopySpinNetwork}). It fully characterizes the gauge-fixed spin network state and reflects the spins dressing the edges around the loop before gauge-fixing. In particular, if one were to choose another edge as the loopy edge, the resulting probability amplitude would be a priori different, though gauge equivalent. The modulus square of the probability amplitude gives the probability distribution $p_{k}(J)$ for the closure defect, which satisfies the normalization $\sum_{J} p_{k}(J)=1$.
%
%Here we adopt subscript $k$ to denote the situation that we have chosen the blue edge (with spin-$k$) as the loopy edge. In fact, the probability distribution $p_{k}(J)$ relies on the path for gauge-fixing, i.e., if one chooses another edge as the loopy edge, the distribution could be different. The reason is due to the fact that they have different boundary holonomies.

\smallskip

Now let us look at the  loop holonomy operator. It acts on the spin network state. But since the operator and state are both gauge invariant, one can legitimately look at the holonomy operator acting on gauge-fixed spin network states.
%
%The action will change spin network state via changing the spin network wave-function. Since the wave-function defines the boundary map, and also the wave-function of left side spin network is same as the wave-function of right side simpler spin network, we can choose to work on the simpler spin network instead.
%
One should be careful:  the entanglement structures are totally different on the two graphs, because the gauge-fixing procedure involves the procedure of contracting vertices, thus changes the number of vertices on graph.
Nevertheless, if one is interested in the dispersion of the loop holonomy operator, then the gauge-fixing does not change anything.
So let us apply the  loop holonomy operator on the self-loop of the gauge-fixed state, depicted on the right hand side of fig.\ref{fig:LoopySpinNetwork}. Its action involves both spins $k$ and $J$. Applying the general formula  (\ref{eq:Amplitudes-6jsymbols}) to this simple setting yields the following wave-function:
\beq
&&
\wh{\chi_{\ell} } \, \act_{W} | \psi_{\pp\Gamma} (G) \ra
=
\chi_{\ell} (G) \, | \psi_{\pp\Gamma} (G) \ra
\\
&=&
\sum_{J} e^{\ri \vphi_{k}(J) }\sqrt{ p_{k}(J) } \sum_{K=|k-\ell|}^{k+\ell} (2K+1) (-1)^{J+\ell + k+K }
\begin{Bmatrix}
J & K & K \\
\ell & k & k
\end{Bmatrix}
\sum_{a,b=-K}^{K}
\sqrt{2k+1}
(-1)^{K-a} D^{K}_{ab}(G)
\begin{pmatrix}
J & K & K \\
M & b & -a
\end{pmatrix} | J,M \ra
\,.
\nn
\eeq
%The holonomy operator can be understood as creating spin-$K$ states range from $| k- \ell |$ to $k+\ell$.
%The loop holonomy operator does not see the boundary holonomies and phase. Furthermore, these boundary holonomies and phase are eliminated when evaluates expectation $\la \, ( \wh{\chi_{\ell} } \, \act_{W} )^2 \, \ra$ and $\la \, \wh{\chi_{\ell} } \, \act_{W} \, \ra$, that is why we omit them safely.
We compute the mean value and deviation of the operator on the quantum state:
\beq
\la \, ( \wh{\chi_{\ell} } \, \act_{W} )^2 \, \ra
&=&
\int_{ \SU(2)^E } \prod_{e \in E} \rd g_{e} \, \la \psi_{\pp\Gamma} (G) | ( \wh{\chi_{\ell} } \, \act_{W} )^2 | \psi_{\pp\Gamma} (G) \ra
%\Big|\Big| \, \wh{\chi_{\ell} } \, \act_{W} | \psi_{\pp\Gamma} (G) \ra \, \Big|\Big|^2
=
\sum_{J} p_{k}(J) \sum_{s=0(1)}^{2\ell} (-1)^{J+s+2k}
\begin{Bmatrix}
J & k & k \\
s & k & k
\end{Bmatrix}
(2k+1)
\,, \\
\la \, \wh{\chi_{\ell} } \, \act_{W} \, \ra
&=&
\int_{ \SU(2)^E } \prod_{e \in E} \rd g_{e} \, \la \psi_{\pp\Gamma} (G) | \wh{\chi_{\ell} } \, \act_{W} | \psi_{\pp\Gamma} (G) \ra
=
\sum_{J} p_{k}(J) (-1)^{J+\ell+2k}
\begin{Bmatrix}
J & k & k \\
\ell & k & k
\end{Bmatrix}
(2k+1)
\,.
\eeq
The expectation $\la \, \wh{\chi_{\ell} } \, \act_{W} \, \ra$ automatically vanishes when $\ell \in \N+\f12$, due to the triangle condition on $6j$-symbol.
Here we have employed the recoupling formula $\wh{\chi_{\ell} } \, \act_{W} \circ \wh{\chi_{\ell} } \, \act_{W}=\sum_{s=0}^{2\ell} \wh{\chi_{s} } \, \act_{W}$ to compute the expectation $\la \, ( \wh{\chi_{\ell} } \, \act_{W} )^2 \, \ra$.
This gives the 2nd-order derivative of the geometric entanglement at the initial time:
%Hence the 2nd-order derivative of geometric measure of entanglement excited on one-loop trivalent spin network basis state is
\be
\f12\f{ \rd^2 S_{g} }{ \rd t^2} \Big\vert_{t=0}
=
\sum_{J} p_{k}(J) \sum_{s=0(1)}^{2\ell} (-1)^{J+s+2k}
\begin{Bmatrix}
J & k & k \\
s & k & k
\end{Bmatrix}
(2k+1)
-
\left(
\sum_{J} p_{k}(J) (-1)^{J+\ell+2k}
\begin{Bmatrix}
J & k & k \\
\ell & k & k
\end{Bmatrix}
(2k+1)
\right)^2
\,. \label{eq:1LoopTrivalentSpinNetwork-2ndDerivative}
\ee
This gives the excitation of entanglement created by the loop holonomy operator.
We should emphasize that eventhough the probability the $p_{k}(J)$ might depend on the choice of gauge-fixing (through the choice of the loopy edge), these averages are gauge invariant and do not depend on the gauge-fixing.

\medskip

Recalling the triangle condition on $6j$-symbols, there are two points observed from above expression: (i) if $\ell \in \N + \f12$, the second term vanishes. (ii) if $s>2k$, then the contribution from $\begin{Bmatrix} J & k & k \\ s & k & k \end{Bmatrix}$ vanishes. So there is a critical value $\ell_{c}=2k+\f12$ such that once  $\ell$ grows  beyond this critical value $\ell_{c}$, the 2nd-order derivative is a constant with respect to $\ell$.
%
%The critical value is easy to find: we takes $\ell=2k+\f12$ such that the second term in equation (\ref{eq:1LoopTrivalentSpinNetwork-2ndDerivative}) does not vary as the $\ell$ increases further.
%
The plateau value is easily computed using a standard identity on the $\{6j\}$-symbols\footnotemark{} and is a simple averaging of the probability distribution of the closure defect:
%Meanwhile, only the $s \in [0, 2k]$ make contribution to the first term in equation (\ref{eq:1LoopTrivalentSpinNetwork-2ndDerivative}). Therefore, we could relate the entanglement excitation to the closure defect distribution by
%
\footnotetext{We employ $\sum_{s=0(1)}^{2k} (-1)^{J+s+2k}
\begin{Bmatrix} J & k & k \\ s & k & k \end{Bmatrix}=\f{1}{2J+1}$.}
\be
\ell \geq 2k + \f12\,: \qquad
\f12\f{ \rd^2 S_{g} }{ \rd t^2} \Big\vert_{t=0}
=
\sum_{J} \f{ p_{k}(J) }{ 2J+1 }(2k+1)
\,.
\label{eq:2ndDGME-ClosureDefect}
\ee
This quantifies the amount of multibody entanglement created by the action of loop holonomy operator when acting on a pure spin network basis state. The holonomy operator entangles all the vertices around the loop with an entanglement growing in $t^{2}$ and its acceleration is directly related to the distribution of the closure defect -or, in other words, the recoupled boundary spin.

\section{Candy graph: bipartite entanglement} \label{section:CandyGraph}
 %%%%%%%%%%

In this section we look at the example of entanglement excitation on candy graph as fig.\ref{fig:candygraph}. This is a graph with a single loop and a pair of boundary spin insertions. The very simple structure of the graph allows us to study in full details the entanglement between the two vertices of the graph.

%%%
\subsection{Entanglement entropy excitation on candy graph with truncated dynamics}
%%%

%We choose a trivalent spin network basis state as the initial state.
%By studying infinitesimal evolution with truncation up to 1st-order \footnote{One could truncate the dynamics up to 2nd-order, and It turns out that the leading order of entanglement entropy equals to the 1st-order truncation, up to 2nd order time derivative.}, we can compute the entanglement entropy excited by the holonomy operator, and show the entanglement excitation matching with the geometric measure of entanglement up to 2nd-order. %Moreover, we will interpret the entanglement excitation as the density of excited states by holonomy operator.

We consider the holonomy operator acting on the loop of the candy graph. We compute explicitly the bipartite entanglement between the two vertices, defined as the  entropy of the reduced density matrix after tracing over one of the two vertices.
Actually, we compute both the von Neumann entropy and the linear entropy, but we prefer to use the linear entropy as measure of entanglement due to the  non-differentiability of the von Neumann entropy at initial time for an initial separable state.
Then we show that this measure of bipartite entanglement fits exactly with the geometric entanglement formula derived in the previous section up to 2nd order, thereby providing a relevant consistency check of that previous analysis.
\begin{figure}[htb]
\begin{tikzpicture}%[scale=0.7]

\coordinate (O1) at (3,0);
\coordinate (O2) at (5,0);

\draw[thick] (O1) node[right] {$A$}-- ++ (180:1) node[above] {$j_1$};
\draw[thick] (O2) node[left] {$B$} -- ++ (1,0) node[above,midway] {$j_{2}$};

\draw[blue,thick,in=115,out=65,rotate=0] (O1) to node[above,midway] {$k_1$} (O2) node[scale=0.7] {$\bullet$} to [out=245,in=-65] node[below] {$k_2$} (O1) node[scale=0.7] {$\bullet$};

\draw [domain=0:360,dashed] plot ({4+0.5 * cos(\x)}, {0.35 * sin(\x)});

\draw[->,>=stealth,very thick] (6.5,0) -- (7.5,0);

\coordinate (A1) at (10,0);
\coordinate (A2) at (12,0);

\draw[thick] (A1) node[right] {$A$}-- ++ (180:1) node[above] {$j_1$} ++ (180:0.75) node {$\displaystyle{\sum_{K_1,K_2} }$};
\draw[thick] (A2) node[left] {$B$} -- ++ (1,0) node[above,midway] {$j_{2}$};

\draw[blue,thick,in=115,out=65,rotate=0] (A1) to node[above,midway] {$K_1$} (A2) node[scale=0.7] {$\bullet$} to [out=245,in=-65] node[below] {$K_2$} (A1) node[scale=0.7] {$\bullet$};

\end{tikzpicture}
\caption{Loop holonomy operator acts on candy graph spin network, which leads to the spin-superposition over bulk spins.
}
\label{fig:candygraph}
\end{figure}
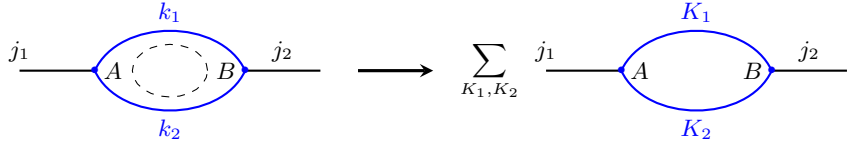

As illustrated on fig.{fig:candygraph}, the Hilbert space of spin networks on the  candy graph for fixed boundary spin s$j_{1}$ and $j_{2}$ is the tensor poroduct of the two intertwiner spaces sitting at the graph vertices:
\be
\cH_{A}
=
\bigoplus_{K_1,K_2}
\textrm{Inv}_{\SU(2)} \Big( \cV_{j_1} \otimes \cV_{K_1} \otimes \cV_{K_2} \Big)
\,, \qquad
\cH_{B}
=
\bigoplus_{K_1,K_2}
\textrm{Inv}_{\SU(2)} \Big( \cV_{j_2} \otimes \cV_{K_1} \otimes \cV_{K_2} \Big)
\,.
\ee
Spin network states can involve superpositions of the bulk spins $K_1$ and $K_2$ (while keeping the boundary spins $j_1$ and $j_2$ fixed). Such bulk spin superposition naturally induces a  superposition of intertwiners. If this superposition carries correlations between the two vertices, this will be reflected in the entanglement between the two vertices.

\smallskip

Starting with an initial spin network basis state $| \Psi_{can,\{j_i,k_i\}} \ra=| j_i,k_1,k_2 \ra_{A} \otimes | j_2,k_1,k_2 \ra_{B}$, we consider the evolution generated by the loop holonomy operator $\wh{\chi_{\ell} }$,
\be
\wh{\chi_{\ell} }
:
\textrm{Inv}_{\SU(2)} \Big( \cV_{j_1} \otimes \cV_{k_1} \otimes \cV_{k_2} \Big)
\otimes
\textrm{Inv}_{\SU(2)} \Big( \cV_{j_2} \otimes \cV_{k_1} \otimes \cV_{k_2} \Big)
\to
\bigoplus_{K_1,K_2}
\textrm{Inv}_{\SU(2)} \Big( \cV_{j_1} \otimes \cV_{K_1} \otimes \cV_{K_2} \Big)
\otimes
\textrm{Inv}_{\SU(2)} \Big( \cV_{j_2} \otimes \cV_{K_1} \otimes \cV_{K_2} \Big)
\,,
\nn
\ee
 For  infinitesimal time $t \to 0$, the unitarity evolution operator is $e^{ - \ri t \, \wh{\chi_{\ell} } } = \id - \ri  t \, \wh{\chi_{\ell} } + O(t^2)$, which acts as
 \beq
&&
e^{ - \ri t \, \widehat{\chi_{\ell} } } | \Psi_{can,\{j_i,k_i\}} \ra
=
e^{ - \ri t \, \widehat{\chi_{\ell} } } | j_1,k_1,k_2 \ra_{A} \otimes | j_2,k_1,k_2 \ra_{B}
\nn \\
&=&
 \sum_{K_i=| k_i-{\ell} |}^{k_i+{\ell} }
 \left( \delta^{K_1}_{k_1}\delta^{K_2}_{k_2}  - \ri t {[} Z(can)_{\ell}^{j_{1},j_{2}} {]}^{ K_1,K_2 }_{ \phantom{ K_1,K_2 } k_1, k_2 } \, \right)
\underbrace{ | j_1,K_1,K_2 \ra_{A} \otimes | j_2,K_1,K_2 \ra_{B} }_{ | \Psi_{can,\{j_i,K_i\}} \ra }
+ O(t^2)
\,, \label{eq:ShortTime-BasisState-candygraph}
\eeq
where $| j_i,K_1,K_2 \ra_{A} \in \cH_{A}$ and $| j_2,K_1,K_2 \ra_{B} \in \cH_{B}$ denote the intertwiners living at respect trivalent vertex.
According to equation (\ref{eq:Amplitudes-6jsymbols}), the transition matrix $Z(can)_{\ell}$ is given by
\beq
{[} Z(can)_{\ell}^{j_{1},j_{2}} {]}^{ K_1,K_2 }_{ \phantom{ K_1,K_2 } k_1, k_2 }
=
(-1)^{j_1+j_2+k_1+k_2+K_1+K_2+2\ell}
\begin{Bmatrix}
   j_1 & k_1 & k_2 \\
   \ell &  K_2 & K_1
  \end{Bmatrix}
\begin{Bmatrix}
   j_{2} & k_1 & k_2 \\
   \ell &  K_2 & K_1
  \end{Bmatrix}
  \prod_{i=1}^{2}\sqrt{(2k_i+1)(2K_i+1)}
\,.
\eeq
The matrix elements are all real numbers.
%
%Now we study the dynamics with truncation up to the 1st-order of $t$, i.e., ignoring $O(t^2)$ terms. One could also truncate up to the 2nd-order of $t$, and it turns out that the leading order (order of $t^2$) of entanglement is same with the case of 1st-order truncation.
We take special care in properly normalizing the truncated state,
%Note that the truncated state is not normalized, and the normalized state with respect to the 1st-order truncation is read below:
\be
| \Psi_{can,\{j_i,k_i\}} (\ell,t) \ra
=
\f{ | j_1,k_1,k_2 \ra_{A} \otimes | j_2,k_1,k_2 \ra_{B} - \ri t
\sum_{ \{K_i\} }
{[} Z(can)_{\ell}^{j_{1},j_{2}} {]}^{ K_1,K_2 }_{ \phantom{ K_1,K_2 } k_1, k_2 }
| j_1,K_1,K_2 \ra_{A} \otimes | j_2,K_1,K_2 \ra_{B} }
{ \sqrt{ 1+t^2 \sum_{s=0}^{2\ell} {[} Z(can)_{s}^{j_{1},j_{2}} {]}^{ k_1,k_2 }_{ \phantom{ k_1,k_2 } k_1, k_2 } }   }
\,. \label{eq:NormalizedShortTimeState-candygraph}
\ee
The normalization factor, at the denominator, can be computed explicitly using the composition rule of $\{6j\}$-symbols\footnotemark{}:
%$\la \Psi | \Psi \ra$ is computed by composition rule (\ref{eq:Composition-Amplitudes}), in the present case read as
%
\footnotetext{We use the following identity on sums of $\{6j\}$-symbols, for which we choose $k_1=\tk_1,k_2=\tk_2$:
\beq
&&\sum_{ K_1,K_2}
(-1)^{\sum_{i}(\tk_i-k_i)}
\begin{Bmatrix}
   j_1 & \tk_1 & \tk_2 \\
   \ell_1 &  K_2 & K_1
  \end{Bmatrix}
\begin{Bmatrix}
   j_{2} & \tk_1 & \tk_2 \\
   \ell_1 &  K_2 & K_1
  \end{Bmatrix}
  \begin{Bmatrix}
   j_1 & k_1 & k_2 \\
   \ell_2 &  K_2 & K_1
  \end{Bmatrix}
\begin{Bmatrix}
   j_{2} & k_1 & k_2 \\
   \ell_2 &  K_2 & K_1
  \end{Bmatrix}
  \prod_{i=1}^{2}\sqrt{(2\tk_i+1)(2k_i+1)}(2K_i+1)
  \,, \nn \\
  &=&
  \sum_{ s=|\ell_1-\ell_2|}^{\ell_1+\ell_2}
(-1)^{j_{1}+j_{2}+\tk_1+\tk_2+k_1+k_2+2s}
\begin{Bmatrix}
   j_1 & \tk_1 & \tk_2 \\
   s &  k_2 & k_1
  \end{Bmatrix}
\begin{Bmatrix}
   j_{2} & \tk_1 & \tk_2 \\
   s &  k_2 & k_1
  \end{Bmatrix}
  \prod_{i=1}^{2}\sqrt{(2\tk_i+1)(2k_i+1)}  
\,. 
%\label{eq:Amplitudes-candygraph-6js}
\nn
\eeq
}
%By means of taking $k_1=\tk_1,k_2=\tk_2$, the normalization factor, which is the denominator in equation (\ref{eq:NormalizedShortTimeState-candygraph}), reads
\beq \label{eq:NormalizationFactor-candygraph}
N_{can,\{j_i,k_i\}} (\ell,t)
&=&
1+ t^2 \sum_{ s=0}^{2\ell}
(-1)^{j_{1}+j_{2}+2k_1+2k_2+2s}
\begin{Bmatrix}
   j_1 & k_1 & k_2 \\
   s &  k_2 & k_1
  \end{Bmatrix}
\begin{Bmatrix}
   j_{2} & k_1 & k_2 \\
   s &  k_2 & k_1
  \end{Bmatrix}
  \prod_{i=1}^{2}(2k_i+1)
\\
&=&
1+t^2
\sum_{ K_1,K_2}
  \begin{Bmatrix}
   j_1 & k_1 & k_2 \\
   \ell &  K_2 & K_1
  \end{Bmatrix}^2
\begin{Bmatrix}
   j_{2} & k_1 & k_2 \\
   \ell &  K_2 & K_1
  \end{Bmatrix}^2
  \prod_{i=1}^{2}(2k_i+1)(2K_i+1)
  \,.
  \label{eq:NormalizationFactor-candygraph-spinshift}
\eeq
It depends on the holonomy operator spin $\ell$, on the boundary spins $j_1$, $j_2$ and on the bulk spins $k_1$, $k_2$. It does not contain a first order term in $t$.

\medskip

We now compute the entanglement entropy from the truncated state (\ref{eq:NormalizedShortTimeState-candygraph}). Since the initial state is unentangled, this entanglement entropy is entirely created by the process.
The density matrix is written explicitly:
\beq
&&
\rho_{can_{AB} }(t)
=| \Psi_{can,\{j_i,k_i\}} (\ell,t) \ra\la \Psi_{can,\{j_i,k_i\}} (\ell,t) |
\\
&=&
\f{1}{ 1+t^2 \sum_{s=0}^{2\ell} {[} Z(can)_{s}^{j_{1},j_{2}} {]}^{ k_1,k_2 }_{ \phantom{ k_1,k_2 } k_1, k_2 } }
\bigg(
| j_1,k_1,k_2 \ra \la j_1,k_1,k_2 |_{A} \otimes | j_2,k_1,k_2 \ra\la j_2,k_1,k_2 |_{B}
\nn \\
&&
+ t^2 \sum_{ \{K_i, K'_i \} }
{[} Z(can)_{\ell}^{j_{1},j_{2}} {]}^{ K_1,K_2 }_{ \phantom{ K_1,K_2 } k_1, k_2 }
{[} Z(can)_{\ell}^{j_{1},j_{2}} {]}^{ K'_1,K'_2 }_{ \phantom{ K'_1,K'_2 } k_1, k_2 }
| j_1,K_1,K_2 \ra \la j_1,K'_1,K'_2 |_{A} \otimes | j_2,K_1,K_2 \ra\la j_2,K'_1,K'_2 |_{B}
\nn \\
&&
- \ri t
\sum_{ \{K_i\} }
{[} Z(can)_{\ell}^{j_{1},j_{2}} {]}^{ K_1,K_2 }_{ \phantom{ K_1,K_2 } k_1, k_2 }
| j_1,K_1,K_2 \ra\la j_1,k_1,k_2 |_{A} \otimes | j_2,K_1,K_2 \ra\la j_2,k_1,k_2 |_{B}
\nn \\
&&
+ \ri t
\sum_{ \{K_i\} }
{[} Z(can)_{\ell}^{j_{1},j_{2}} {]}^{ K_1,K_2 }_{ \phantom{ K_1,K_2 } k_1, k_2 }
| j_1,k_1,k_2 \ra\la j_1,K_1,K_2 |_{A} \otimes | j_2,k_1,k_2 \ra\la j_2,K_1,K_2 |_{B}
\bigg)
\quad \in \textrm{End} (\cH_{A} \otimes \cH_{B} )
\,.
\nn
\eeq
The reduced density matrix $\rho_{can_A } \in \textrm{End} (\cH_{A})$ is obtained via partial tracing over $\cH_{B}$, which is done via choosing orthonormal basis $| j_2,K_1,K_2 \ra_{B}$ to implement $\sum_{K_1,K_2} \la j_2,K_1,K_2 | \rho_{can_{AB} }(t) | j_2,K_1,K_2 \ra_{B}$, so the reduced density matrix $\rho_{can_A }(t)$ reads:
\be
\rho_{can_A }(t)
=
\f{ | j_1,k_1,k_2 \ra \la j_1,k_1,k_2 |_{A}
+ t^2 \sum_{ \{K_i \} }
\Big( {[} Z(can)_{\ell}^{j_{1},j_{2}} {]}^{ K_1,K_2 }_{ \phantom{ K_1,K_2 } k_1, k_2 } \Big)^2
| j_1,K_1,K_2 \ra \la j_1,K_1,K_2 |_{A} }
{ 1+t^2 \sum_{s=0}^{2\ell} {[} Z(can)_{s}^{j_{1},j_{2}} {]}^{ k_1,k_2 }_{ \phantom{ k_1,k_2 } k_1, k_2 } }
\,.
\ee
%Another manner is to have $\rho_{can_B} \in \textrm{End} (\cH_{B})$ via partial tracing over $\cH_{A}$. The two choices for partial trace do not matter to entanglement entropy, since the reduced density matrices share same Schmidt eigenvalues (nonzero eigenvalues).
%
%\smallskip
%
The eigenvalues of $\rho_{can_A }(t)$ can be read off directly from this formula since the reduced density matrix is diagonal in the  $| j_1,K_1,K_2 \ra_{A}$ basis,
\be
\lambda_{ \rho_{can_A } }[K_1,K_2]
=
\f{   \delta^{K_1}_{k_1} \, \delta^{K_2}_{k_2}
+ t^2 \,
\begin{Bmatrix}
   j_1 & k_1 & k_2 \\
   \ell &  K_2 & K_1
  \end{Bmatrix}^2
\begin{Bmatrix}
   j_{2} & k_1 & k_2 \\
   \ell &  K_2 & K_1
  \end{Bmatrix}^2
  \prod_{i=1}^{2}(2k_i+1)(2K_i+1)   }
{ 1+t^2 \sum_{s=0}^{2\ell} \, (-1)^{j_{1}+j_{2}+2k_1+2k_2+2s}
\begin{Bmatrix}
   j_1 & k_1 & k_2 \\
   s &  k_2 & k_1
  \end{Bmatrix}
\begin{Bmatrix}
   j_{2} & k_1 & k_2 \\
   s &  k_2 & k_1
  \end{Bmatrix}
  \prod_{i=1}^{2}(2k_i+1) }
\,. \label{eq:eigenvalues-candygraph}
\ee
%{\bf CRYPTIC COMMENT}
%It should be careful about the difference between the case of half-integer $\ell\in \N +\f12$ and integer $\ell\in \N$: when comes to zero-shifting $\veps_1 \equiv K_1-k_1=0$, $\veps_2\equiv K_2-k_2=0$, the only zero-shifting term comes from initial basis state $| j_1,k_1,k_2 \ra_A \otimes | j_2,k_1,k_2 \ra_B$ if $\ell\in \N+\f12$, whereas there are two zero-shifting term if $\ell\in \N$, one comes from initial basis state, another from the dynamics.
%
%\medskip
%
%
%{\bf CRYPTIC LOGIC}
%%
%We use below expressions for the entropy:
%\be \label{eq:LeadingExpansion}
%\f{1+c t^2}{1+f t^2} \ln \f{1+a t^2}{1+b t^2}
%=
%(a-b) t^2 + O(t^4)
%\,, \qquad
%\f{c t^2}{1+f t^2} \ln \f{a t^2}{1+b t^2}
%=
%0 \ln 0
%\,,
%\ee
%where the second prepares for the case of zero spin-shifting when $\ell \in \N+\f12$, which simply vanishes in computation due to the information entropy convention.
%
%The expansion implies that the only $O(t^2)$-contribution comes from the proportion of the initial spin network basis state in final state.
%
%
%Hence, eigenvalues (\ref{eq:eigenvalues-candygraph}) give us the von Neumann entropy for the truncated state (\ref{eq:NormalizedShortTimeState-candygraph})
%

We have a diagonal reduced density matrix of the type:
\be
\rho_{A}\approx \textrm{diag}\big{[}
(1-\Lambda\eps),\,a_{1}\eps,\,a_{2}\eps,\,\dots
{]} + \cO(\eps^{2})
\,,\qquad
\tr \rho_{A}=1 \Rightarrow \Lambda= \sum_{m\ge 1} a_{m}\,,
\ee
at linear order in the infinitesimal parameter $\eps$, which is to be identified to the squared time, $\eps=t^{2}$.
If one considers the von Neumann entropy as the measure of entanglement, one gets:
\be
S_{vN}[\rho_{A}]=-\tr(\rho_{A}\ln\rho_{A})
\approx
-\Lambda \eps\ln\eps +\eps (\Lambda-\sum_{m\ge 1}a_{m}\ln a_{m})\,,
\ee
which looks regular  at first glance but actually has a divergent derivative at $\eps=0$ due to the $\eps\ln\eps$ term. This is simply traced back to  the vanishing eigenvalues at initial time, i.e. our choice of initial separable state. Although we could go on working with the von Neumann entropy, it appears simpler to turn to the linear entropy (or quadratic Tsallis entropy), which is one minus the fidelity:
\be
S_{lin}[\rho_{A}]
=
1-\tr( \rho_{A}^{2})
\approx
2\Lambda \eps\,.
\ee
The leading order coefficient is the same as the coefficient in front of the divergent derivative term $ \eps\ln\eps $ of the von Neumann entropy, so they are understood to reflect the same growth rate of entanglement. 
Furthermore, one should realize that the largest eigenvalue is actually the projection of the the density matrix onto the initial separable state, which gives directly, in our case, the geometric entanglement at leading order:
\be
S_{g}\approx -\ln(1-\Lambda \eps) \approx \Lambda \eps,
\ee
which once again gives the coefficient $\Lambda$ as the growth factor of the entanglement at leading order in $\eps=t^{2}$.

\medskip

Coming back to the expression of the eigenvalues $\lambda_{ \rho_{can_A } }[K_1,K_2]$ in terms of $6j$-symbols, one extract  the growth factor $\Lambda$ from the Taylor expansion of the largest eigenvalue, obtained for $(K_{1},K_{2})=(k_{1},k_{2})$,
\be
\lambda_{ \rho_{can_A } }[k_1,k_2]\approx 1-\Lambda t^{2}+\cO(t^{3})\,,
\ee
which gives the leading order linear entropy::
%
%The von Neumann entropy $S=-\sum_{i} \lambda_{i} \ln \lambda_{i}$ can be re-written at 2nd-order in $t$ in a more synthetic way:
%%
%%we mention a thermodynamical way to read the entropy (\ref{eq:IntertwinerEE-candygraph-6j}): the entropy may be rewritten in terms of normalization factor (\ref{eq:NormalizationFactor-candygraph-spinshift})
%%
%\be
%S(\rho_{can_{A} },t)
%=
%\ln \f{ N_{can,\{j_i,k_i\}} (\ell,t) }{ N^0_{can,\{j_i,k_i\}} (\ell,t) } +O(t^4)
%\,, \qquad
%N^0_{can,\{j_i,k_i\}} (\ell,t)
%=
%1+
%t^2
%\begin{Bmatrix}
%   j_1 & k_1 & k_2 \\
%   \ell &  k_2 & k_1
%  \end{Bmatrix}^2
%\begin{Bmatrix}
%   j_{2} & k_1 & k_2 \\
%   \ell &  k_2 & k_1
%  \end{Bmatrix}^2
%  \prod_{i=1}^{2}(2k_i+1)^2
%\,, \label{eq:ThermalEntropy-candy}
%\ee
%where the $N_{can,\{j_i,k_i\}} (\ell,t)$ is interpreted as the density of spin networks and $N^0_{can,\{j_i,k_i\}} (\ell,t)$ is interpreted as the contribution of the initial basis state $| j_1,k_1,k_2 \ra_{A} \otimes | j_2,k_1,k_2 \ra_{B}$ to the final state. 
%%So the entropy change is interpreted as the thermal entropy contributed from the new created spin network basis states.
%%
%We expand the entropy up to 2nd-order in $t$ via Taylor series around the initial time $t=0$,
%%
\beq
{{\f12
}}
S(\rho_{can_{A} },t)
&=&
 t^2 \prod_{i=1}^{2}(2k_i+1) \sum_{s=0}^{2\ell} \, (-1)^{j_{1}+j_{2}+2k_1+2k_2+2s}
\begin{Bmatrix}
   j_1 & k_1 & k_2 \\
   s &  k_2 & k_1
  \end{Bmatrix}
\begin{Bmatrix}
   j_{2} & k_1 & k_2 \\
   s &  k_2 & k_1
  \end{Bmatrix}
\label{eq:IntertwinerEE-candygraph-6j}  \\
&&-t^2
\begin{Bmatrix}
   j_1 & k_1 & k_2 \\
   \ell &  k_2 & k_1
  \end{Bmatrix}^2
\begin{Bmatrix}
   j_{2} & k_1 & k_2 \\
   \ell &  k_2 & k_1
  \end{Bmatrix}^2
  \prod_{i=1}^{2}(2k_i+1)^2
 + O(t^4)
\,. \nn
\eeq
Let us first point out that the second term vanishes automatically when the holonomy operator spin is odd, $\ell \in \N+\f12$.
Then this bipartite entanglement exhibits the same plateau behavior as the geometric entanglement studied in the previous section: beyond the critical value $\ell_{c}=\min\{ 2k_1,2k_2\}+\f12$, the entropy $S(\rho_{can_{A} },t)$ (at second order) does not depend on $\ell\ge\ell_{c}$.
The origin of this plateau is simply the triangle condition of the spins: the $\{6j\}$-symbols do not vanish  only if  $s\leq 2k_1$ and $s\leq 2k_2$. A consequence is that the factor of the first term becomes constant (with respect to $\ell$) as soon as  $\ell \leq \min\{ k_1,k_2\}$ while the second term similarly does not depend on $\ell$ as soon as $\ell > \min\{ 2k_1,2k_2\}$.
%
%Similar to the analysis below the equation (\ref{eq:1LoopTrivalentSpinNetwork-2ndDerivative}), the entropy has critical point, due to the triangle condition on $6j$-symbols: non-vanishing $6j$-symbols imply $s\leq 2k_1$ and $s\leq 2k_2$, together with upper bound of $s_{\max}=2\ell$ in summation, meanwhile from the second term non-vanishing $6j$-symbols imply $\ell\leq k_1$ and $\ell \leq k_2$, i.e., $\ell \leq \min\{ k_1,k_2\}$. If $\ell > \min\{ k_1,k_2\}$, the first term on right side will be stable as $\ell$ increase. For $\ell \in \N$, second term will be also stable if $\ell > \min\{ 2k_1,2k_2\}$. The entropy then becomes constant as $\ell \geq \min\{ 2k_1,2k_2\}+\f12$.
%
Hence this confirms the critical value analysis for the geometric entanglement as given by equation (\ref{eq:2ndDGME-ClosureDefect}) in the previous section.

\smallskip

To conclude this section, we remark that, on top of this similar plateau behavior of the linear entropy for large spins $\ell$, this entanglement entropy  looks also very close to the geometric entanglement (\ref{eq:1LoopTrivalentSpinNetwork-2ndDerivative}) computed previously. Indeed, comparing the formulas, it appears that the geometric entanglement at 2nd order in $t$ corresponds to the eigenvalue $\lambda_{ \rho_{can_A } }[K_1,K_2]$ with no spin shift, $(K_{1},K_{2})=(k_{1},k_{2})$. We look into the relation between these two measures of entanglement  in more details below and show that they are indeed equal at 2nd order.

%Not only from the form of entropy (\ref{eq:IntertwinerEE-candygraph-6j}), and also from the analysis of critical point (stable point), the entanglement entropy computed in this subsection is very closed to the dispersion (\ref{eq:1LoopTrivalentSpinNetwork-2ndDerivative}).
%Indeed, the contribution almost comes from the zero-shifting eigenvalue of (\ref{eq:eigenvalues-candygraph}), namely, when $K_1=k_1$ and $K_2=k_2$, which is the maximal Schmidt eigenvalue at the neighborhood of $t=0$, thus it is the entanglement eigenvalue from the view of geometric measure of entanglement.
%%It implies the entangled state maximally projected onto the initial state,
%In the next subsection, we will show that they are equivalent in the sense of 1st- and 2nd-order of $t$.

%%%
\subsection{Geometric entanglement and holonomy operator dispersion}
%\subsection{Dispersion and entanglement entropy}
%%%

In this subsection we compare the holonomy operator dispersion (\ref{eq:1LoopTrivalentSpinNetwork-2ndDerivative}), which gives the 2nd term coefficient of the geometric entanglement, with the reduced density matrix  entropy computed above. It turns out that the geometric entanglement and the bipartite entanglement entropy are equal at 2nd-order of $t$, which is a neat consistency check of our approach in the simple example of the candy graph.

\smallskip

In order to compute the geometric entanglement via the  holonomy operator dispersion formula (\ref{eq:1LoopTrivalentSpinNetwork-2ndDerivative}), we first need to gauge-fix the bulk spin network and derive the closure defect distribution. Here, in the case of the candy graph, there are two possible gauge-fixing choices: either we gauge-fix the holonomy along the second egde $k_2$ to the identity as in fig.\ref{fig:gaugefixing-candygraph-1}, or we gauge-fix the first edge  $k_1$ to a trivial holonomy as in fig.\ref{fig:gaugefixing-candygraph-2}.
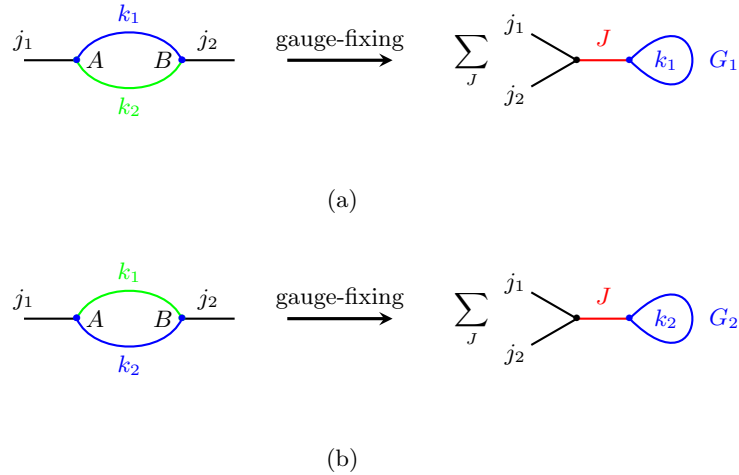
\begin{figure}[htb]
\begin{subfigure}[t]{0.5\linewidth}
\begin{tikzpicture}[scale=0.7]

\coordinate (A1) at (0,0);
\coordinate (A2) at (2,0);

\draw[thick] (A1) node[right] {$A$}-- ++ (180:1) node[above] {$j_1$} ;
\draw[thick] (A2) node[left] {$B$} -- ++ (1,0) node[above,midway] {$j_{2}$};

\draw[blue,thick,in=115,out=65,rotate=0] (A1) to node[above,midway] {$k_1$} (A2) node[scale=0.7] {$\bullet$};
\draw[green, thick,out=245,in=-65] (A2) to node[below] {$k_2$} (A1) node[scale=0.7] {$\bullet$};

\draw (A1) node[blue,scale=0.7] {$\bullet$};
\draw (A2) node[blue,scale=0.7] {$\bullet$};

\draw[->,>=stealth,very thick] (4,0) -- node[above] {gauge-fixing} (6,0);

\coordinate (O) at (9.5,0);

\draw[thick,red] (O) -- ++ (0:1) node[midway,above=2] {$J$} coordinate (B) node[blue,scale=0.7] {$\bullet$};
\draw[blue,thick,in=-45,out=45,scale=4.5,rotate=0] (B)  to[loop] node[midway,left=2] {$k_1$} (B) ++(0:0.4) node {$G_1$};

\draw[thick] (O) -- ++ (150:1) ++(150:0.35) node {$j_1$};
\draw[thick] (O) ++ (180:2) node {$\displaystyle{\sum_{J} }$};
\draw[thick] (O) -- ++ (210:1) ++(210:0.35) node {$j_2$};

\draw (O) node[scale=0.7] {$\bullet$};

\end{tikzpicture}
\caption{}
\label{fig:gaugefixing-candygraph-1}
\end{subfigure}
\begin{subfigure}[t]{0.5\linewidth}
\begin{tikzpicture}[scale=0.7]
\coordinate (A1) at (0,0);
\coordinate (A2) at (2,0);

\draw[thick] (A1) node[right] {$A$}-- ++ (180:1) node[above] {$j_1$} ;
\draw[thick] (A2) node[left] {$B$} -- ++ (1,0) node[above,midway] {$j_{2}$};

\draw[green,thick,in=115,out=65,rotate=0] (A1) to node[above,midway] {$k_1$} (A2) node[scale=0.7] {$\bullet$};
\draw[blue, thick,out=245,in=-65] (A2) to node[below] {$k_2$} (A1) node[scale=0.7] {$\bullet$};

\draw (A1) node[blue,scale=0.7] {$\bullet$};
\draw (A2) node[blue,scale=0.7] {$\bullet$};

\draw[->,>=stealth,very thick] (4,0) -- node[above] {gauge-fixing} (6,0);

\coordinate (O) at (9.5,0);

\draw[thick,red] (O) -- ++ (0:1) node[midway,above=2] {$J$} coordinate (B) node[blue,scale=0.7] {$\bullet$};
\draw[blue,thick,in=-45,out=45,scale=4.5,rotate=0] (B)  to[loop] node[midway,left=2] {$k_2$} (B) ++(0:0.4) node {$G_2$};

\draw[thick] (O) -- ++ (150:1) ++(150:0.35) node {$j_1$};
\draw[thick] (O) ++ (180:2) node {$\displaystyle{\sum_{J} }$};
\draw[thick] (O) -- ++ (210:1) ++(210:0.35) node {$j_2$};

\draw (O) node[scale=0.7] {$\bullet$};
\end{tikzpicture}
\caption{}
\label{fig:gaugefixing-candygraph-2}
\end{subfigure}
\caption{The gauge-fixings on candy graph. The {\color{green}{green}} labels the maximal tree.
}
\end{figure}
These two gauge-fixings lead to different closure defect probability distributions $p_{k_1}(J)$ and $p_{k_2}(J)$, explicitly given by
\be
p_{k_1}(J)=(2k_2+1)(2J+1)\begin{Bmatrix}
J & k_1 & k_1 \\
k_2 & j_1 & j_2
\end{Bmatrix}^2
\,, \qquad
p_{k_2}(J)=(2k_1+1)(2J+1)\begin{Bmatrix}
J & k_2 & k_2 \\
k_1 & j_1 & j_2
\end{Bmatrix}^2
\,.
\ee
Both distributions are normalized, $\sum_{J} p_{k_1}(J)=\sum_{J} p_{k_2}(J)=1$. 
The inequality $p_{k_1}(J) \neq p_{k_2}(J)$ reflects the fact that different choices of gauge-fixing path translate into different boundary maps from bulk holonomies onto boundary states.
Numerically, the difference can be striking. For instance, if set $k_1=6,k_2=8$, $j_1=5,j_2=4$, the closure defect distributions are, at three decimals,

\vspace*{2mm}
%\begin{tabular}{ |p{1.3cm} | p{1.3cm}| p{1.3cm}|p{1.3cm} | p{1.3cm}| p{1.3cm}|p{1.3cm} | p{1.3cm}| p{1.3cm}| p{1.3cm} |}
% \hline
% & $ J=1 $ & $J=2$ & $J=3$ & $J=4$ & $J=5$ & $J=6$ & $J=7$ & $J=8$ & $J=9$ \\
% \hline
% $p_{k_1}(J)$ & $\f{85}{429}$ & $\f{85}{9438}$ & $\f{81685}{490776}$ & $\f{11163}{163592}$ & $\f{1849}{22308}$ & $\f{26645}{89661}$ & $\f{6083045}{39630162}$ & $\f{18375}{777062}$ & $\f{588}{600457}$ \vspace*{1mm}\\
% \hline
%  $p_{k_2}(J)$ & $\f{455}{6732}$ & $\f{4095}{28424}$ & $\f{250}{3553}$ & $\f{224}{138567}$ & $\f{3773}{33592}$ & $\f{1035}{7106}$ & $\f{110515}{14133834}$ & $\f{12075}{92378}$ & $\f{19320}{60401}$ \vspace*{1mm}\\
% \hline
%\end{tabular}
\begin{tabular}{ |p{1.3cm} | p{1.3cm}| p{1.3cm}|p{1.3cm} | p{1.3cm}| p{1.3cm}|p{1.3cm} | p{1.3cm}| p{1.3cm}| p{1.3cm} |}
 \hline
 & $ J=1 $ & $J=2$ & $J=3$ & $J=4$ & $J=5$ & $J=6$ & $J=7$ & $J=8$ & $J=9$ \\
 \hline
 $p_{k_1}(J)$ & $0.198$ & $0.009$ & $0.166$ & $0.068$ & $0.083$ & $0.297$ & $0.153$ & $0.024$ & $0.001$ \\
 \hline
  $p_{k_2}(J)$ & $0.068$ & $0.144$ & $0.070$ & $0.002$ & $0.112$ & $0.146$ & $0.008$ & $0.131$ & $0.320$ \\
 \hline
\end{tabular}

\vspace*{2mm}
Although the closure defect probability distribution depends on the gauge-fixing, the loop holonomy operator is gauge-invariant, and thus its dispersion computed from either gauge-fixing choice turns out to be the same. Using $p_{k_1}(J)$ we have explicitly:
\beq
\f12\f{ \rd^2 S_{g} }{ \rd t^2} \Big\vert_{t=0}
&=&
\sum_{J} (2k_2+1)(2J+1)\begin{Bmatrix}
J & k_1 & k_1 \\
k_2 & j_1 & j_2
\end{Bmatrix}^2 \sum_{s=0(1)}^{2\ell} (-1)^{J+s+2k_1}
\begin{Bmatrix}
J & k_1 & k_1 \\
s & k_1 & k_1
\end{Bmatrix}
(2k_1+1)
\nn
\\
&&
-
\left(
\sum_{J} (2k_2+1)(2J+1)\begin{Bmatrix}
J & k_1 & k_1 \\
k_2 & j_1 & j_2
\end{Bmatrix}^2 (-1)^{J+\ell+2k_1}
\begin{Bmatrix}
J & k_1 & k_1 \\
\ell & k_1 & k_1
\end{Bmatrix}
(2k_1+1)
\right)^2
\,. \label{eq:2ndDerivative-candygraph}
\eeq
We can prove that the  formula above is actually equal to the geometric entropy (\ref{eq:IntertwinerEE-candygraph-6j}). More precisely, what we need to prove is:
\be
\begin{Bmatrix}
j_1 & k_1 & k_2 \\
\ell & k_2 & k_1
\end{Bmatrix}
\begin{Bmatrix}
j_2 & k_1 & k_2 \\
\ell & k_2 & k_1
\end{Bmatrix}
=
(-1)^{j_1+j_2+2k_1+2k_2+2\ell}
\sum_{J} (2J+1)\begin{Bmatrix}
J & k_1 & k_1 \\
k_2 & j_1 & j_2
\end{Bmatrix}^2 (-1)^{J+\ell+2k_1}
\begin{Bmatrix}
J & k_1 & k_1 \\
\ell & k_1 & k_1
\end{Bmatrix}
\,,
\ee
which is simply a particular case of the general Biedenharn-Elliot identity (see e.g. \cite{Bonzom:2009zd})
\be
\begin{Bmatrix}
j & h & g \\
k & a & b
\end{Bmatrix}
\begin{Bmatrix}
j & h & g \\
f & d & c
\end{Bmatrix}
=
\sum_{l} (-1)^{a+b+c+d+f+k+h+g+j+l}(2l+1)
\begin{Bmatrix}
k & f & l \\
d & a & g
\end{Bmatrix}
\begin{Bmatrix}
a & d & l \\
c & b & j
\end{Bmatrix}
\begin{Bmatrix}
b & c & l \\
f & k & h
\end{Bmatrix}
\,.
\label{eq:B-E-identity}
\ee

Summing over the spin $J$, one recovers exactly the formula for the 2nd order coefficient of the reduced density matrix linear entropy (\ref{eq:IntertwinerEE-candygraph-6j}). This is not only a check that the two measures of entanglements -the bipartite entanglement between the two candy graph vertices and the multipartite geometric entanglement- are equal at leading order in the time $t$, but it also confirms that the geometric entanglement does not depend on the  choice of gauge-fixing tree.

%\smallskip
%
%At the end of the day, we have shown that the dispersion of loop holonomy operator (the 2nd-order derivative of geometric measure of entanglement) is indeed equal to the entanglement entropy last subsection in the sense of leading term: both of them have vanishing 1st-order of $t$, and both have equal 2nd-order derivative. We also show that although the closure defect probability distribution relies on the path of gauge-fixing, the dispersion does not rely on it.
%{\bf{A more important clue may be: the boundary (closure defect) and bulk (spin networks) can be linked via entanglement excitation.}}

%%%%%%

%%%
\subsection{Geometric interpretation of the entanglement in the semi-classical regime}
%\subsection{From discreteness to smoothness}
%%%

We would like to provide the entanglement calculations with a geometric interpretation, for instance understand the extrema of the entanglement in terms of the geometry represented by the spin network states. To this purpose, we work in the semi-classical regime of spin networks at large spins, $j_i,k_i\,\gg1$.
The spin network geometry can be interpreted in terms of the  dual triangulation as illustrated on fig.\ref{fig:candygraph-Triangulation}.
%
%In this subsection, we analyze asymptotic behavior of entanglement excitation, for the sake of seeking the link between entanglement and geometry, also recovering the smooth geometry from discretization.
%So we are led into the regime of large spins. 
%Two cases are concerned here: $\ell \ll \{j_i,k_i\}$ and $\ell \sim \{j_i,k_i\}$, namely, small loop holonomy spin-$\ell$ and large loop holonomy spin-$\ell$.
It turns out the entanglement excitation can be described in terms of the dual triangulation angles. Moreover, we find the maximal growth rates of the entanglement corresponds either to flat bulk geometry or maximally curved bulk geometry.
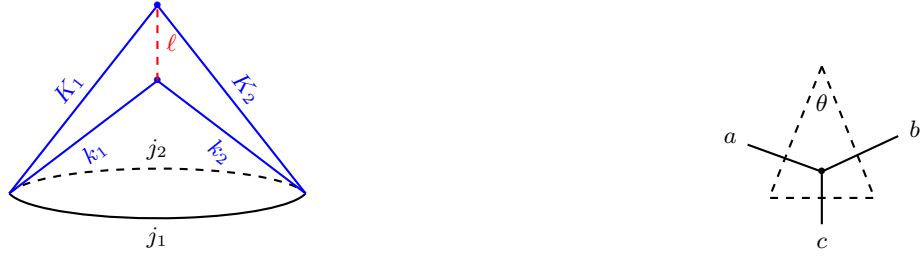
\begin{figure}
\begin{subfigure}[t]{0.55\linewidth}
\begin{tikzpicture}[thick,scale=1]
    \draw[dashed] (0,0) arc (170:10:2cm and 0.4cm) node[sloped, midway, above]{$j_2$} coordinate[pos=0] (a);
    \draw (0,0) arc (-170:-10:2cm and 0.4cm) node[sloped,midway,below]{$j_1$} coordinate (b);

    \draw[blue] (a) --node[sloped,midway,below] {$k_1$} ([yshift=1.5cm]$(a)!0.5!(b)$) node[scale=0.7] {$\bullet$} --node[sloped,midway,below] {$k_2$} (b);
        \draw[blue] (a) --node[sloped,midway,above] {$K_1$} ([yshift=2.5cm]$(a)!0.5!(b)$) node[scale=0.7] {$\bullet$} --node[sloped,midway,above] {$K_2$} (b);
  
\draw[dashed,red] ([yshift=1.5cm]$(a)!0.5!(b)$) --node[right,midway] {$\ell$} ([yshift=2.5cm]$(a)!0.5!(b)$);      
\end{tikzpicture}
\caption{The triangulation for the candy graph. The two triangles of respect length $j_1,k_1,k_2$ and $j_2,k_1,k_2$ are glued by matching bulk spins $k_1,k_2$. Loop holonomy operator changes the height of the cone. The bulk geometry can be viewed from the shape of the cone.
}
\label{fig:candygraph-Triangulation}
\end{subfigure}
\hspace*{10mm}
\begin{subfigure}[t]{0.3\linewidth}
\begin{tikzpicture}[thick,scale=0.7]

\coordinate (O) at (0,0);
\coordinate (A) at (0,2);
\coordinate (B) at (-1,-0.5);
\coordinate (C) at (1,-0.5);

\draw[dashed] (A) -- (B) -- (C) -- cycle;

\draw (O) -- ++ (160:1.5) ++ (160:0.35) node {$a$};

\draw (O) -- ++ (25:1.6) ++ (25:0.35) node {$b$};

\draw (O) -- ++ (270:1.0) ++ (270:0.35) node {$c$} ;

\draw (O) node[scale=0.7] {$\bullet$};

\draw (A) ++ (270:0.7) node {$\theta$};

\end{tikzpicture}
\caption{The triangle given by three spins. The spin-$c$ opposes to angle $\theta$.
}
\label{fig:triangle-3spins}
\end{subfigure}
\caption{The triangulation for candy graph spin network.
}
\end{figure}

\smallskip

We will look at two cases: small loop holonomy operator spin $\ell$ and large loop holonomy operator spin $\ell$.
Let us first look into the case of small loop holonomy operator spin, for which we have $\ell \ll \{j_i,k_i\}$. We use the Racah's approximation for $s \in \N$ and $s \ll a,b,c$ (cf. \cite{osti_4824659}):
\be
\begin{Bmatrix}
c & a & b \\
s & b & a
\end{Bmatrix}
\approx
\f{(-1)^{a+b+s+\ell} }{ \sqrt{ (2a+1)(2b+1) } } P_{s} (\cos\theta)
\,, \quad \text{with} \quad
\cos\theta=\f{ a(a+1)+b(b+1)-c(c+1) }{ 2\sqrt{ a(a+1)b(b+1) } }
\,. \label{eq:Angle-Asymptotic-6j-Racah}
\ee
The $P_{s}$'s are the  Legendre polynomials, while $\theta$ is the angle opposite of the edge of length $c$ in the triangle with edge lenghts $a,b,c$, as drawn in fig.\ref{fig:triangle-3spins}.
%
%{\bf IS THIS USEFUL or to be put in footnote ??The Racah's approximation is the prototype in the original derivation of Ponzano-Regge asymptotic formula for $6j$-symbol \cite{osti_4824659}. Indeed, when $n$ is large, the Legendre polynomial in Racah's approximation is rewritten by Laplace’s formula
%\be \label{eq:Legendre polynomial-Laplace’sFormula}
%P_{n}(\cos\theta)
%\approx
%\sqrt{ \f{2}{\pi n \sin\theta} } \cos\Big[ \Big(n+\f12\Big)\theta - \f{\pi}{4} \Big] + O(n^{-\f32})
%\,, \qquad
%\begin{Bmatrix}
%c & a & b \\
%s & b & a
%\end{Bmatrix}
%\approx
%\f{ (-1)^{a+b+c+s} }{ \sqrt{12\pi V} } \cos\left[ (s+\f12)\theta - \f14 \pi \right]
%\,,
%\ee
%where $6V \approx (a+\f12)(b+\f12)(s+\f12)\sin\theta$.
%One could rescale the spins with large number $\lambda$, such that $s$ is large but small with respect to $a,b,c$, the approximation is rewritten
%\be
%s,a,b,c \to \lambda s, \lambda a, \lambda b, \lambda c\,, \qquad
%\begin{Bmatrix}
%\lambda c & \lambda a & \lambda b \\
%\lambda s & \lambda b & \lambda a
%\end{Bmatrix}
%\approx
%\f{ (-1)^{\lambda(a+b+c+s)} }{ \sqrt{12\pi \lambda^3 V} } \cos\left[ \lambda(s+\f12)\theta - \f14 \pi \right]
%\,.
%\ee
%}
%
%
%\smallskip
%
By plugging the Racah's approximation into the linear entropy formula (\ref{eq:IntertwinerEE-candygraph-6j}), we derive an approximation in terms of the triangle angles:
\beq
\ell\in \N+\f12: \qquad
\f12 S(\rho_{can_{A} },t)
&\approx&
\f12 S(\rho^{\theta}_{can_{A} },t)
=
t^2 \sum_{s=0}^{2\ell} \, P_s (\cos\theta_1) P_s (\cos\theta_2) + O(t^4)
\label{eq:IntertwinerEE-candygraph-hN}
\\
\ell\in\N: \qquad
\f12 S(\rho_{can_{A} },t)
&\approx&
\f12 S(\rho^{\theta}_{can_{A} },t)
=
t^2 \sum_{s=0}^{2\ell} \, P_s (\cos\theta_1) P_s (\cos\theta_2) - t^2 [ P_{\ell}(\cos\theta_1) P_{\ell}(\cos\theta_2) ]^2 + O(t^4)
\,.
\label{eq:IntertwinerEE-candygraph-N}
\eeq
The trivial case with vanishing spin $\ell=0$ excites no entanglement as expected. The two expressions for odd and even spins would be exactly the same if one assumed the convention that half-integer Legendre polynomials vanish. The plots fig.\ref{fig:IntertwinerEE-ShortTimes-candygraph} show the growth rate of intertwiner entanglement provided by the approximation. The fig.\ref{fig:NumericalComparisonLS-candygraph} compares above approximation with entropy (\ref{eq:IntertwinerEE-candygraph-6j}).

%{\bf CLARIFY \& REWRITE}
%%
%{\bf Notice that approximation (\ref{eq:IntertwinerEE-candygraph-hN},\ref{eq:IntertwinerEE-candygraph-N}) is still acceptable approximation even for the cases $\ell \sim \{j_i,k_i\}$, as shown in fig.\ref{fig:NumericalComparisonLS-candygraph}. When all spins are large enough, the values flatten. The reason is when one sums over the spin-$s$ for normalization factor, the large-$s$ Legendre polynomials are suppressed due to the Laplace’s formula (\ref{eq:Legendre polynomial-Laplace’sFormula}) with factoring rescaling power $\lambda^{-\f32}$.
%}
%
\begin{figure}[!]
\begin{subfigure}[t]{.45\linewidth} \includegraphics[width=1.0\textwidth]{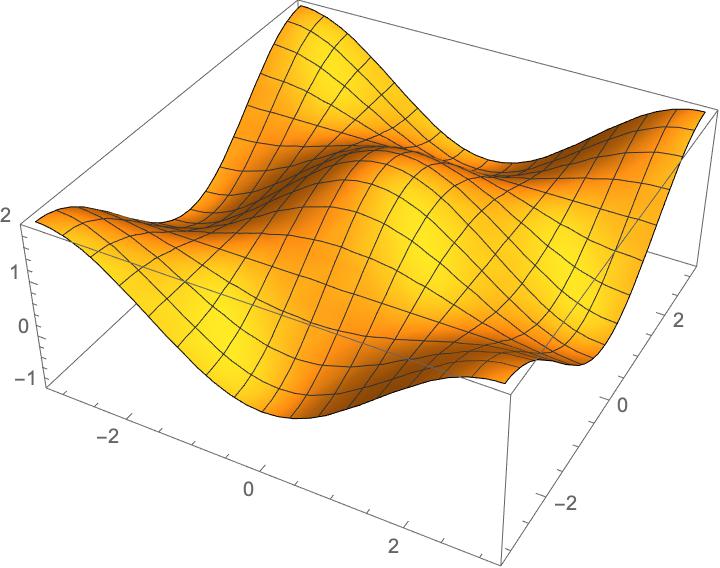}
\caption{$\ell=\f12$.}
\end{subfigure}
\begin{subfigure}[t]{.45\linewidth} \includegraphics[width=1.0\textwidth]{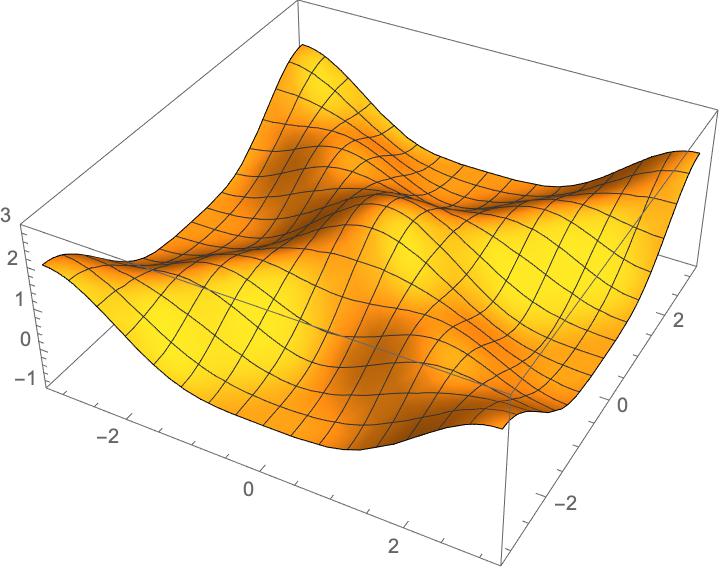}
\caption{$\ell=1$.}
\end{subfigure}
\begin{subfigure}[t]{.45\linewidth} \includegraphics[width=1.0\textwidth]{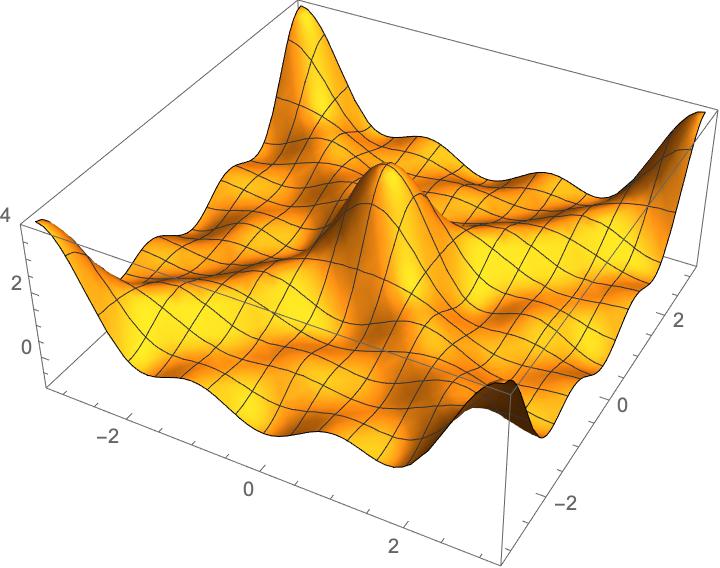}
\caption{$\ell=\f32$.}
\end{subfigure}
\begin{subfigure}[t]{.45\linewidth} \includegraphics[width=1.0\textwidth]{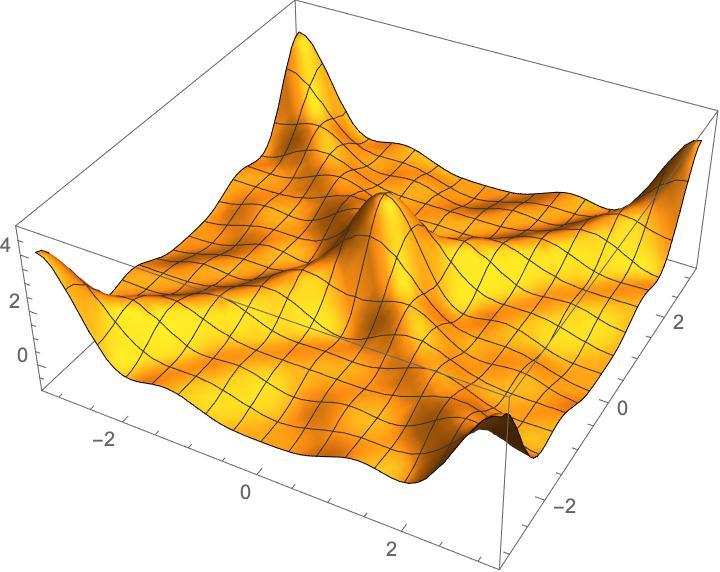}
\caption{$\ell=2$.}
\end{subfigure}
\begin{subfigure}[t]{.45\linewidth} \includegraphics[width=1.0\textwidth]{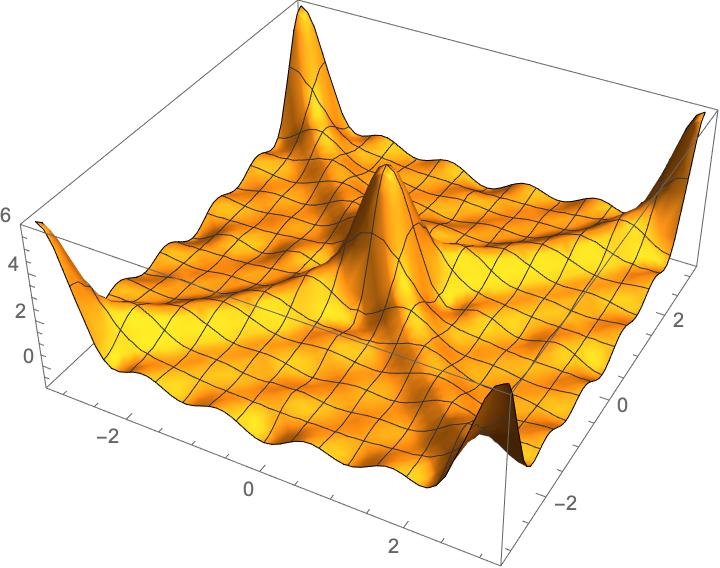}
\caption{$\ell=\f52$.}
\end{subfigure}
\begin{subfigure}[t]{.45\linewidth} \includegraphics[width=1.0\textwidth]{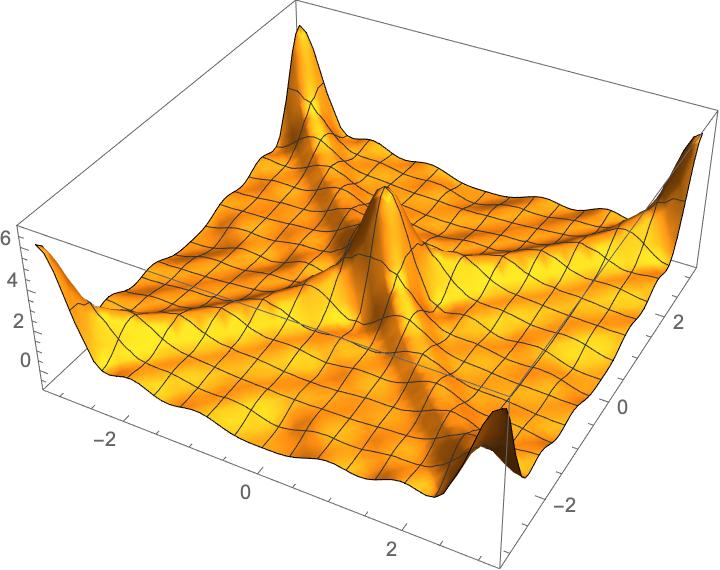}
\caption{$\ell=3$.}
\end{subfigure}
\caption{The plots of approximation (\ref{eq:IntertwinerEE-candygraph-hN},\ref{eq:IntertwinerEE-candygraph-N}) with low spin-$\ell$. The $x$ and $y$ axis presents corner angle $\theta_1$ and $\theta_{2}$ respectively}
\label{fig:IntertwinerEE-ShortTimes-candygraph}
\end{figure}
\begin{figure}
\begin{subfigure}[t]{.45\linewidth} \includegraphics[width=1.0\textwidth]{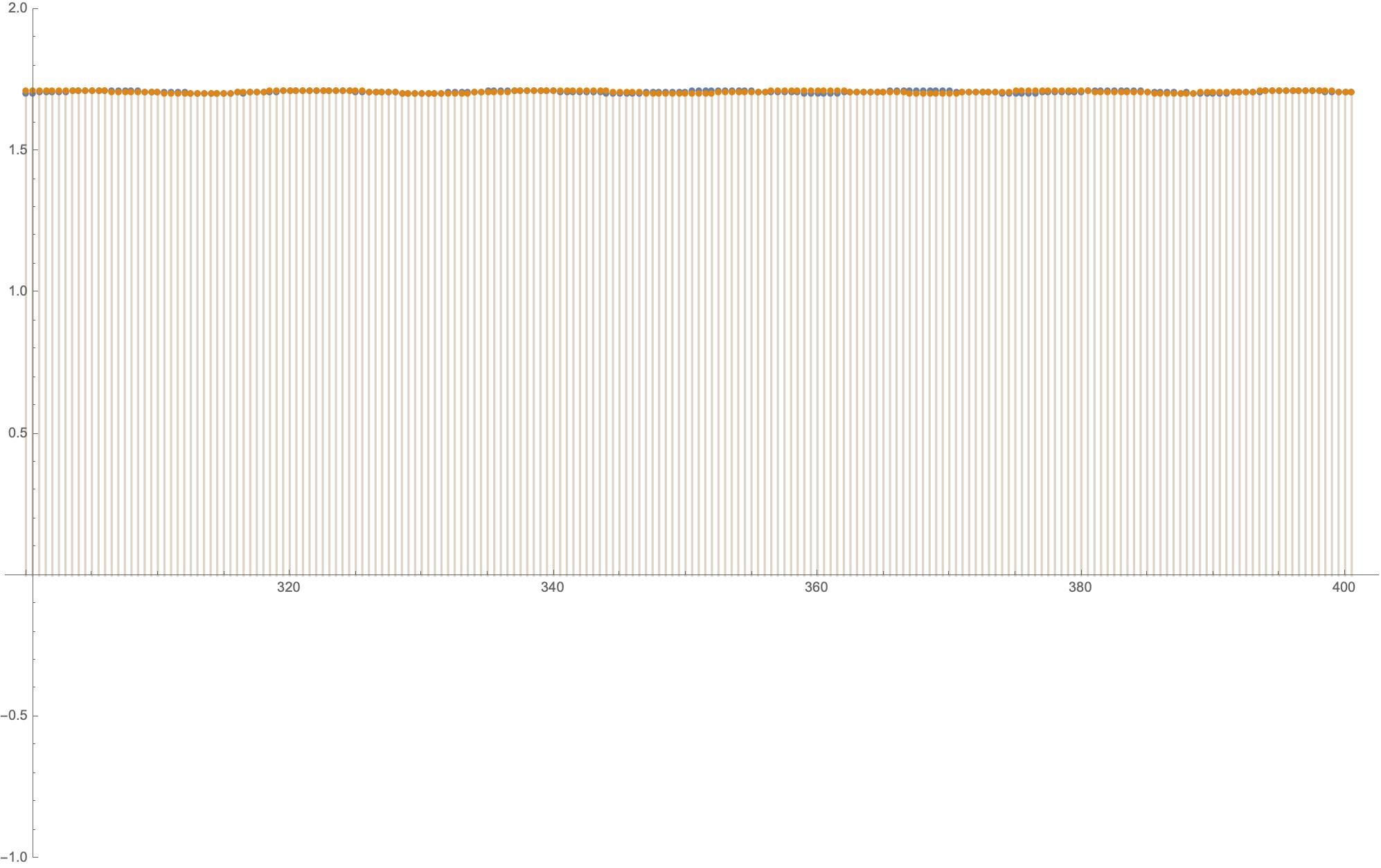}
\caption{Spins $j_1=500,j_2=400,k_1=700,k_2=600$, corner angles $\cos\theta_1=1502 \sqrt{\f{2}{8847321} }$ and $\cos\theta_2=\f{329}{2} \sqrt{\f{21}{842602} }$. Loop spin-$\ell$ ranges from $300$ to $400.5$.}
\end{subfigure}
\begin{subfigure}[t]{.45\linewidth} \includegraphics[width=1.0\textwidth]{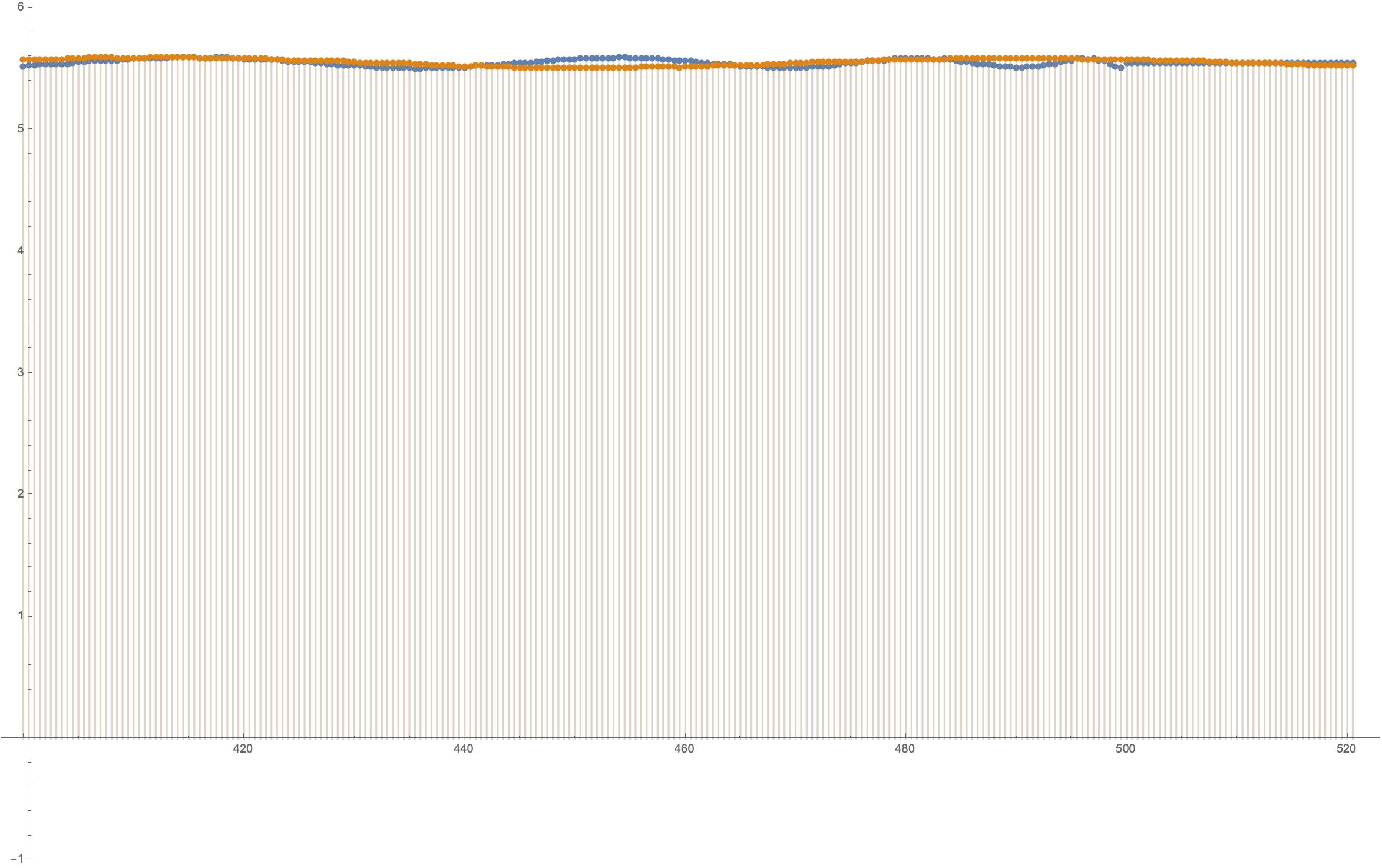}
\caption{Spins $j_1=100,j_2=120,k_1=500,k_2=510$, corner angles $\cos\theta_1=\f{50101}{30 \sqrt{2901458} }$ and $\cos\theta_2=\f{16553}{10 \sqrt{2901458} }$. Loop spin-$\ell$ ranges from $400$ to $520.5$.}
\end{subfigure}
\begin{subfigure}[t]{.45\linewidth} \includegraphics[width=1.0\textwidth]{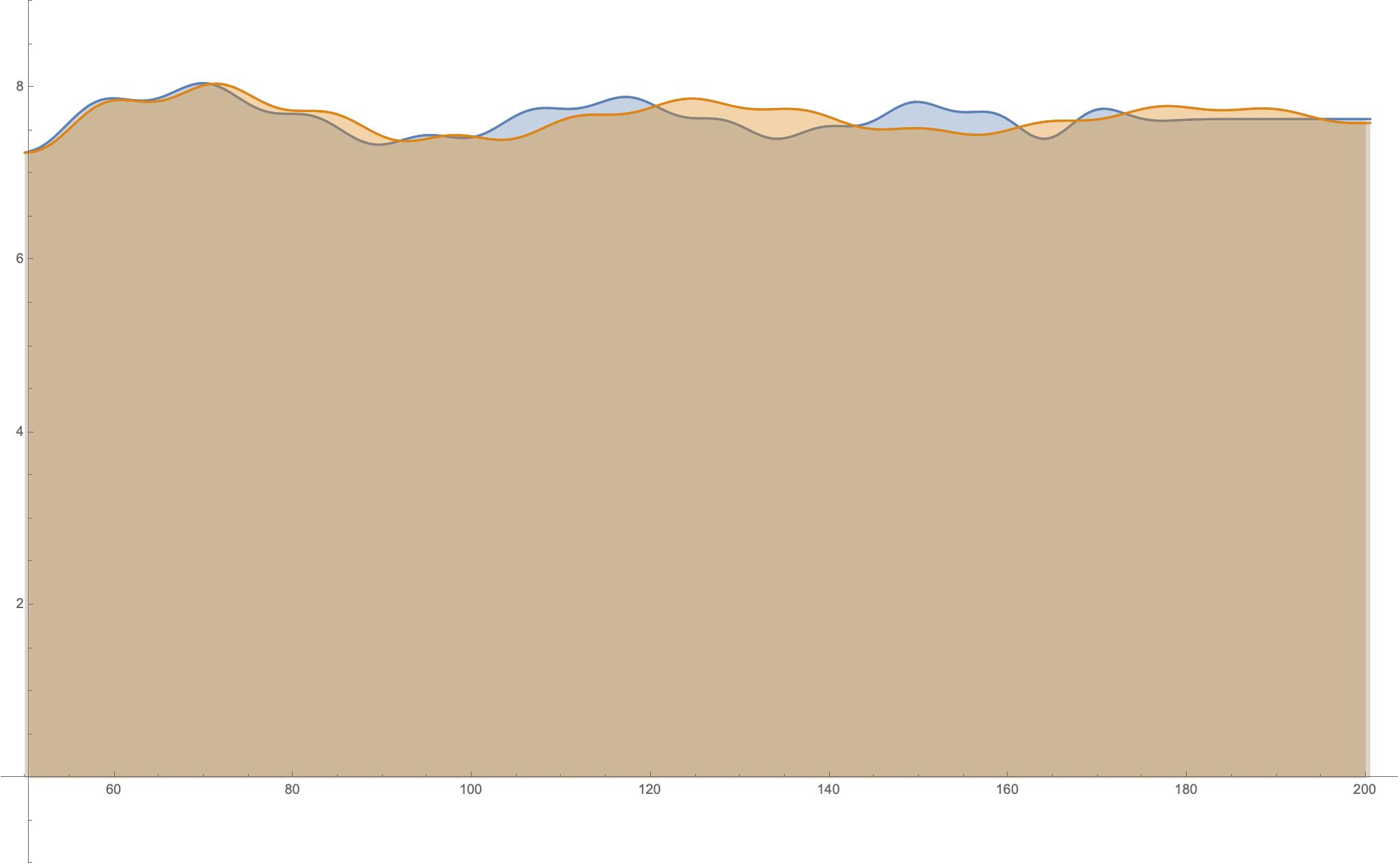}
\caption{Spins $j_1=20,j_2=30,k_1=200,k_2=190$, corner angles $\cos\theta_1=\f{7607}{4 \sqrt{3647145} }$ and $\cos\theta_2=\f{1889}{\sqrt{3647145} }$. Loop spin-$\ell$ ranges from $50$ to $200.5$.}
\end{subfigure}
\begin{subfigure}[t]{.45\linewidth} \includegraphics[width=1.0\textwidth]{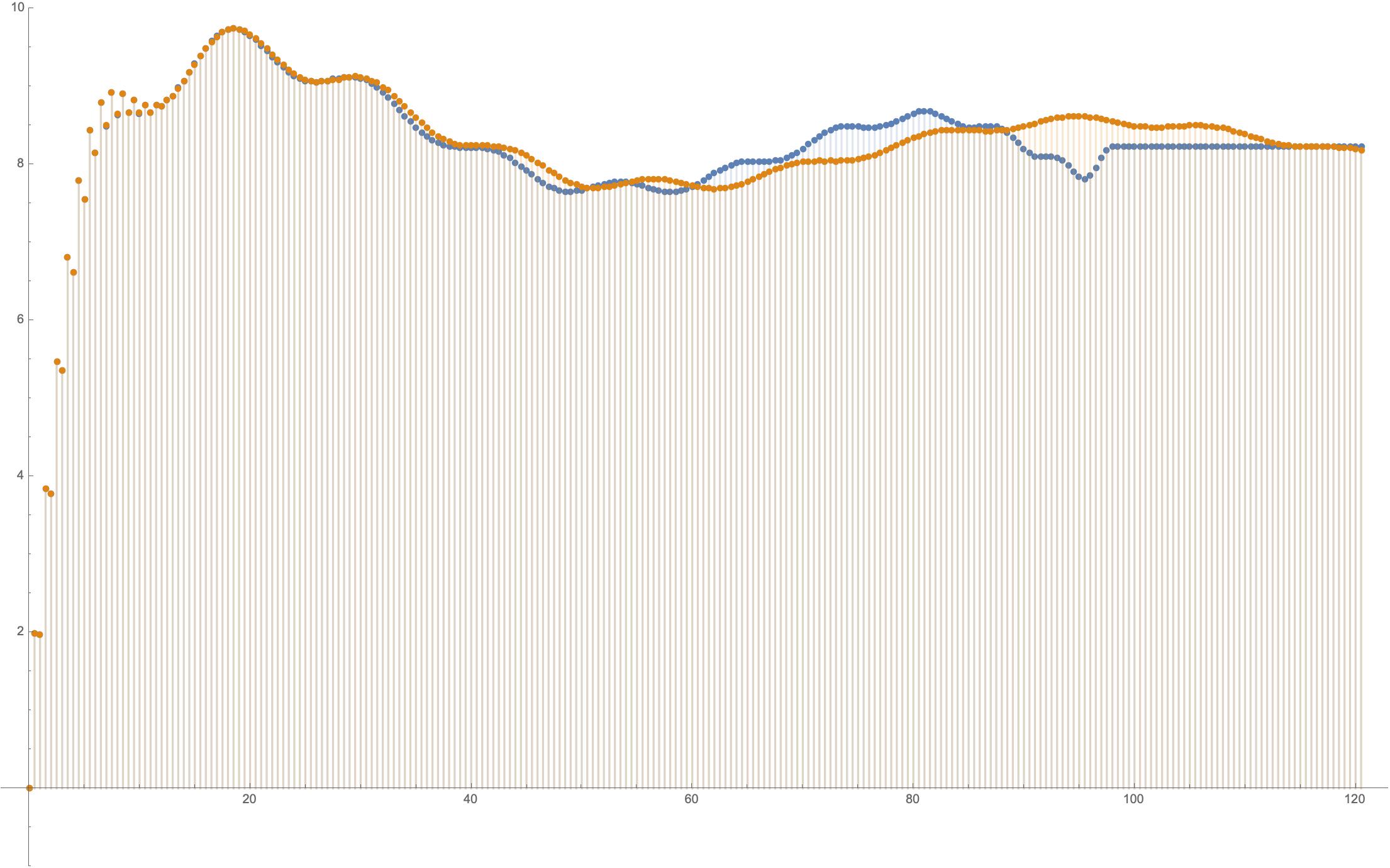}
\caption{Spins $j_1=10,j_2=14,k_1=100,k_2=98$, corner angles $\cos\theta_1=\f{1641}{35 \sqrt{2222}}$ and $\cos\theta_2=\f{2449\sqrt{2/1111}}{105}$. Loop spin-$\ell$ ranges from $0$ to $120.5$.}
\end{subfigure}
\begin{subfigure}[t]{.45\linewidth} \includegraphics[width=1.0\textwidth]{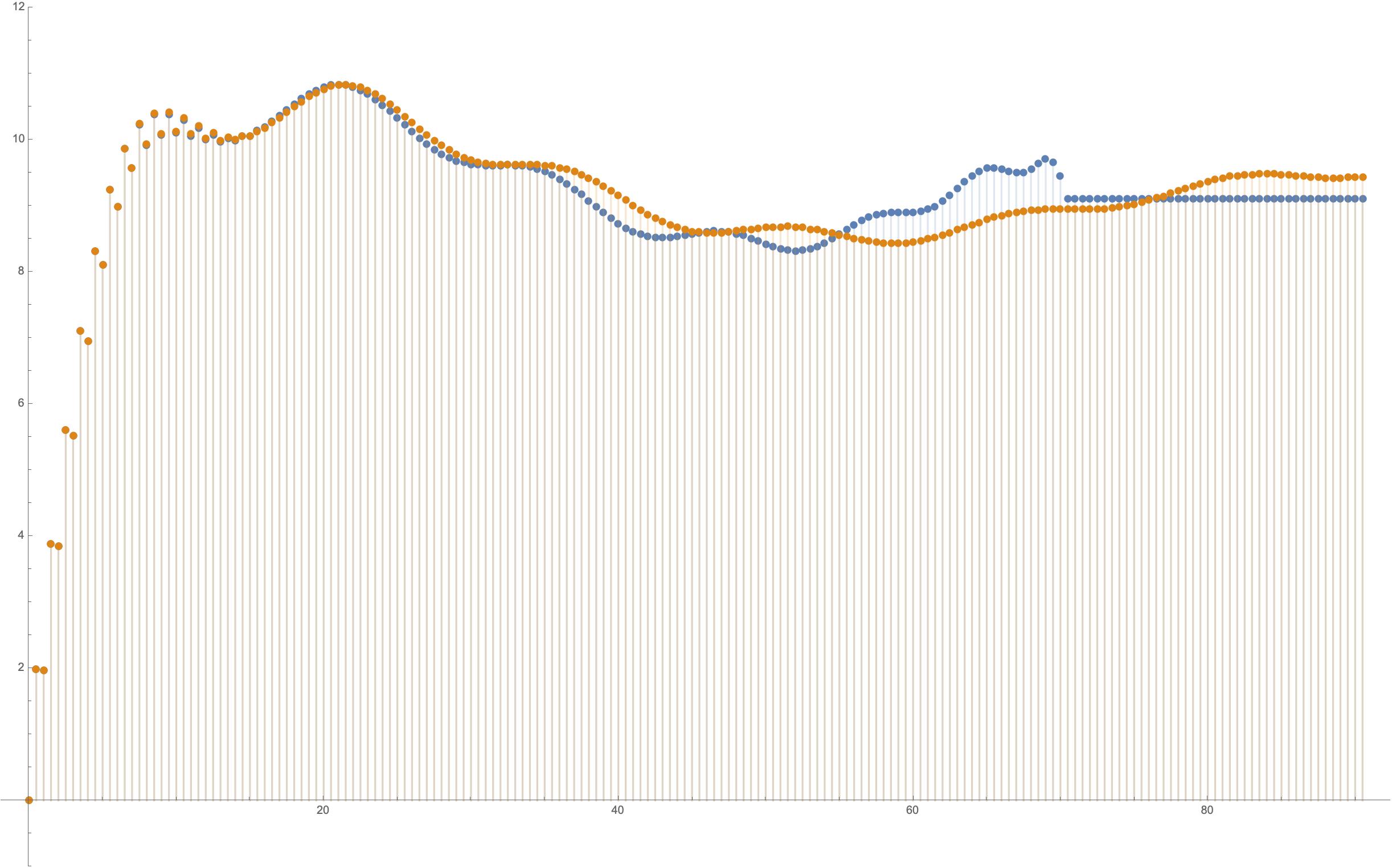}
\caption{Spins $j_1=5.5,j_2=8.5,k_1=70.5,k_2=71$, corner angles $\cos\theta_1=0.996503$ and $\cos\theta_2=0.992071$. Loop spin-$\ell$ ranges from $0$ to $90.5$.}
\end{subfigure}
\begin{subfigure}[t]{.45\linewidth} \includegraphics[width=1.0\textwidth]{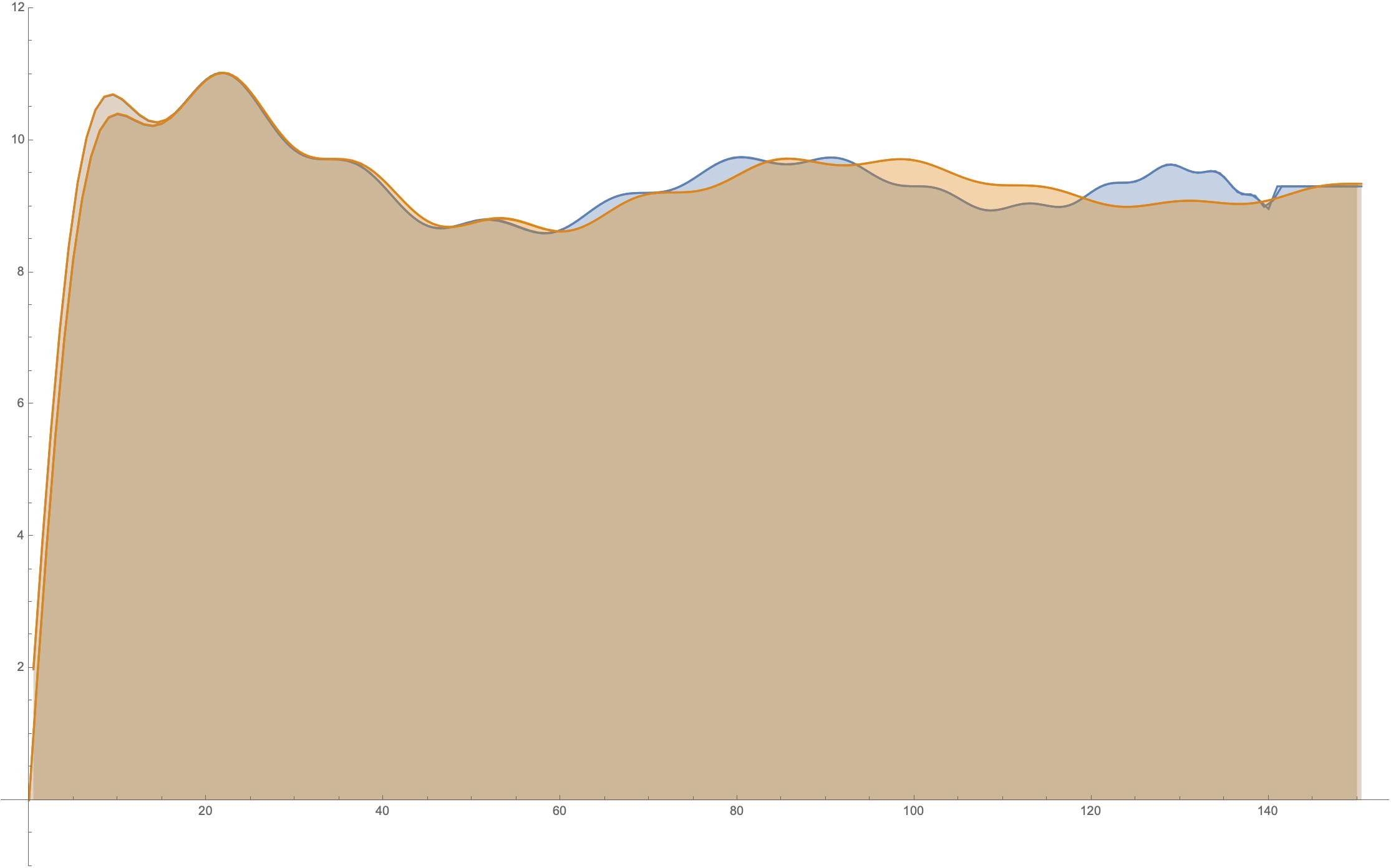}
\caption{Spins $j_1=11,j_2=17,k_1=141,k_2=142$, corner angles $\cos\theta_1=\f{773 \sqrt{13/1551}}{71}$ and $\cos\theta_2=\f{20011}{142 \sqrt{20163}}$. Loop spin-$\ell$ ranges from $0$ to $150.5$.}
\end{subfigure}
\caption{The numerical comparison between expression of $6j$-symbols (\ref{eq:IntertwinerEE-candygraph-6j}) {\color{blue}{blue}} and of Legendre polynomials (\ref{eq:IntertwinerEE-candygraph-hN},\ref{eq:IntertwinerEE-candygraph-N}) {\color{orange}{orange}} with different spins. The expression of Legendre polynomials still behaves well even $\ell$ grows.}
\label{fig:NumericalComparisonLS-candygraph}
\end{figure}

\smallskip

The approximation $S(\rho^{\theta}_{can_{A} },t)$ now allows us to study how the entanglement excitation depends on the initial state $|\{j_{i},k_{i}\}\ra$. More precisely, the angles $\theta_{1,2}$ encapsulates the relevant initial data and we would like to determine the extremal initial configurations, i.e. that maximize or minimize the leading order entanglement entropy given by the 2nd order coefficient.
%provides a smooth picture to analyze the extremal values problem about entanglement entropy excitation.
It turns out that:
\begin{itemize}
\item (i) the balanced confugurations $\theta_1=\theta_2$ always stabilize the entanglement excitation;
\item (ii) both extremally curved and totally flat geometry maximize the entanglement excitation.
\end{itemize}

Let us start with (i). We introduce the total angle $\alpha=\theta_1+\theta_2$ and relative angle $\beta=\theta_1-\theta_2$. Let us keep $\alpha$ fixed and let the entanglement vary in terms of the difference $\beta$. Differentiating\footnote{
The derivative on Legendre polynomial is given by the recursion relation:
\be
\f{\rd P_{n} (x) }{ \rd x }
=\f{n}{x^2-1}[ x P_{n}(x)-P_{n-1}(x) ]
\,. \nn
\ee
}
$P_{s}(\cos\theta_1)P_{s}(\cos\theta_2)=P_{s}(\cos\f{\alpha+\beta}{2} )P_{s}(\cos\f{\alpha-\beta}{2})$ with respect to $\beta$ gives the stationarity equation:
\beq
0
&=&
\f{ \pp }{ \pp \beta }
\Big(  P_{s}(\cos\f{\alpha+\beta}{2} )P_{s}(\cos\f{\alpha-\beta}{2}) \Big)
\nn \\
&=&
\f{s}{2}\f{ -\sin\beta }{\sin\f{\alpha+\beta}{2}\sin\f{\alpha-\beta}{2} } P_{s}(\cos\f{\alpha+\beta}{2} )P_{s}(\cos\f{\alpha-\beta}{2})
+\f{s}{2\sin\f{\alpha-\beta}{2} } P_{s}(\cos\f{\alpha+\beta}{2} )P_{s-1}(\cos\f{\alpha-\beta}{2})
\nn \\
&&-\f{s}{2\sin\f{\alpha+\beta}{2} } P_{s-1}(\cos\f{\alpha+\beta}{2} )P_{s}(\cos\f{\alpha-\beta}{2})
\,, \qquad 0 \leq \alpha, \beta \leq 2\pi
\,.
\nn
\eeq
The extremal angles $\beta=0$ and $\beta=\pi$ are both clear solutions. Similarly differentiating with respect to $\alpha$, we get stationary points $\alpha=0$ and $\alpha=\pi$ when $\beta$ is kept fixed.

The configurations $\theta_1=\theta_2=0$ and $\theta_1=\theta_2=\pi$, in the $\beta=0$ branch, give the maximal entanglement excitation.
A vanishing relative angle $\beta=0$ corresponds to equal triangle angles $\theta_1=\theta_2$, which means that the boundary spins are equal $j_{1}=j_{2}$.
This is interpreted as the flat connection configuration. Indeed, the flat constraint operator $\delta(g)=\sum_{\ell}(2\ell+1)\chi_{\ell}(g)$, imposing that the loop holonomy be trivial, annihilates initial spin network state $| \Psi_{can,\{j_i,k_i\}} \ra=| j_1,k_1,k_2 \ra_{A} \otimes | j_2,k_1,k_2 \ra_{B}$ as soon as $j_{1}\neq j_{2}$. Thus imposing the flatness of the connection imposes that $j_{1}=j_{2}$.
For such $\beta=0$ configurations, the entanglement evaluates to:
\be
\label{eq:IntertwinerEE-candygraph-sc}
\beta=0\quad\Rightarrow\quad
\f12 S(\rho^{\theta}_{can_{A} },t)
=
t^2 \sum_{s=0}^{2\ell} \, \left[ P_s (\cos\f{\alpha}{2}) \right]^2 - t^2 \left[ P_{\ell} (\cos\f{\alpha}{2}) \right]^4 + O(t^4)
\quad\textrm{with}\quad\alpha=\theta_{1}+\theta_{2}\,.
\ee
As plotted on fig.\ref{fig:Oscillation-candygraph}, the entanglement first decreases then increases but it never vanishes.
\begin{figure} \includegraphics[width=0.7\textwidth]{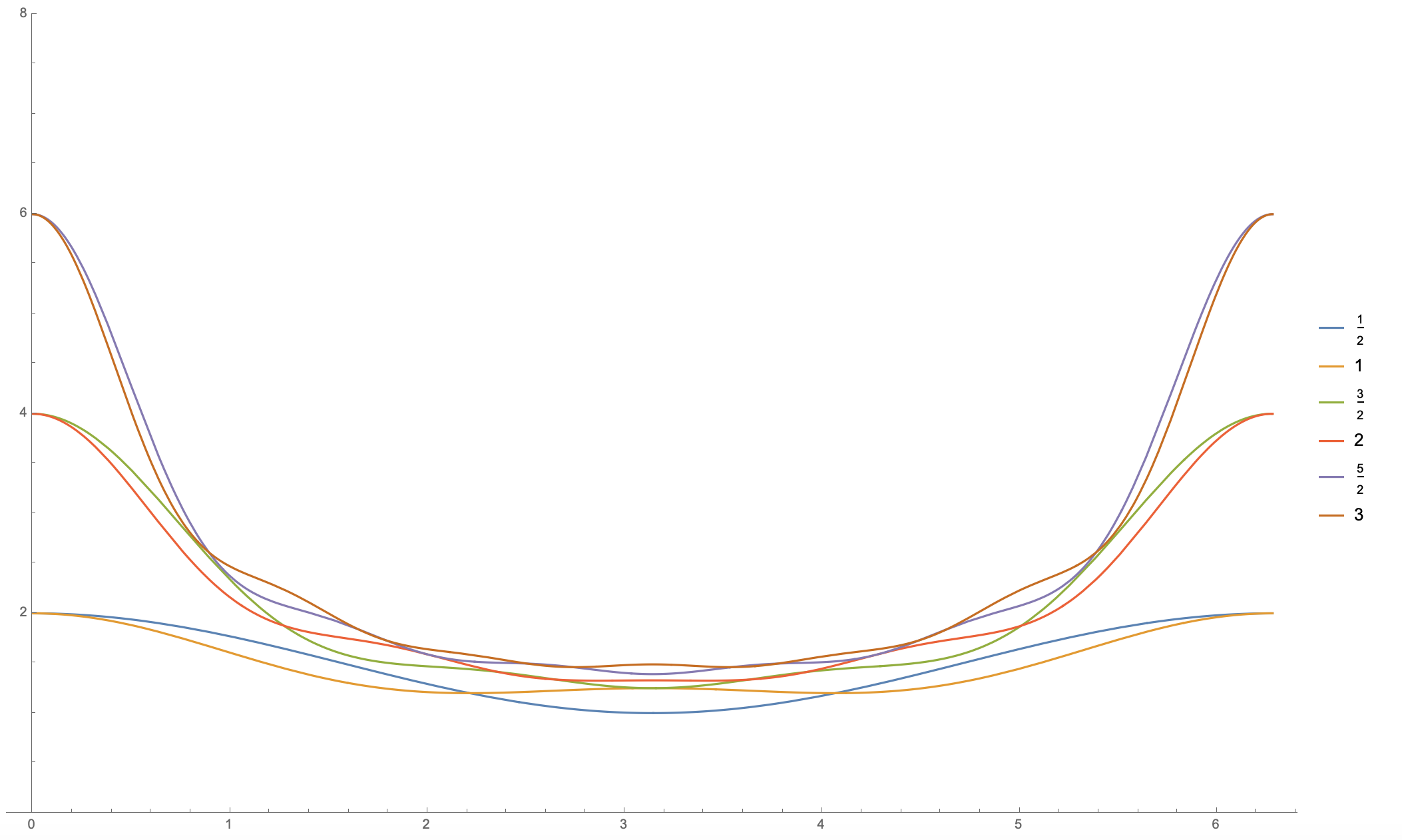}
\caption{When $\theta_1=\theta_2$, the entanglement (\ref{eq:IntertwinerEE-candygraph-sc}) for various spin values$\ell$ plotted in terms of the total angle $\alpha=\theta_1+\theta_2\,\in[0,2\pi]$.
}
\label{fig:Oscillation-candygraph}
\end{figure}
Recalling that the deficit angle given by $\delta=2\pi-\theta_1-\theta_2=2\pi-\alpha$ at the cone summit provides a measure of the bulk curvature  in fig.\ref{fig:candygraph-Triangulation}, we can interpret the optimal angle confugirations in terms of discrete bulk geometry.
%\cite{Rovelli:2014ssa}.
The two maximal points correspond to dramatically different bulk geometries. Indeed, the angular-configuration $\theta_1=\theta_2=0$ has a maximal deficit angle; it corresponds to the spin configurations $j_1,j_2 \ll k_1,k_2$ and $k_1 \approx k_2$, with extremely elongated triangles, creating an extremely  spiky bulk with maximal bulk curvature. On the other hand, the angle configuration $\theta_1=\theta_2=\pi$ has a vanishing deficit angle; it corresponds to  spin configuratiosn $j_1 \approx k_1+k_2 \approx j_2$, with flatten triangles and flat bulk geometry.

The  $\beta=\pi$
%and $\alpha=\pi$
branch lead to a minimal entanglement, with a vanishing leading order in $t^{2}$, $S(\rho^{\theta}_{can_{A} },t)={\cal O}(t^{4)}$. It corresponds to the maximal  difference between the two triangle angles $\theta_1$ and $\theta_2$.Geometrically,  this corresponds to unbalanced configuration, with $j_1\ll j_{2}$ or $j_{2}\ll j_{1}$, with one flatten triangle and one elongated triangle.

%\medskip
%
%However, approximation $S(\rho^{\theta}_{can_{A} },t)$ does not imply critical value for stability of $\ell$, because $\ell$ could contribute to $S(\rho^{\theta}_{can_{A} },t)$ arbitrarily. For instance, if $\theta_1=\theta_2=0$, then $S(\rho^{\theta}_{can_{A} },t)$ could have arbitrarily large value, which is a distinct feature compared with $S(\rho_{can_{A} },t)$. The reason may be explained by the fact that the angles are not restrictly continuous due to the gap of spins. For instance, there doesn't exist a half-integer or integer $c$ satisfying $\cos\theta=0$, since it is solved by $c=\f12[-1\pm\sqrt{1+4a(a+1)b(b+1)} ]$, while $4a(a+1)b(b+1)$ is an integer, so $c$ is neither an half-integer nor integer due to the gap.
%Therefore, $\cos\theta_1=\cos\theta_2=0$ should correspond to the sense of limit that $k_1,k_2 \to \infty$ such that the stable point with respect to $\ell$ given by (\ref{eq:2ndDGME-ClosureDefect}) is $\ell \to \infty$, which explain the arbitrarily large $S(\rho^{\theta}_{can_{A} },t)$ from approximation.
%

\medskip

Now that we have looked into the small holonomy operator spin $\ell \ll \{ j_i,k_i\}$, let us  investigate the  large spin  regime,  $\ell \sim \{ j_i,k_i\}$.  Since all the spins are assumed of the same order of magnitude, one can use the Ponzano-Regge asymptotic formula in terms of the Regge action for a  tetrahedron\cite{osti_4824659,Roberts:1998zka,Rovelli:2014ssa}:
\be
\begin{Bmatrix}
j_1 & k_1 & k_2 \\
\ell & K_2 & K_1
\end{Bmatrix}
\approx
\f{1}{ \sqrt{ 12\pi V} } \cos \left( \sum_{ab} ( j_{ab} + \f12) \theta_{ab} + \f{\pi}{4} \right)
=
\f{ 1 }{ 2\sqrt{-12 \ri \pi V} } e^{\ri S}
+
\f{ 1 }{ 2\sqrt{12 \ri \pi V} } e^{-\ri S}
\,. \label{eq:Regge-6js}
\ee
The tetrahedron's edge lengths are given by the spins $\ell,j_1,k_1,k_2,K_1,K_2$. Its volume is $V$ and $\theta_{ab}$ is the dihedral angle along the edge of length $j_{ab}$. 
The Regge action is given by $S=\sum_{ab} ( j_{ab} + \f12) \theta_{ab}$, which is understood to be the discretization of Einstein-Hilbert action (or more precisely of the Hartle-Hawking boundary term when the bulk curvature vanishes on-shell).
This allows to rewrite (\ref{eq:IntertwinerEE-candygraph-6j}) as:
\beq
\ell\in \N+\f12: \qquad
\f12 S(\rho_{can_{A} },t)
\approx
\f12 S(\rho^{(\textrm{Regge})}_{can_{A} },t)
&=&
t^2 (-1)^{\vphi} \sum_{s=0}^{2\ell} \,
\f{ \cos (S_{A}[s]-S_{B}[s]) - \sin( S_{A}[s]+S_{B}[s] ) }{ 24\pi\sqrt{V_A[s] V_B[s]} }
  \prod_{i=1}^{2}(2k_i+1)
  \nn
  \\
  &&+ O(t^4)
  \,,
\label{eq:IntertwinerEE-candygraph-hN-Regge}
\\
\ell\in\N: \qquad
\f12 S(\rho_{can_{A} },t)
\approx
\f12 S(\rho^{(\textrm{Regge})}_{can_{A} },t)
&=&
t^2 (-1)^{\vphi} \sum_{s=0}^{2\ell} \,
\f{ \cos (S_{A}[s]-S_{B}[s]) - \sin( S_{A}[s]+S_{B}[s] ) }{ 24\pi\sqrt{V_A[s] V_B[s]} }
  \prod_{i=1}^{2}(2k_i+1)
 \nn
 \\
&&-t^2
\f{ [ 1-\sin (2S_A[\ell]) ][ 1-\sin (2S_B[\ell]) ]  }{ 576 \pi^2 V_A[\ell] V_B[\ell]}
  \prod_{i=1}^{2}(2k_i+1)^2
 + O(t^4)
\,. \label{eq:IntertwinerEE-candygraph-N-Regge}
\eeq
The Regge actions $S_A[s]$ and $S_B[s]$ corresponds respectively to the two tetrahedra dual to the vertices $A$ and $B$, with edge lengths $(j_i,k_1,k_2,s,k_2,k_1)$ respectivelk with $i=1$ and $i=2$.
The  phase $\vphi=j_1+j_2+2k_1+2k_2$ and face amplitude factor $(2k_1+1)(2k_2+1)$ both depend on the bulk spins.

\medskip

Let us conclude this section with a brief summary of our results on the example of the candy graph.
The candy graph is the simplest graph with a non-trivial bulk. It is a special case of the 2-vertex graph \cite{Borja:2010gn,Borja:2010rc,Livine:2011up,Aranguren:2022nzn} where the two vertices are linked by a couple of edges forming a loop in the bulk. This configuration allows to define non-trivial dynamical operators acting on the bulk geometry, for instance the loop holonomy operator, and to study the entanglement between the two bulk vertices generated by such dynamics.

In this context, we have explicitly computed the entanglement excitation created by the loop holoonomy operator acting on a pure spin network basis state at leading order in time, that is in $t^2$. We have considered and compared two measures of entanglement: the geometric entanglement studied in the the previous section and the  bipartite  entanglement entropy  shared between the two vertices (given by the linear entropy of the reduced density matrix). We have shown that these two notions of entanglement match exactly a leading order, thus leading to a consistent picture of the entanglement excitation on the spin network state.
Furthermore, in the semi-classical regime at large spins where spin networks can be provided with a discrete geometry interpretation in terms of dual triangulations, we have identified the initial configurations that optimize the excitation of entanglement by the loop operator as the (discrete) geometries with either maximal curvature or vanishing curvature.
%For the scenario of transition from discretized geometry to smooth geometry, we have seen that two configurations could maximize the entanglement excitation: (a) high curved bulk; (b) flat bulk. We will revisit the number of maximal value points when comes to triangle graph.
%Finally, we note that in the discretized geometry, we need $4$ spin network spins $j_1,j_2,k_1,k_2$ plus loop holonomy spin-$\ell$ to describe the entanglement excitation, while in the smooth geometry, we only need $2$ corner angles $\theta_1,\theta_2$ plus loop holonomy spin-$\ell$ to describe the process, where the rescale factors are missing.
%
%{\bf{Could we understand the corner angles from the view of renormalization?}}

%%%%%%

%%%
\section{Triangle graph: multipartite entanglement} \label{section:TriangleGraph}
%%%

In this section, we wish to provide the reader with an example beyond the bipartite entanglement. We introduce the triangle graph, i.e. a triangle made with three vertices and with three boundary edges poking out from those vertices. The bulk triangle forms a loop on which we can act with the holonomy operator. This configuration allows to study both bipartite and tripartite entanglements by partionning the triplet of vertices, and to compare them with the geometric entanglement.

%This section looks at an example of multipartite entanglement on triangle graph. Again, we consider a truncation dynamics generated by loop holonomy operator, and investigate bipartite entanglement excitation from the tripartite system. On the other hand, the geometric measure of entanglement offers a straightforward generalization from the candy graph. The main objective is to compare the bipartite entanglement with the geometric measure of entanglement.

%%%%%%

%%%
\subsection{Bipartite entanglement on triangle graph}
%%%

The triangle graph, as drawn on fig.\ref{fig:trianglegraph}, is the direct extension of the candy graph to a bulk made of three vertices arranged around a single loop.
\begin{figure}[htb!]
\begin{tikzpicture}[scale=1]

\coordinate (A) at (0,0);

\draw[thick] (A) ++(60:1.2) coordinate(B3) node[scale=0.7,blue] {$\bullet$} node[right]{$C$} --++ (60:1) node[right] {$j_3$};
\draw[thick] (A) ++(180:1.2) coordinate(B1) node[scale=0.7,blue] {$\bullet$} node[above]{$A$}--++ (180:1) node[above] {$j_1$};
\draw[thick] (A) ++(300:1.2) coordinate(B2) node[right]{$B$} --++ (300:1) node[right] {$j_2$};

\draw[blue,thick] (B1) -- node[below] {$k_3$} (B2) node[scale=0.7,blue] {$\bullet$} -- node[right] {$k_1$} (B3) node[scale=0.7,blue] {$\bullet$} --node[above=2,midway] {$k_2$} (B1) node[scale=0.7,blue] {$\bullet$};

\path (B1) ++ (0:0.3) coordinate(C1);
\path (B2) ++ (120:0.3) coordinate(C2);
\path (B3) ++ (240:0.3) coordinate(C3);

\draw [line width=1pt,red,->,>=latex,rounded corners] (C1) -- (C2) -- (C3) -- (C1);

\draw[->,>=stealth,very thick] (1.75,0) -- (2.75,0);

\coordinate (D) at (7,0);

\draw[thick] (D) ++(60:1.2) coordinate(E3) node[scale=0.7,blue] {$\bullet$} node[right]{$C$} --++ (60:1) node[right] {$j_3$};
\draw[thick] (D) ++(180:1.2) coordinate(E1) node[scale=0.7,blue] {$\bullet$} node[above]{$A$}--++ (180:1) node[above] {$j_1$};
\draw[thick] (D) ++(300:1.2) coordinate(E2) node[right]{$B$} --++ (300:1) node[right] {$j_2$};

\draw[thick] (E1) -- ++ (180:1) node[above] {$j_1$} ++ (180:1) node {$\displaystyle{\sum_{K_1,K_2,K_3} }$};

\draw[blue,thick] (E1) -- node[below] {$K_3$} (E2) node[scale=0.7,blue] {$\bullet$} -- node[right] {$K_1$} (E3) node[scale=0.7,blue] {$\bullet$} --node[above=2,midway] {$K_2$} (E1) node[scale=0.7,blue] {$\bullet$};

\end{tikzpicture}
\caption{The action of loop holonomy operator on triangle graph: it endows spin-superposition along the bulk edges.
}
\label{fig:trianglegraph}
\end{figure}
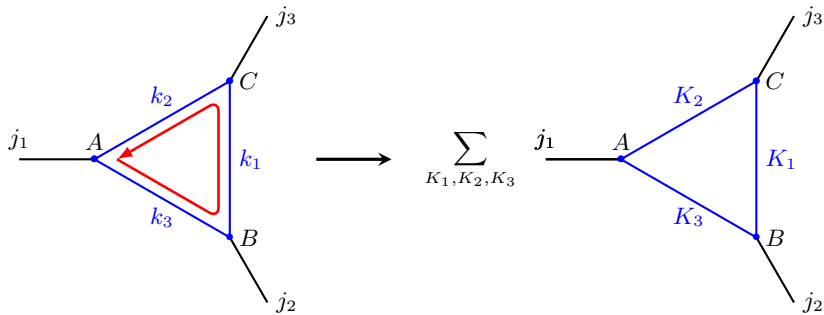
The spin network Hilbert space on this triangle graph consists in the tensor product of three 3-valent intertwiner spaces, similarly to the candy graph,
\be
\cH_{tri}=\bigotimes_{v\in tri} \cH_{v}
=
\cH_A \otimes \cH_B \otimes \cH_C
\,, \qquad
\cH_{A}
=
\bigoplus_{K_2,K_3}
\textrm{Inv}_{\SU(2)} \Big( \cV_{j_1} \otimes \cV_{K_2} \otimes \cV_{K_3} \Big)
\,,
\nn
\ee
and similarly for $\cH_B$ and $\cH_C$.

We wish to study the action of the loop holonomy operator  on an arbitrary spin network basis state chosen as our initial state,
\beq
| \Psi_{ tri, \{ j_i,k_i\} } \ra
=
| j_1, k_2, k_3 \ra_A \otimes | j_2, k_3, k_1 \ra_B \otimes | j_3, k_1, k_2 \ra_C
\quad \in
\cH_A \otimes \cH_B \otimes \cH_C
\,. \nn
\eeq
The loop holonomy operator affects the bulk spins $\{ K_i \}$, but not boundary spins $\{ j_i\}$,
\beq
\wh{\chi_{\ell} }
:&&
\textrm{Inv}_{\SU(2)} \Big( \cV_{j_1} \otimes \cV_{k_2} \otimes \cV_{k_3} \Big)
\otimes
\textrm{Inv}_{\SU(2)} \Big( \cV_{j_2} \otimes \cV_{k_3} \otimes \cV_{k_1} \Big)
\otimes
\textrm{Inv}_{\SU(2)} \Big( \cV_{j_3} \otimes \cV_{k_1} \otimes \cV_{k_2} \Big)
\nn \\
&&\to
\bigoplus_{K_1,K_2,K_3}
\textrm{Inv}_{\SU(2)} \Big( \cV_{j_1} \otimes \cV_{K_2} \otimes \cV_{K_3} \Big)
\otimes
\textrm{Inv}_{\SU(2)} \Big( \cV_{j_2} \otimes \cV_{K_3} \otimes \cV_{K_1} \Big)
\otimes
\textrm{Inv}_{\SU(2)} \Big( \cV_{j_3} \otimes \cV_{K_1} \otimes \cV_{K_2} \Big)
\,.
\nn
\eeq
Following the same logic as with the candy graph, we compute the evolution of the state, truncated to leading order and properly normalized,
\be
| \Psi_{tri, \{j_i,k_i\}} (\ell,t) \ra
=
\f{ | \Psi(tri)_{\{j_i,k_i\}} \ra - \ri t \, \sum_{ \{K_i\} }
\, [ Z(tri)_{\ell}^{ \{j_i\} } ]^{\{ K_i \} }_{ \phantom{ \{K_i\} }\{ k_i \} }
\, | \Psi(tri)_{\{j_i,K_i\}} \ra }
{ \sqrt{ 1+t^2 \sum_{s=0}^{2\ell} \, [ Z(tri)_{s}^{ \{j_i\} } ]^{\{ k_i \} }_{ \phantom{ \{k_i\} }\{ k_i \} } }   }
\,. \label{eq:NormalizedShortTimeState-TriangleGraph}
\ee
The transition matrix $Z$ is expressed in terms of $6j$-symbols according to the general formula  (\ref{eq:Composition-Amplitudes}) for the action of the loop operator  $\wh{\chi_{\ell} }$,
\be
\label{eq:State-shorttime-trianglegraph}
[ Z(tri)_{\ell}^{ \{j_i\} } ]^{ \{K_i\} }_{ \phantom{ \{K_i\} }\{ k_i \} }
=
(-1)^{\sum_{i=1}^{3}( j_i+k_i+K_i + \ell)} \begin{Bmatrix}
   j_1 & k_2 & k_3 \\
   \ell &  K_3 & K_2
  \end{Bmatrix}\begin{Bmatrix}
   j_2 & k_3 & k_1 \\
   \ell &  K_1 & K_3
  \end{Bmatrix} \begin{Bmatrix}
   j_3 & k_1 & k_2 \\
   \ell & K_2 & K_1
  \end{Bmatrix} \, \prod_{i=1}^{3} \sqrt{(2k_i+1)(2K_i+1)}
\,.
\ee
The normalization factor can then be evaluated to
\beq \label{eq:NormalizationFactor-trianglegraph}
N_{tri,\{j_i,k_i\}} (\ell,t)
&=&
1+ t^2 \sum_{ s=0}^{2\ell}
(-1)^{j_{1}+j_{2}+j_3+2k_1+2k_2+2k_3+3s}
\begin{Bmatrix}
   j_1 & k_2 & k_3 \\
   s &  k_3 & k_2
  \end{Bmatrix}
\begin{Bmatrix}
   j_{2} & k_3 & k_1 \\
   s &  k_1 & k_3
  \end{Bmatrix}
  \begin{Bmatrix}
   j_{3} & k_1 & k_2 \\
   s &  k_2 & k_1
  \end{Bmatrix}
  \prod_{i=1}^{3}(2k_i+1)
\\
&=&
1+
t^2\sum_{ K_1,K_2,K_3}
  \begin{Bmatrix}
   j_1 & k_2 & k_3 \\
   \ell &  K_3 & K_2
  \end{Bmatrix}^2
\begin{Bmatrix}
   j_{2} & k_3 & k_1 \\
   \ell &  K_1 & K_3
  \end{Bmatrix}^2
  \begin{Bmatrix}
   j_{3} & k_1 & k_2 \\
   \ell &  K_2 & K_1
  \end{Bmatrix}^2
  \prod_{i=1}^{3}(2k_i+1)(2K_i+1)
  \,.
  \label{eq:NormalizationFactor-trianglegraph-spinshift}
\eeq
As for candy graph, this normalization factor is to be interpreted as the density of spin network basis states after the action of loop holonomy operator.

The initial spin network state is a basis state, thus is fully separable with vanishing entanglement. The final state is given by the following density matrix:
\beq \label{eq:DensityMatrix-shorttime-trianglegraph}
&&
\rho_{tri_{ABC} }(t)
=| \Psi_{tri,\{j_i,k_i\}} (\ell,t) \ra\la \Psi_{can,\{j_i,k_i\}} (\ell,t) |
\quad \in \textrm{End} (\cH_{A} \otimes \cH_{B} \otimes \cH_{C} )
\\
&=&
\f{1}{ 1+t^2 \sum_{s=0}^{2\ell} [ Z(tri)_{s}^{ \{j_i\} } ]^{ \{k_i\} }_{ \phantom{ \{k_i\} }\{ k_i \} } }
\bigg(
| j_1,k_2,k_3 \ra \la j_1,k_2,k_3 |_{A} \otimes | j_2,k_3,k_1 \ra\la j_2,k_3,k_1 |_{B} \otimes | j_3,k_1,k_2 \ra\la j_3,k_1,k_2 |_{C}
\nn \\
&&
+ t^2 \sum_{ \{K_i, K'_i \} }
[ Z(tri)_{\ell}^{ \{j_i\} } ]^{ \{K_i\} }_{ \phantom{ \{K_i\} }\{ k_i \} }
[ Z(tri)_{\ell}^{ \{j_i\} } ]^{ \{K'_i\} }_{ \phantom{ \{K'_i\} }\{ k_i \} }
| j_1,K_2,K_3 \ra \la j_1,K'_2,K'_3 |_{A} \otimes | j_2,K_3,K_1 \ra\la j_2,K'_3,K'_1 |_{B}
\nn \\
&&\otimes | j_3,K_1,K_2 \ra\la j_3,K'_1,K'_2 |_{C}
\nn \\
&&
- \ri t
\sum_{ \{K_i\} }
[ Z(tri)_{\ell}^{ \{j_i\} } ]^{ \{K_i\} }_{ \phantom{ \{K_i\} }\{ k_i \} }
| j_1,K_2,K_3 \ra\la j_1,k_2,k_3 |_{A} \otimes | j_2,K_3,K_1 \ra\la j_2,k_3,k_1 |_{B} \otimes | j_3,K_1,K_2 \ra\la j_3,k_1,k_2 |_{C}
\nn \\
&&
+ \ri t
\sum_{ \{K_i\} }
[ Z(tri)_{\ell}^{ \{j_i\} } ]^{ \{K_i\} }_{ \phantom{ \{K_i\} }\{ k_i \} }
| j_1,k_2,k_3 \ra\la j_1,K_2,K_3 |_{A} \otimes | j_2,k_3,k_1 \ra\la j_2,K_3,K_1 |_{B} \otimes | j_3,k_1,k_2 \ra\la j_3,K_1,K_2 |_{C}
\bigg)
\,.
\nn
\eeq
Now there are three possible bipartition of the triple tensor product $\cH_{A} \otimes \cH_{B} \otimes \cH_{C}$. For the sake of convenience, let us look at the reduced density matrices for $\cH_{A}$, $\cH_{B}$ and $\cH_{C}$ respectively, moreover, we only need to write down $\rho_{tri_{A} }(t)$ and the other two cases are similar.
Via tracing over $\cH_{BC}$, the reduced density matrix is read
\beq
\rho_{tri_{A} }(t)
=
\tr_{BC} \rho_{tri_{ABC} }(t)
&=&
\f{1}{ 1+t^2 \sum_{s=0}^{2\ell} [ Z(tri)_{s}^{ \{j_i\} } ]^{ \{K_i\} }_{ \phantom{ \{K_i\} }\{ k_i \} } }
\bigg(
| j_1,k_2,k_3 \ra \la j_1,k_2,k_3 |_{A}
\\
&&+ t^2 \sum_{K_1, K_2,K_3}
\Big( [ Z(tri)_{\ell}^{ \{j_i\} } ]^{ K_1 K_2 K_3 }_{ \phantom{ K_1 K_2 K_3 }\{ k_i \} } \Big)^2
| j_1,K_2,K_3 \ra \la j_1,K_2,K_3 |_{A}
\bigg)
\quad \in \textrm{End} (\cH_{A} )
\,.
\nn
\eeq
This is a diagonal matrix with respect to basis $| j_1,K_2,K_3 \ra \in \cH_{A}$, thus the eigenvalues of $\rho_{tri_{A} }(t)$ can be directly read off from the expression above:
\be
\lambda_{ \rho_{tri_A } }[K_2,K_3]
=
\f{   \delta^{K_2}_{k_2} \, \delta^{K_3}_{k_3}
+ t^2 \sum_{K_1=| k_1-\ell |}^{k_1+\ell} 
\begin{Bmatrix}
   j_1 & k_2 & k_3 \\
   \ell &  K_3 & K_2
  \end{Bmatrix}^2
  \begin{Bmatrix}
   j_2 & k_3 & k_1 \\
   \ell &  K_1 & K_3
  \end{Bmatrix}^2
  \begin{Bmatrix}
   j_3 & k_1 & k_2 \\
   \ell & K_2 & K_1
  \end{Bmatrix}^2 \, \prod_{i=1}^{3} (2K_i+1)(2k_i+1)
 }
{ 1+t^2 \sum_{s=0}^{2\ell} (-1)^{\sum_{i=1}^{3}( j_i+k_i+K_i + s) } \begin{Bmatrix}
   j_1 & k_2 & k_3 \\
   s &  k_3 & k_2
  \end{Bmatrix}\begin{Bmatrix}
   j_2 & k_3 & k_1 \\
   s &  k_1 & k_3
  \end{Bmatrix} \begin{Bmatrix}
   j_3 & k_1 & k_2 \\
   s & k_2 & k_1
  \end{Bmatrix} \, \prod_{i=1}^{3} (2k_i+1)  }
\,. \label{eq:eigenvalues-shortime-trianglegraph}
\ee
The eigenvalues are labeled by the two spins $K_2$ and $K_3$. This gives us the bipartite entanglement entropy of $\cH_{A} | \cH_{B} \otimes \cH_{C}$ up to  second order in $t$:
\beq
&&
\f12 S(\rho_{tri_A })
=
{{\f12
\left(1
-\sum_{K_2,K_3}  \lambda_{ \rho_{tri_A } }[K_2,K_3] ^2 \right)
}}
\\
&=&
t^2 \bigg(
\sum_{ s=0}^{2\ell}
(-1)^{j_{1}+j_{2}+j_3+2k_1+2k_2+2k_3+3s}
\begin{Bmatrix}
   j_1 & k_2 & k_3 \\
   s &  k_3 & k_2
  \end{Bmatrix}
\begin{Bmatrix}
   j_{2} & k_3 & k_1 \\
   s &  k_1 & k_3
  \end{Bmatrix}
  \begin{Bmatrix}
   j_{3} & k_1 & k_2 \\
   s &  k_2 & k_1
  \end{Bmatrix}
  \prod_{i=1}^{3}(2k_i+1)
  \nn
  \\
&&
-
\sum_{K_1=| k_1-\ell |}^{k_1+\ell}
\begin{Bmatrix}
   j_1 & k_2 & k_3 \\
   \ell &  k_3 & k_2
  \end{Bmatrix}^2
  \begin{Bmatrix}
   j_2 & k_3 & k_1 \\
   \ell &  K_1 & k_3
  \end{Bmatrix}^2
  \begin{Bmatrix}
   j_3 & k_1 & k_2 \\
   \ell & k_2 & K_1
  \end{Bmatrix}^2 \, (2K_1+1)(2k_1+1)(2k_2+1)^2(2k_3+1)^2 \bigg)
+O(t^4)
\,.\nn
\eeq
%Similar to the case of candy graph, note that only when $K_2=k_2, \ K_3=k_3$, the logarithm function has nontrivial $O(t^2)$-contribution.
As for the candy graph, no 1st-order of $t$ appears at all and the leading order is directly in $t^{2}$.

One gets the linear entropies for the vertices $B$ and $C$ by a straightforward cyclic permutation of the labels.
%Similarly, we can write down other bipartite entanglement entropy, $S(\rho_{tri_B })$ and $S(\rho_{tri_C })$, via replacing the $k_2 \to K_2\,, \ K_1 \to k_1$ and $k_3 \to K_3\,, \ K_1 \to k_1$ respectively.
%
Generally, the three bipartite entanglement entropies have different values due to the a priori different values of the initial spins. Nonetheless, when $\ell \in \N +\f12$, all three bipartite entanglement entropies are equal to each other at the leading order in $t^2$. We trace this back  to the triangle condition on the $\{6j\}$-symbols, which eliminates the second term in the entropy formula above. Hence we conclude
\begin{res}
For an odd holonomy spin $\ell\in N+\f12$, the bipartite entanglement entropies on triangle graph satisfy $S(\rho_{tri_A })=S(\rho_{tri_B })=S(\rho_{tri_C })$ at leading order in $t^2$. This common leading order term is given
\be
%\ell \in \N+\f12\,, \qquad
{{\f12}} S(\rho_{tri_A })
\sim
t^2 \bigg(
\sum_{ s=0}^{2\ell}
(-1)^{j_{1}+j_{2}+j_3+2k_1+2k_2+2k_3+3s}
\begin{Bmatrix}
   j_1 & k_2 & k_3 \\
   s &  k_3 & k_2
  \end{Bmatrix}
\begin{Bmatrix}
   j_{2} & k_3 & k_1 \\
   s &  k_1 & k_3
  \end{Bmatrix}
  \begin{Bmatrix}
   j_{3} & k_1 & k_2 \\
   s &  k_2 & k_1
  \end{Bmatrix}
  \prod_{i=1}^{3}(2k_i+1)
  \bigg)
  \,.
\ee
\end{res}
Those triangle inequalities further implies that that $S(\rho_{tri_A })=S(\rho_{tri_B })=S(\rho_{tri_C })$ reach the same plateau value when the holonomy operator spin is larger than its critical value, $\ell > \min\{ 2k_1,2k_2,2k_3 \}$, no matter whenever $\ell$ is odd or even.

%\smallskip

We now compare the bipartite entanglement to the tripartite entanglement, defined as the geometric entanglement.

%%%%%%

%%%
\subsection{Tripartite entanglement}
%%%

In order to quantify the multipartite entanglement on the triangle graph, we use the geometric  entanglement. 
We can compute the geometric measure of entanglement as explained previously: the geometric entanglement is given by the maximal projection of the evolving state onto the spin network basis, which is, at very early times $t$, its projection onto the initial spin network basis state, that is by the contribution $K_{1,2,3}=k_{1,2,3}$:
%based on the truncation of dynamics, the final state has the maximal projection on the initial spin network basis state, since $t \to 0$. This maximal projection is the entanglement eigenvalue from the definition of geometric measure of entanglement, which is given by remembering (\ref{eq:NormalizedShortTimeState-TriangleGraph})
\be
\lambda_{\max}
=
\big{|} \la \Psi_{tri, \{j_i,k_i\}} (\ell,0) | \Psi_{tri, \{j_i,k_i\}} (\ell,t) \ra \big{|}^2
=\f{ 1 + t^2 \Big( [ Z(tri)_{\ell}^{ \{j_i\} } ]^{\{ k_i \} }_{ \phantom{ \{K_i\} }\{ k_i \} } \Big)^2  }
{ 1+t^2 \sum_{s=0}^{2\ell} \, [ Z(tri)_{s}^{ \{j_i\} } ]^{\{ k_i \} }_{ \phantom{ \{k_i\} }\{ k_i \} }   }
\,.
\ee
Hence it leads to the geometric measure of entanglement on triangle graph:
\beq
S_{g} [ \Psi_{tri, \{j_i,k_i\}} (\ell,t) ]
&=&
 t^2 \prod_{i=1}^{3}(2k_i+1) \sum_{s=0}^{2\ell} \, (-1)^{\sum_{i=1}^{3}( j_i+2k_i + s)} \begin{Bmatrix}
   j_1 & k_2 & k_3 \\
   s &  k_3 & k_2
  \end{Bmatrix}\begin{Bmatrix}
   j_2 & k_3 & k_1 \\
   s &  k_1 & k_3
  \end{Bmatrix} \begin{Bmatrix}
   j_3 & k_1 & k_2 \\
   s & k_2 & k_1
  \end{Bmatrix}
\nn \\
&&-t^2
  \prod_{i=1}^{3}(2k_i+1)^2
  \begin{Bmatrix}
   j_1 & k_2 & k_3 \\
   \ell &  k_3 & k_2
  \end{Bmatrix}^2 \begin{Bmatrix}
   j_2 & k_3 & k_1 \\
   \ell &  k_1 & k_3
  \end{Bmatrix}^2 \begin{Bmatrix}
   j_3 & k_1 & k_2 \\
   \ell & k_2 & k_1
  \end{Bmatrix}^2
 + O(t^4)
\,.
\label{eq:IntertwinerEE-trianglegraph-6j}
\eeq
This is a straightforward extension of the formula (\ref{eq:IntertwinerEE-candygraph-6j}) derived for the candy graph, where we've added the relevant $\{6j\}$-symbols.
% The method of generalization is simple: put the blocks of $6j$-symbol together.
%
As earlier, the geometric entanglement excitation reaches a plateau, due to the triangle inequalities on the spins, when the holonomy operator spin grows beyond a critical value, $\ell > \min\{ 2k_1,2k_2,2k_3\}$.
Moreover, when the spin is odd, $\ell \in \N+\f12$, the second term in (\ref{eq:IntertwinerEE-trianglegraph-6j}) vanishes, as for the bipartite entanglement in last subsection.

The plots \ref{fig:EE-GME-trianglegraph} compare the bipartite entanglement entropies and geometric entanglement:
\begin{itemize}
\item (i) they are equal when the holonomy spin is odd, $\ell \in \N+\f12$;
\item (ii) for small spins $\ell$, those measures of entanglement  are clearly distinct, but they tend to converge as the spin $\ell$ increases, and they are eventually constant and equal  beyond the critical value $\ell > \min\{2k_1,2k_2,2k_3\}$.
\end{itemize}
\begin{figure}
\begin{subfigure}[t]{.45\linewidth} \includegraphics[width=1.0\textwidth]{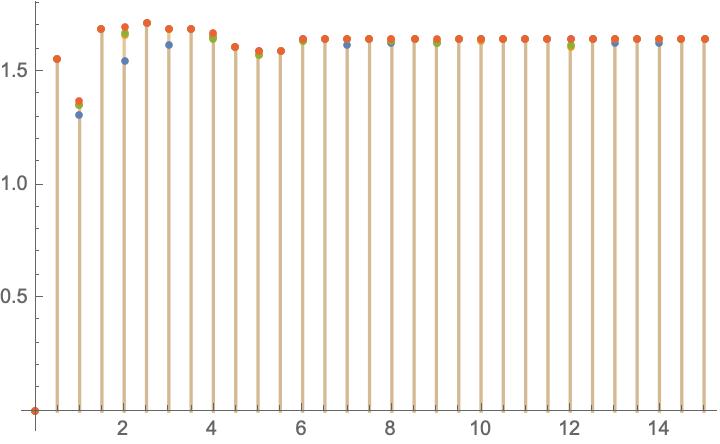}
\caption{Spins $j_1=3,j_2=4,j_3=5,k_1=6,k_2=8,k_3=7$.}
\end{subfigure}
\begin{subfigure}[t]{.45\linewidth} \includegraphics[width=1.0\textwidth]{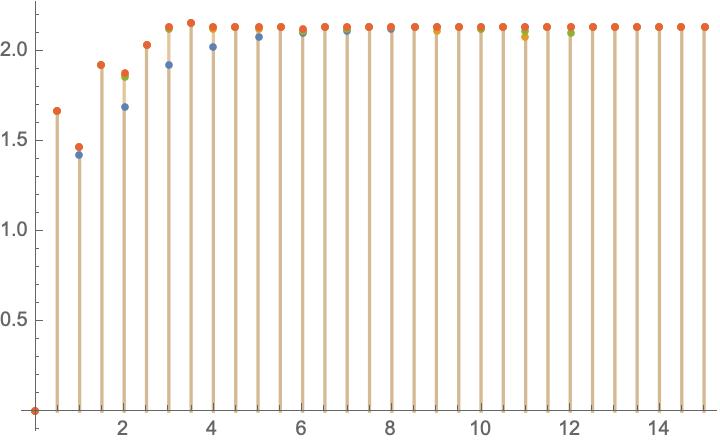}
\caption{Spins $j_1=5/2,j_2=7/2,j_3=4,k_1=6,k_2=8,k_3=11/2$.}
\end{subfigure}
\begin{subfigure}[t]{.45\linewidth} \includegraphics[width=1.0\textwidth]{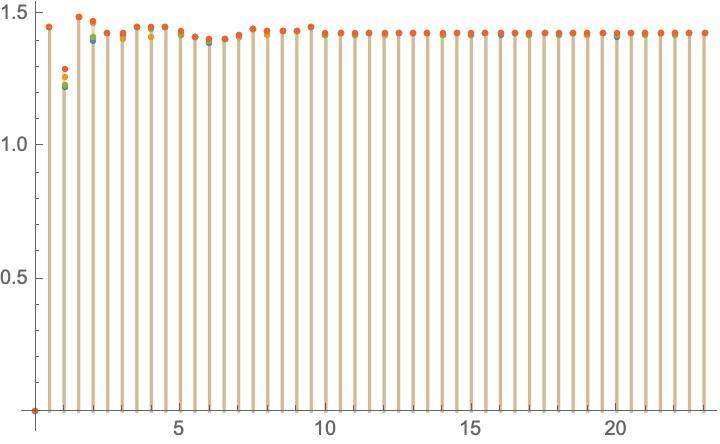}
\caption{Spins $j_1=6,j_2=9,j_3=7,k_1=12,k_2=11,k_3=10$.}
\end{subfigure}
\begin{subfigure}[t]{.45\linewidth} \includegraphics[width=1.0\textwidth]{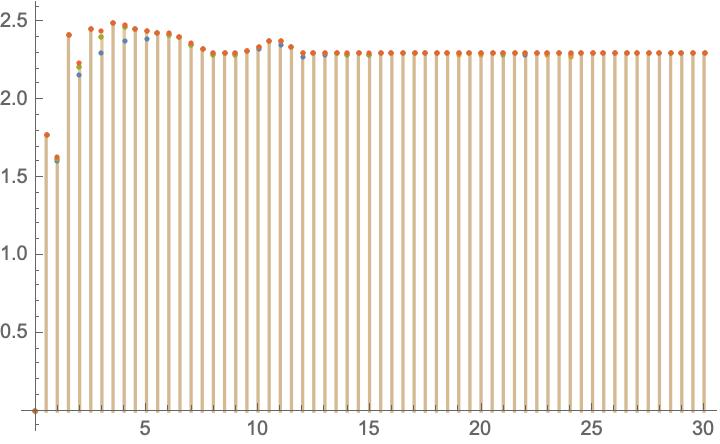}
\caption{Spins $j_1=9/2,j_2=11/2,j_3=7,k_1=12,k_2=16,k_3=27/2$.}
\end{subfigure}
\caption{The comparison between bipartite entanglement entropies and geometric measure of entanglement, where {\color{blue}{blue}} point is for
$\f12S_A$, {\color{orange}{orange}} point  for $\f12S_B$, {\color{green}{green}} point  for $\f12S_C$
and {\color{red}{red}} point for $S_g$.}
\label{fig:EE-GME-trianglegraph}
\end{figure}

One could also have computed the geometric entanglement by gauge-fixing and working out the closure defect's probability distribution. Since the bulk graph consists in a single loop made of three edges, there are three ways to gauge-fix: choosing $k_1$ as the ``loopy spin'', i.e. the edge which is not gauge-fixed and  gauge-fixing the edges associated with $k_{2}$ and $k_{3}$, plus the two other possibilities of choosing $k_{2}$ or $k_{3}$ as the ``loopy spin''.

For instance, contracting the edges carrying $k_{2}$ and $k_{3}$ and keeping $k_{1}$, we obtain the probability distribution of the closure defect  as
\be
p_{k_1}(J)
=
(2J+1) (2k_2+1)(2k_3+1)\sum_{L} (2L+1) \begin{Bmatrix} j_3 & j_2 & L \\ k_1 & k_1 & J \\ k_2 & k_3 & j_1 \end{Bmatrix}^2
\,,
\ee
with the normalization $\sum_{J} p_{k_1}(J)=1$.
%due to the identity of $9j$-symbol
%Similarly, other paths of gauge-fixing give the corresponding probability distribution, i.e., $p_{k_2}(J)$ and $p_{k_3}(J)$.
This probability distribution gives the 2nd-order time derivative of geometric  entanglement:
\beq
\f12\f{ \rd^2 S_{g} }{ \rd t^2} \Big\vert_{t=0}
&=&
\sum_{J,L}
(2J+1) (2L+1) (2k_1+1)(2k_2+1)(2k_3+1) \begin{Bmatrix} j_3 & j_2 & L \\ k_1 & k_1 & J \\ k_2 & k_3 & j_1 \end{Bmatrix}^2 \sum_{s=0}^{2\ell}
(-1)^{J+s+2k_1} \begin{Bmatrix} J & k_1 & k_1 \\ s & k_1 & k_1 \end{Bmatrix}
\label{eq:IntertwinerEE-trianglegraph-6j-ClosureDefect}
\\
&&
-\left(
\sum_{J,L}
(2J+1) (2L+1) (2k_1+1)(2k_2+1)(2k_3+1) \begin{Bmatrix} j_3 & j_2 & L \\ k_1 & k_1 & J \\ k_2 & k_3 & j_1 \end{Bmatrix}^2
(-1)^{J+\ell+2k_1} \begin{Bmatrix} J & k_1 & k_1 \\ s & k_1 & k_1 \end{Bmatrix}
\right)^2
\,.
\nn
\eeq
It is straightforward to check that these results does not actually depend on the choice of gauge-fixing and that it gives the same expression for the leading order entanglement computed above in eqn.\eqref{eq:IntertwinerEE-trianglegraph-6j}.
%2nd-order derivative identical with the result by previous method. Again, the result does not rely on the path of gauge-fixing.
%
This can be shown by using the Biedenharn-Elliott identity.
% (\ref{eq:B-E-identity}) to show eqn.(\ref{eq:IntertwinerEE-trianglegraph-6j-ClosureDefect}) equal to eqn.(\ref{eq:IntertwinerEE-trianglegraph-6j}).
More precisely, we need to prove the following relation:
\beq
&&
 \begin{Bmatrix}
   j_1 & k_2 & k_3 \\
   \ell &  k_3 & k_2
  \end{Bmatrix}\begin{Bmatrix}
   j_2 & k_3 & k_1 \\
   \ell &  k_1 & k_3
  \end{Bmatrix} \begin{Bmatrix}
   j_3 & k_1 & k_2 \\
   \ell & k_2 & k_1
  \end{Bmatrix}
\nn
\\
&=&
(-1)^{\sum_{i=1}^{3}( j_i+2k_i + \ell)}
\sum_{J,L}
(2J+1)(2L+1)
 \begin{Bmatrix}
   j_3 & j_2 & L \\
   k_1 &  k_1 & J \\
   k_2 & k_3 & j_1
  \end{Bmatrix}^2
(-1)^{J+\ell+2k_1}
 \begin{Bmatrix}
   J & k_1 & k_1 \\
   \ell &  k_1 & k_1
  \end{Bmatrix}
  \label{eq:6jsymbols-9jsymbol}
  \,.
\eeq
%As soon as it is proven, the terms in eqn.(\ref{eq:IntertwinerEE-trianglegraph-6j}) and (\ref{eq:IntertwinerEE-trianglegraph-6j-ClosureDefect}) that involve to $\ell$ and $s$, are equal respectively.
Starting from the left side, we use the  Biedenharn-Elliot identity twice:  we first recouple the first two $6j$-symbols via the Biedenharn-Elliot identity, which results in three $6j$-symbols, then we take one of them and recouple it to the third  $6j$-symbol of the original prooductone via the Biedenharn-Elliot identity again. This gives:
%Due to the symmetry of $6j$-symbols, we can firstly rewrite the $6j$-symbols into an easier form for sake of using Biedenharn-Elliot identity.
\beq
&&
\begin{Bmatrix}
   \ell & k_3 & k_3 \\
   j_1 &  k_2 & k_2
  \end{Bmatrix}\begin{Bmatrix}
   \ell & k_1 & k_1 \\
   j_2 &  k_3 & k_3
  \end{Bmatrix} \begin{Bmatrix}
   \ell & k_2 & k_2 \\
   j_3 & k_1 & k_1
  \end{Bmatrix}
\\
  &=&
  \sum_{L',J}
  (-1)^{j_1+j_2+j_3+2k_1+2k_2+2k_3+\ell}(-1)^{J+\ell+2k_1}
  \begin{Bmatrix}
   k_1 & k_1 & J \\
   k_1 & k_1 & \ell
  \end{Bmatrix}
  (2J+1)(2L'+1)
    \begin{Bmatrix}
   j_1 & j_2 & L' \\
   k_1 & k_2 & k_3
  \end{Bmatrix}^2
    \begin{Bmatrix}
   J & j_3  & L' \\
   k_2 & k_1 & k_1
  \end{Bmatrix}^2
  \,,
  \nn
\eeq
Finally, we deal with the two squared  $6j$-symbols  by the contraction formula of Racah coefficients (see \cite{brink1968angular}, page 143)
\be
(-1)^{2h+2d}
\sum_{c}(2c+1)
    \begin{Bmatrix}
   a & i & k \\
   f & b & c
  \end{Bmatrix}
    \begin{Bmatrix}
   a & b  & c \\
   d & e & f \\
   g & h & i
  \end{Bmatrix}
=
    \begin{Bmatrix}
   a & i & k \\
   h & d & g
  \end{Bmatrix}
    \begin{Bmatrix}
   b & f  & k \\
   d & h & e
  \end{Bmatrix}
  \,,
\ee
which allows to write
\beq
&&
\sum_{L'}(2L+1)
    \begin{Bmatrix}
   j_1 & j_2 & L' \\
   k_1 & k_2 & k_3
  \end{Bmatrix}^2
    \begin{Bmatrix}
   J & j_3  & L' \\
   k_2 & k_1 & k_1
  \end{Bmatrix}^2
  \nn
  \\
  &=&
  \sum_{L,L''}(2L+1)(2L''+1)
    \begin{Bmatrix}
   j_1 & j_2 & L' \\
   j_3 & J & L
  \end{Bmatrix}
      \begin{Bmatrix}
   j_1 & j_2 & L' \\
   j_3 & J & L''
  \end{Bmatrix}
    \begin{Bmatrix}
   j_1 & J  & L \\
   k_2 & k_1 & j_3 \\
   k_3 & k_1 & j_2
  \end{Bmatrix}
      \begin{Bmatrix}
   j_1 & J  & L'' \\
   k_2 & k_1 & j_3 \\
   k_3 & k_1 & j_2
  \end{Bmatrix}
=
    \begin{Bmatrix}
   j_1 & J  & L \\
   k_2 & k_1 & j_3 \\
   k_3 & k_1 & j_2
  \end{Bmatrix}^2
  \,,
\eeq
By symmetry of $9j$-symbol, this proves the wanted eqn.(\ref{eq:6jsymbols-9jsymbol}).

\medskip

Finally, we can look at the semi-classical regime at large spins $\{ j_i,k_i\}$, in which the spin network has a clear interpretation in terms of dual triangulation: the triangle graph is dual to a open tetrahedron to which one has removed one triangle, as drawn on fig.\ref{fig:trianglegraph-triangulation-spike}. This missing triangle corresponds to the the boundary edge of the graph, decorated with the spins $(j_{1},j_{2},j_{3})$.
For large spins, assuming that $\ell \ll \{ j_i,k_i\}$, one can apply Racah's approximation (\ref{eq:Angle-Asymptotic-6j-Racah}) for $6j$-symbols to the geometric entanglement formula \eqref{eq:IntertwinerEE-trianglegraph-6j} and write it in terms of Legendre polynomials in the cosine of the triangle angles, as illustrated on  figure \ref{fig:DualTriangulation-TriangleGraph}:
\be
S_{g} [ \Psi_{tri, \{j_i,k_i\}} (\ell,t) ]
\sim
%S_{g} [ \Psi^{\theta}_{tri, \{j_i,k_i\}} (\ell,t) ]
%=
t^2 \sum_{s=0}^{2\ell} P_s (\cos\theta_1) P_s (\cos\theta_2) P_s (\cos\theta_3) - t^2 [ P_s (\cos\theta_1) P_s (\cos\theta_2) P_s (\cos\theta_3) ]^2 + O(t^4)
\label{eq:IntertwinerEE-trianglegraph-N}
\ee
with the convention that  half-integer Legendre polynomials $P_{\ell}(x)$ vanish, for $\ell \in \N+\f12$.
%This approximation expression can be also obtained via putting three Legendre polynomials $P_{s}$ together as the exact expression (\ref{eq:IntertwinerEE-trianglegraph-6j}).
%
The angles are given in terms of the spins by
\be
\begin{aligned}
&\cos\theta_{1}=\f{ k_2(k_2+1) + k_3(k_{3}+1) - j_1(j_1+1)  }{ 2 \sqrt{ k_{2} (k_{2}+1) k_{3} (k_{3}+1) } }
\,,
\\
&\cos\theta_{2}=\f{ k_3(k_3+1) + k_1(k_{1}+1) - j_2(j_2+1)  }{ 2 \sqrt{ k_{3} (k_{3}+1) k_{1} (k_{1}+1) } }
\,,
\\
&\cos\theta_{3}=\f{ k_1(k_1+1) + k_2(k_{2}+1) - j_3(j_3+1)  }{ 2 \sqrt{ k_{1} (k_{1}+1) k_{2} (k_{2}+1) } }
\,.
\end{aligned}
\label{eq:Angles-TriangleGraph}
\ee
This expression allows us to study how the entanglement varies with respect to the spin network initial data.
\begin{figure}[htb!]
\begin{tikzpicture}[scale=0.7]
\coordinate (O) at (0,0);

\draw[thick] (O) ++(60:1.5) coordinate(E3) node[scale=0.7,blue] {$\bullet$} --++ (60:1.2) node[right] {$j_3$};
\draw[thick] (O) ++(180:1.5) coordinate(E1) node[scale=0.7,blue] {$\bullet$} --++ (180:1.2) node[above] {$j_1$};
\draw[thick] (O) ++(300:1.5) coordinate(E2) --++ (300:1.2) node[right] {$j_2$};

\draw[blue,thick] (E1) -- (E2) node[scale=0.7,blue] {$\bullet$} -- (E3) node[scale=0.7,blue] {$\bullet$} -- (E1) node[scale=0.7,blue] {$\bullet$};

\path (O) ++(0:3.5) coordinate (O1);
\path (O) ++(120:3.5) coordinate (O2);
\path (O) ++(240:3.5) coordinate (O3);

\draw[thick,brown] (O1) --(O2) -- (O3) --(O1);
\draw[thick,brown] (O) -- node[above,midway,blue] {$k_1$} (O1);
\draw[thick,brown] (O) -- node[right,midway,blue] {$k_2$} (O2);
\draw[thick,brown] (O) -- node[right,midway,blue] {$k_3$} (O3);

\draw[thick] (O) node[scale=0.7,brown] {$\bullet$};

\pic [draw, ->, "$\theta_3$", angle radius=0.3cm, angle eccentricity=1.75] {angle = O1--O--O2};
\pic [draw, ->, "$\theta_1$", angle radius=0.3cm, angle eccentricity=1.75] {angle = O2--O--O3};
\pic [draw, ->, "$\theta_2$", angle radius=0.3cm, angle eccentricity=1.75] {angle = O3--O--O1};

\end{tikzpicture}

\caption{The dual triangulation ({\color{brown}{brown}} lines) on triangle graph. The angles are given by (\ref{eq:Angles-TriangleGraph}).
}
\label{fig:DualTriangulation-TriangleGraph}
\end{figure}
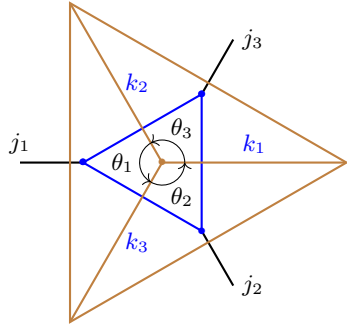

As for the bipartite system defined on the candy graph, the extremal configurations for the entanglement excitation by the holonomy operator are given by extremal geometries: spiky tetrahedral geometries with flatten triangles and maximal curvature.
More precisely, to start with, vanishing angles $\theta_1=\theta_2=\theta_3=0$  maximizes the multipartite entanglement excitation, as for the candy graph geometry. It corresponds to the asymptotic limit of spiky tetrahedron, as drawn on fig.\ref{fig:trianglegraph-triangulation-spike}, where the summit is sent to infinity, which is interpreted as an extremal bulk curvature.
Similarly, angular configurations $\theta_1=0,\theta_2=\theta_3=\pi$, and its two other permutations, also produce a maximal entanglement growth. In terms of spins and thus of edge lengths, this corresponds to $j_1 \ll k_2,k_3$, $j_2\approx k_1+k_3$, $j_3 \approx k_1+k_2$, which satisfy the triangular inequalities for $(j_1,j_2,j_3)$.

On the contrary, the angle configurations $\theta_1=\theta_2=0,\theta_3=\pi$ and its other two permutations, as well as $\theta_1=\theta_2=\theta_3=\pi$, all lead to a minimal entanglement, i.e. vanishing entanglement excitation at leading order in $t^{2}$. Translating in terms of spins and edge lengths, the former case corresponds to $j_1, j_2 \ll k_1,k_2,k_3$, $j_3 \approx k_1+k_2$, which doesn't satisfy the triangular inequalities for $(j_1,j_2,j_3)$, and thus is not an allowed spin network configuration. The latter case, with equal angles  $\theta_1=\theta_2=\theta_3=\pi$ , corresponds to $ j_1\approx k_2+k_3$, $j_2 \approx k_3+k_1$, $j_3 \approx k_1+k_2$, which is allowed by the triangular inequalities.
%but which do satisfy (one could also check the $6j$-symbols by letting $j_1=k_2+k_3$, $j_2=k_3+k_1$, $j_3=k_1+k_2$).
%The equilateral angles case $\theta_1=\theta_2=\theta_3=\pi$ now minimizes the entanglement growth rate, totally different from the case of bipartite entanglement.
%
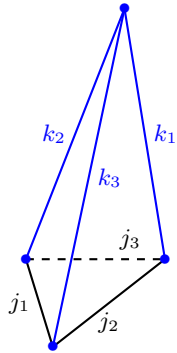
\begin{figure}[htb]
\begin{tikzpicture}[thick,scale=1.5]

\coordinate (A) at (0,2.061,0);
\coordinate (B) at (-0.25,-0.554,1);
\coordinate (C) at (-1.066,-0.354,-0.5);
\coordinate (D) at (0.166,-0.354,-0.5);

\draw[blue] (A) -- (B) node[midway,right] {$k_3$};
\draw[blue] (A) -- (C) node[midway,left] {$k_2$};
\draw[blue] (A) -- (D) node[midway,right] {$k_1$};
\draw[black] (B) -- (C) node[midway,left] {$j_1$};
\draw[dashed,black] (C) -- (D) node[near end,above] {$j_3$};
\draw[black] (D) -- (B)  node[midway,below] {$j_2$};

\draw[blue] (A) node{$\bullet$};
\draw[blue] (B) node{$\bullet$};
\draw[blue] (C) node{$\bullet$};
\draw[blue] (D) node{$\bullet$};
\end{tikzpicture}
%\vspace*{4mm}
\caption{A spike triangulation.
}
\label{fig:trianglegraph-triangulation-spike}
\end{figure}

\smallskip

%We briefly sum up this section: we have computed all possible bipartite entanglement entropies on triangle graph, via the manner of truncation of dynamics. It is remarkable that when $\ell \in \N +\f12$, the three bipartite entanglement are equal at least up to 2nd-order of $t^2$. In addition, as $\ell$ increases, the distinction between three bipartite entanglement trend to vanish, and eventually stable and equal (again in the sense up to order-$t^2$). We also compute the geometric measure of entanglement. It turns out that there is a simple rule to read it from multi-vertices spin network, as done in candy graph. In the cases of $\ell \in \N +\f12$, the geometric measure of entanglement equals to bipartite entanglement entropy in the sense of up to order-$t^2$. While $\ell \in \N$ they are different. But as the $\ell$ increases, all entanglement measures trend to be equal, and final stable and equal after beyond some critical value.
%On the other hand, we also show the extremal value points in the large spin limit. Compared with the case of candy graph, the triangle graph admits much more extremal configuration to have a maximal or minimal entanglement excitation. It is clear that the number will increase as the graph containing more vertices. Hence it is possible to look at the amounts of entanglement excitation to probe the complexity of graph: more extremal points, more complicated structure.

To conclude this section, we extended to the triangular graph our analysis of the entanglement excitation by the loop holonomy operator on a spin network basis state. This has confirmed our general analysis of the geometric entanglement for an arbitrary graph. More precisely,  we have computed analytically and numerically both the linear entanglement entropies for bipartitions of the graph and the geometric entanglement and checked that they match at leading order in the time $t$. Since the initial state is unentangled, the evolution starts with a vanishing entanglement then grows in $t^{2}$. Moreover, we have identified the spin network configurations with extremal entanglement excitations, which unsurprisingly correspond to discrete geometries with extremal curvature.
This is a positive step towards establishing a more thorough and precise dictionary between quantum entanglement, quantum geometry and its dynamics in the framework of loop quantum gravity.

%%%
\section{Conclusion \& Outlook}
%%%

The present paper is dedicated to the study of entanglement in loop quantum gravity. More precisely, we looked into the evolution of entanglement under the action of the holonomy operator. Intuitively, the holonomy operator acting on a loop of a spin network state will act at once on all the intertwiners living at the nodes on the loop and will entangle them. More precisely, we studied the unitary evolution generated by a holonomy operator on an initial state given by a spin network basis state. Such an initial state has fixed spins and intertwiners, it is thus a separable state carrying absolutely no entanglement. The evolution will naturally create entanglement between the vertices of the spin network and our goal was to compute the excitation of entanglement by the holonomy operator.
The holonomy around a loop representing a discretized measure of curvature, this leads to a relation between excitations of the geometry -``quanta of curvature''- and excitations of the information -multipartite entanglement- over spin network basis states. This fits in the larger program of interpreting and reconstructing the quantum geometry of space-time from quantum information concepts.

At the technical level, we introduced a notion of geometric entanglement to measure the multipartite entanglement carried by a spin network state, defined as the distance between that state and the set of separable states, identified (up to minor subtleties) to the set of spin network basis states. This is understood as witness of how much the intertwiners of a spin network, i.e. its quanta of volume, are entangled with each other.
Starting from an initial spin network basis state, thus with vanishing entanglement, we find that the first order time derivative at initial time always vanishes, so that the leading order behavior is a quadratic growth of the entanglement. We show that the second order derivative is simply given by the dispersion of the holonomy operator, which can in turn be computed from the probability distribution of the ``closure defect'' around the loop (i.e. the total spin recoupling the spins  of all the edges attached to the loop). Considering holonomy operators with arbitrary spin $\ell\in\N/2$, we find that this entanglement excitation grows with the spin $\ell$ and exhibits a plateau behavior: the entanglement excitation saturates and reaches its maximal value when the spin $\ell$ becomes larger than all the spins around the loop.

%%%%%%%%
%
%In present paper, we study the entanglement issue in LQG, from both bipartite and tripartite entanglement. And we find the notion of geometric measure of entanglement is very suitable when comes to the dynamics driven by loop holonomy operator.
%
%We investigate the evolution of geometric measure of entanglement up to 2nd-order derivative with respect to time, presenting the formulas for computing 1st- and 2nd-order derivative. In particular, we find the 1st-order derivative always vanishes if the state is unentangled at that time, meanwhile, the 2nd-order derivative exhibits itself in the form of dispersion.
%
%We focus on the case whose initial state is unentangled, and investigate the entanglement excitation on trivalent spin network. By considering general one-loop (in the sense of topology) spin network, we relate the 2nd-order derivative to the closure defect probability distribution. Moreover, from the analysis of the dispersion, we find a simple relation that reveals the critical (stable) value with respect to loop holonomy spin-$\ell$, which will provide a stable regime for the 2nd-order derivative of geometric measure of entanglement.

We illustrate this analysis through its application to the simplest spin networks with non-trivial bulk and boundary: the candy graph, consisting in two vertices linked by a single loop, and the triangle graph, consisting in three vertices linked by a single loop, both with arbitrary number of boundary edges poking out of the bulk vertices. These configurations allow to study explicitly the correlation and entanglement between the intertwiners living at the graph vertices. We compute the geometric entanglement and express it in terms of 6j-symbols from spin recoupling, which allows to study its low-spin and large-spin regimes. We further compute the bipartite entanglement created by the holonomy operator, defined as the entropy of the reduced density matrix, and show that it fits exactly at leading order with the notion of geometric entanglement  we introduced.

Since holonomy operators are the basic building blocks of the Hamiltonian dynamics of loop quantum gravity, this work gives a first hint of the effect of the dynamics on the quantum information carried by spin network states. We hope that further characterizing the various operators of loop quantum gravity through their effect on the correlation and entanglement between intertwiners will allow to reformulate the whole dynamics and notion of physical states entirely in quantum information terms.

%\smallskip
%
%Via working on concrete examples of candy graph and triangle graph, we show the geometric measure of entanglement is equal to entanglement entropy, which is obtained via truncated dynamics, at least up to 2nd-order of $t$. This conclusion validates our assertion made previously: in the dynamics of loop holonomy operator, we don't need to concern the local unitary which may affect the geometric measure of entanglement.
%
%\smallskip
%
%Furthermore, through the two examples, we investigate the possibility that connects the notion of geometry to the notion of entanglement. The conclusion is that, the extremal curved bulk geometry always maximizes the entanglement excitation, thus it implies that entanglement could be positive-related to the curvature.
%On the other hand, not all maximal entanglement excitation configuration is curved: they could even be flat. In fact, the number of extremal value points with respect to entanglement excitation, could reflect the complexity of graph. Hence it may provide a way to probe the graph complexity from the perspective of entanglement.

%%%%%%%%%%%%%%%%%
\section*{Acknowledgement}
Q.C. is financially supported by the China Scholarship Council.
%%%%%%%%%%%%%%%%%

%%%%%%%%
\begin{appendix}
\section{Examples of entanglement evolution}
%%%%%%

%%%
\subsection{Time-dependent Bell state}
%\subsection{Example 1: Time-dependent Bell state}
\label{eg:TDBell}
%%%
%\begin{example}[Time-dependent Bell state] 
%
Let $| \Psi(t) \ra=\cos t | 0 0 \ra+\sin t | 11 \ra$. The Hermitian operator in the case has matrix representation $\wh{H}=\begin{pmatrix} 0 & -\ri \\ \ri & 0 \end{pmatrix}$ with respect to basis $\{ | 00 \ra, |11\ra \}$.
This bipartite state initiates with the unentangled state $| 00 \ra$. Within $t \in [0, \pi / 4)$, the $| \Phi(t) \ra$ is chosen as $| 00 \ra$, while within $t \in ( \pi /4, 3\pi /4)$, the $| \Phi(t) \ra$ is chosen as $| 11 \ra$.
%\end{example}
At the instant $t=\pi/4$, there is ambiguity to define the unique $| \Phi(t) \ra$. Also, the $| \Phi(t) \ra$ is discontinuous at $t=\pi/4$.
However, the maximal projection $| \la \Phi(t) | \Psi(t) \ra |^2$ is still differentiable along $t \to s+$ and $t \to s-$ (the left and right derivative could be different).

%\smallskip

%%%
\subsection{Black hole evaporation toy model}
%\subsection{Example 2: Black hole evaporation toy model}
\label{eg:BHtoy}
%%%

We give another simple example whose maximal projection $| \la \Phi(t) | \Psi(t) \ra |^2$ is also differentiable.
%\begin{example}[Black hole evaporation toy model] 
Consider the Hamiltonian $\wh{H}=\ri ( a^{\dagger} b^{\dagger} - a b )$ acting on the unentangled initial state $| 0 \ra_A \otimes | 0 \ra_B$.
The evolution produces the entangled quantum state:
 \be
 | \Psi(t) \ra=\f{1}{\cosh t} \sum_{n=0}^{ \infty } \tanh^n t | n \ra_A \otimes | n \ra_B
 \,.
 \ee
The Schmidt eigenvalues are labelled by an interger  $n$ and read $\lambda_n=\tanh^{2n} t / \cosh^2 t$. The bipartite entanglement entropy as a  function of the time $t$ is given by:
\be
S(t)=\ln \cosh^2 t-\sinh^2 t \ln \tanh^2 t\,,\qquad\textrm{with}\quad
\lim_{t \to 0}\f{\rd S(t)}{ \rd t}=0\,,\quad
\lim_{t \to 0}\f{\rd^2 S(t)} {\rd t^2} \to \infty
\,.\ee
This entropy evolution function grows asymptotically linearly as $2t$.
However, the maximal Schmidt eigenvalue is always given by the 0-mode, for the number of quanta $n=0$,  since $|\tanh^2 t| < 1$, so that the separable projection is constant, $| \Phi(t) \ra=| 0 \ra_A \otimes | 0 \ra_B$ for all times $t$. Then the maximal Schmidt eigenvalue $\lambda_{\max}(t)=\lambda_{0}=1 / \cosh^2 t$ is smooth and gives the geometric measure of entanglement.
%\end{example}

%%%%%%

\end{appendix}

%\newpage

%%%%%%%%%%
% BIBLIOGRAPHY

\bibliographystyle{bib-style}
\bibliography{LQG}

\end{document}